\documentclass[10pt,journal,compsoc]{IEEEtran}
\usepackage{amsmath}
\usepackage{amsfonts}
\usepackage{amsthm}
\usepackage{amssymb}
\usepackage{extarrows}
\usepackage{color}

\usepackage{graphicx}
\usepackage{longtable}
\usepackage{rotating}
\usepackage{multirow}
\usepackage{booktabs} 
\usepackage{subfigure}
\usepackage{makecell}

\usepackage[linesnumbered,ruled,vlined]{algorithm2e}

\usepackage{mathrsfs}
\newtheorem{definition}{Definition}
\newtheorem{theorem}{Theorem}

\newtheorem{assumption}{Assumption}

\makeatletter
\newcommand{\thickhline}{%
    \noalign {\ifnum 0=`}\fi \hrule height 1pt
    \futurelet \reserved@a \@xhline
}
\newcolumntype{I}{!{\vrule width 1pt}}
\makeatother

\ifCLASSINFOpdf
\else
\fi
\begin{document}
%
\title{ Lightning-Fast and Privacy-Preserving Outsourced Computation  in the Cloud}
%
%
%

\author{Ximeng Liu~\IEEEmembership{Member,~IEEE,}
Robert H. Deng,~\IEEEmembership{Fellow,~IEEE,} Pengfei Wu,~\IEEEmembership{Student Member,~IEEE,} Yang Yang,~\IEEEmembership{Member,~IEEE}
                       
         \thanks{X. Liu,  Y. Yang are with  College of Mathematics and Computer Science, Fuzhou University, Fuzhou, China. E-mail: snbnix@gmail.com (X. Liu),  yang.yang.research@gmail.com (Y. Yang).}
         
         \thanks{X. Liu,  R.H. Deng, Y. Yang  are with the School of Information Systems, Singapore
Management University, Singapore. E-mail:   robertdeng@smu.edu.sg (R.H. Deng)}
         
       \thanks{P. Wu is with School of Software and Microelectronics, Peking University, Beijing, China. E-mail: wpf9808@pku.edu.cn.}

 }

%
%

\markboth{ }%
{ }
%





\IEEEcompsoctitleabstractindextext{
\begin{abstract}
In this paper, we propose a framework for lightning-fast   privacy-preserving   outsourced computation framework in the cloud, which we refer to as LightCom. Using LightCom, a user can securely achieve the outsource     data storage and fast secure data processing   in a single cloud server  different  from  the existing  multi-server outsourced computation  model.
Specifically, we first present a general secure computation framework for LightCom under the cloud server equipped with multiple Trusted Processing Units (TPUs) which face the side-channel attack. Under the LightCom, we design two specified fast processing toolkits which allow the user to achieve the commonly-used  secure integer computation and secure floating-point computation against the  side-channel information leakage of TPUs, respectively.
 Furthermore, our LightCom can also  guarantee access pattern protection during the data processing and achieve user private information  retrieve after the computation. We prove that the proposed LightCom can successfully achieve the goal of single cloud outsourced data processing to avoid the extra computation server and trusted computation server, and demonstrate the utility and the efficiency of LightCom using simulations.

\end{abstract}

}

\maketitle

%
\IEEEpeerreviewmaketitle

\section{Introduction}
\label{sec:introduction}

\IEEEPARstart{T}HE internet of things (IoT), embedded with electronics, Internet connectivity, and other forms of hardware (such as sensors), is a computing concept that describes the idea of everyday physical objects being connected to the internet and being able to identify themselves to other devices. With large numbers  of IoT devices, huge amount of data are generated for usage. 
 According to  IDC\footnote{http://www.vebuso.com/2018/02/idc-80-billion-connected-devices-2025-generating-180-trillion-gb-data-iot-opportunities/}, the connect IoT devices  will reach  80 billion  in 2025, and help to  generate 180 trillion gigabytes of new data that year.  A quarter of the data will create in real time, and  95\% is to come from  IoT real time data. With such large volume real-time data are generated, it is impossible  for  the resource-limited IoT devices to store and do the data analytics in time.  
Cloud computing, equipped     
 almost unlimited power of storage and computing, provides diversity of   services on    demand, such as,  storage, databases, networking, software, analytics, intelligence. With the help of cloud computing, 49 percent of data will be stored in public cloud environments by 2025 \footnote{ https://economictimes.indiatimes.com/tech/internet/global-data-to-increase-10x-by-2025-data-age-2025/articleshow/58004862. cms?from=mdr}. Thus, it is unsurprisingly that
 the huge volume  data generated by IoT devices are outsourced  to the cloud for long-term storage and achieve real-time online processing.




Despite the advantages provided by IoT-cloud  data outsourcing  architecture,  the individual IoT users are hesitated to the   system for data storage and processing  without any  protection method. 
In the Internet of Medical Things example \cite{dimitrov2016medical},  patients' wearable mHealth devices that always  equipped with the biometric measurements sensors (such as heart rate, perspiration levels, oxygen levels) to  record  the physical sign of the patient.  The hospital can use client’s PHI decision-making model to automatically check a patients’ health status.  If no   protection method is adopted, patients's physical sign can be capture by adversary.  Moreover, the hospital model can be got  by other third-party company to make profit.  Use the traditional  encryption technique  can protect the data from leakage, however,  the ciphertext lost the original meaning of the plaintext which cannot doing any computations.

  Protecting the data and achieve the secure outsource computation   simultaneously is an eye-catching field to solve the above problems. 
Currently, there are typically   two aspects of  techniques  to achieve secure outsourced computation: theoretical cryptography solution and system security  solution. 
  For the cryptography point of view, homomorphic encryption  \cite{naehrig2011can}  is considered as a super-excellent solution for the outsourced computation which allows the third-party to perform the computation on the encrypted data without reveal the content of the plaintext.  Fully homomorphic encryption \cite{van2010fully} can achieve arbitrary computation on the plaintext corresponding to the    complex operations on ciphertext. However, the computation overhead is still tremendous which is not fit for the piratical usage (e.g., it requires 29.5 s to run secure integer multiplication computation with a common PC \cite{liu2018privacyxxx}).  
  Semi-homomorphic encryption \cite{bendlin2011semi, farokhi2016secure} only supports one types of  homomorphic (e.g. additive homomorphic), can achieve  complex data computation on the encrypted data with the help of  extra honest-but-curious servers. But, the extra computation server will increase possibility of the information leakage. Recently, for the industrial community,  trusted execution environment  (TEE, such as Intel$^\circledR$ Software Guard Extensions (SGX)\footnote{https://software.intel.com/en-us/sgx} and ARM TrustZone\footnote{https://developer.arm.com/ip-products/security-ip/trustzone}) is developed to achieve the secure computation which 
   allows user-level or operating system code to define private regions of memory, also called enclaves.  The data in the enclave are protected and unable to be either read or saved by any process outside the enclave itself. The performance of the TEE is equivalent to the plaintext     computation overhead. Unfortunately, TEE     easily faces the side-channel attack, and the information inside the enclave can be leaked to the adversary\footnote{https://software.intel.com/en-us/articles/intel-sgx-and-side-channels} \footnote{https://www.arm.com/products/silicon-ip-security/side-channel-mitigation}.
Thus,
an fascinating problem appears for creating a system to balance the  usage of   practical outsourced  computation system and  eliminate the extra information leakage risk: \textit{how can a single cloud securely perform the  arbitrary outsourced computation without the help of extra third-party computation server or trusted authority, which interactions between the user and the cloud kept to a minimum.}


In this paper, we seek to address the above-mentioned challenge by presenting a framework for   lightning-fast  and privacy-preserving {o}utsourced {c}omputation Framework in a Cloud (LightCom). We regard the contributions of this paper to be six-fold, namely:
\begin{itemize}

\item  \textit{Secure Data Outsourced Storage.} The LightCom allows each user  to outsource his/her individual data to a cloud data center   for secure storage without compromising the privacy of his/her own   data to the other unauthorized storage.  

\item    \textit{Lightning-fast and Secure Data Processing in   Single Cloud.} The LightCom   can allow  in a single cloud  equipped  with multiple Trusted Processing Units (TPUs), which provides a TEE to achieve the user-centric  outsourced   computation on the user's encrypted data. 
Moreover, the data in  outside untrusted storage are secure against  chosen ciphertext attack for long-term, while data  insider TPUs can be  protected against side-channel attack.

\item    \textit{Outsourced  Computation Primitive Combinable.}   Currently,   the outsourced computation methods focus on a special computation task, such as outsourced exponential computation. Different specific outsourced tasks are constructed with different crypto preliminary. Thus, the previous computation result cannot be directly used for the input of the next computation.
Our LightCom can directly solve the problem with uniform design method which can achieve computation combinable.

\item   \textit{No Trusted Authority Involved.}
In most of the existing cryptosystem,  trusted authority is fully trusted which is an essential party in charge of distributing the public/private keys for all the other parties in the system. Our LightCom does not involve an additional fully trusted party in the system which makes the system more efficient and practical.



\item   \textit{Dynamic Key/Ciphertext Shares Update.} To reduce the user's private key and data leakage risk during the processing, our LightCom randomly splits the key and data into different shares which are processed in different TPUs, cooperatively. To avoid long-term shares leaking for recovering the original secrets, 
our LightCom allows TPUs updating  user's   ``old'' data/private-key shares  into the ``new'' shares on-the-fly dynamically 
 without  the participation of the data user.

\item   \textit{High User Experience.} 
Most existing privacy-preserving computation technique requires user to  preform different  pre-processing technique   according the function type  prior to  data outsourcing. 
The LightCom  does not need the data owner to perform any pre-processing procedure - only needs to encrypt and outsource the data to the cloud for storage. 
Thus, interactions between the user and the cloud     kept to a minimum - send the encrypted data to the cloud,  and received  outsourced  computed results in a single round.
\end{itemize}
\textbf{Motivation and Technique Overview.} As the sensitive information contained inside   TPU   can be attacked, our primary goal of the LightCom framework is to achieve secure computation in a single cloud without the help of an additional party.  The idea is to let the data store in the outside storage, and achieve privacy-preserving computation insider TPU.
  The main challenges are \textit{how to achieve both secure data storage and data processing against side-channel attacks, simultaneously}.  To solve the previous challenge, 
 we use a new  Paillier Cryptosystem Distributed Decryption (PCDD) which can achieve semantic secure data storage.  To prevent information leakage inside  TPU,   our LightCom uses \textit{one-time pad}   by adding some random numbers on plaintext of the  PCDD ciphertext.
Even the ``padded'' ciphertext for TPU enclave for decryption and process, the attacker still cannot get the original message of the plaintext.  To achieve ciphertext decryption, our LightCom uses multiple TPUs, and each TPU only stores a share of the private key to prevent the user's key leakage risk. 
Even some partial private key/data shares may leak to the adversary; our framework can successfully update these shares dynamically inside the TPU to make the leaked shares useless.    
More importantly, all the secure execution environment  (called TPU enclaves) in TPUs are dynamically building and release for the secure computation in our LightCom framework,  which can further decrease the information leak risk in the enclave.

\section{Preliminary}

\label{sec:preliminary}
\subsection{Notations}
 
Throughout the paper, we use $\|x\|$ to denote bit-length of $x$, while  ${\cal{L}}(x)$  denotes the number of element in $x$. Moreover, we use $pk_a$ and $sk_a$ to denote the public and private keys of a Request User (RU) $a$,  $sk_a^{(1)}, sk_a^{(2)}$ to denote the partial private keys that form $sk_a$, $[\![x]\!]_{pk_a}$ to denote the encrypted data of $x$ using $pk_a$ in public-key cryptosystem. For simplicity, if all ciphertexts belong to a specific RU, say $a$, we   simply use $[\![x]\!]$ instead of $[\![x]\!]_{pk_a}$. We use notion $\langle m  \rangle$ to denote the data share of $m$, i.e., each party  $i$ $(i = 1,\cdots, {\cal{P}})$ holds $m_i$, such that $\sum_{i=1}^{{\cal{P}}}m_i = m$.
 
  \subsection{Additive Secret Sharing Scheme (ASS)}
 
 Give $m\in \mathbb{G}$ ($\mathbb{G}$ is a finite abelian
group under addition),  the  
additive secret sharing scheme (a.k.a. $\cal{P}$-out-of-$\cal{P}$ secret sharing scheme) can be classified into   the following two algorithms -- Data Share Algorithm (\texttt{Share}) and Data Recovery Algorithm (\texttt{Rec}): 

$\texttt{Share}(m):$ Randomly generate $X_1,\cdots, X_{{\cal{P}}-1} \in \mathbb{G}$,  the algorithm  computes $X_{\cal{P}} = m -(X_1+\cdots +X_{{\cal{P}}-1})$, and outputs $X_1,\cdots, X_{{\cal{P}}}.$

 $\texttt{Rec}(X_1,\cdots,X_{\cal{P}}):$ With the shares  $X_1,\cdots,X_{\cal{P}}$, the algorithm can recover the message $m$ by computing  
 with $m = X_1+\cdots+X_{\cal{P}}$ under group $\mathbb{G}$.
 
\subsection{Additive Homomorphic Encryption Scheme}
\label{sec:AHES}

To reduce the communication cost of the LightCom, we used an Additive Homomorphic Encryption (AHE)  scheme as the basis.  Specifically, we use one of the AHE support threshold decryption called  Paillier Cryptosystem Distributed Decryption  (PCDD) in our previous work which contains six algorithms called 
 {Key Generation (\texttt{KeyGen})},  {Data Encryption (\texttt{Enc})},  {Data Decryption (\texttt{Dec})}, 
 {Private Key Splitting (\texttt{KeyS})},   
 {Partially decryption (\texttt{PDec})}, 
 {Threshold decryption  (\texttt{TDec}).}    
 The plaintext is belonged to $\mathbb{Z}_N$ and the ciphertext is belonged to $\mathbb{Z}_{N^2}$.
 The construction of the above algorithms can be found in Supplementary Materials Section C. 
Here, we introduce the two properties of the PCDD as follows:
1)  \textit{Additive Homomorphism}:  Given ciphertexts $[\![m_1]\!]$ and $[\![m_2]\!]$ under a same public key $pk$, the additive homomorphism can be computed  by ciphertext multiplication, i.e., compute
$[\![m_1]\!] \cdot [\![m_2]\!]  = [\![m_1+m_2]\!] $.  
2) \textit{Scalar-multiplicative Homomorphism}: Given ciphertext $[\![m]\!]$  and a constant number $c \in \mathbb{Z}_N$, it has 
$([\![m]\!])^{c}  
= [\![cm]\!].$

With the two properties given above,  we   show that our PCDD have the polynomial  homomorphism property, i.e.,  given $[\![x_1]\!],\cdots, [\![x_n]\!]$ and $a_1,\cdots,a_n$,  it has 
 \begin{center}
 $[\![a_1 \cdot x_1 + a_2 \cdot x_2 + \cdots a_n x_n ]\!]  \leftarrow [\![x_1]\!]^{a_1} \cdot [\![x_2]\!]^{a_2} \cdots [\![x_n]\!]^{a_n}$.
  \end{center}
 
   \subsection{Mathematical Function Computation}
In this section, we define the function which is used for data processing in our LightCom. 

 \begin{definition}[Deterministic Multiple-output  Multivariable Functions]
  Let $D = \{(x_1,\cdots, x_v) : x_i \in \mathbb{G} \}$ be a subset of $\mathbb{G}^v.$ We define the deterministic multiple-output  multivariable function as follows:
 (I) A multiple-output  multivariable function $\cal{F}$ of $v$ variables is a rule which assigns  each
ordered vector $(x_1,\cdots, x_v)$   in $D$ to a unique vector denoted $(y_1,\cdots, y_w)$, denote $ (y_1,\cdots, y_w) \leftarrow {\cal{F}}(x_1,\cdots,x_v)$. 
 (II) The set $D$ is called the domain of $\cal{F}$. 
 (III) The set ${\{\cal{F}}(x_1,\cdots,x_v)| (x_1,\cdots, x_v) \in D\}$ is called the range of $\cal{F}$.
  \end{definition}
  
     Note that the deterministic  multiple-output  multivariable  function is the general case of the deterministic   multiple-output  single-variable function   ($v=1$),  deterministic   single-output  multivariable  function ($w =1$), and deterministic single-output single variable function  ($v =1, w=1$). As all the functions used in our paper can  be successfully executed by a polynomial deterministic Turing machine (See Supplementary materials Section A), we omit the   word  ``deterministic'' in the rest of the paper.

  
  
 \begin{figure}[htbp]
\centering
\includegraphics[width=0.45\textwidth]{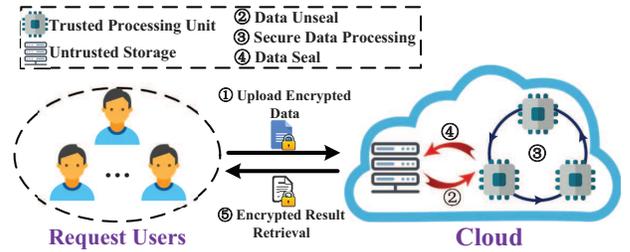}
\caption{System model under consideration}
\label{fig:systemmodel}
\end{figure}

\section{System Model \& Privacy Requirement}
\label{SEC:SPPP}
In this section, we formalize the LightCom system model,  and define the attack model.

\subsection{System Model}
\label{subsec:systemmodel}
In our LightCom system, we mainly focus on how the cloud server responds to a user request on outsourced computation   in a  privacy-preserving manner. The system comprises Request User (\textbf{RU}) and  a Cloud with Untrusted Storage  (\textbf{UnS}) and Trusted Processing Units (\textbf{TPUs})  - see   Fig. \ref{fig:systemmodel}.

\begin{itemize}

  \item  A \textbf{RU}   generates his/her public key, private key shares, and data shares. After that, the RU can securely outsource the public key and private/data shares to the cloud‘s  UnS for secure storage (See \textcircled{1}).  
   Moreover, the RU can also request a cloud to perform some secure outsourced computations on the outsourced data, and securely retrieve the final encrypted results (See  \textcircled{5}).
  \item  A \textbf{UnS} of the cloud has `unlimited' data storage space to store and manage data outsourced from the registered RU.  Also, the UnS also stores all the intermediate and final results for the RU in encrypted form.
  \item The \textbf{TPUs} of the cloud provides online computation ability for each  RUs.   Each TPU provides isolation secure computing environment for individual RU and can load   RU's data shares from UnS (See \textcircled{2}), perform certain calculations over the data shares (See \textcircled{3}), and then securely seal the data shares in UnS for storage (See \textcircled{4})\footnote{See the algorithm \texttt{Seal} and \texttt{UnSeal} in Section \ref{STPUDATA}.}.  Note that one TPU cannot load other TPU's sealed data which are stored in UnS.  
\end{itemize}

\subsection{Attack Model} \label{subsec:securitymodel}
 In our attack model,  the cloud is \textit{curious-but-honest} party, which strictly follow the protocol, but are also interested in learning  data belonged  to the RUs.  The UnS inside the cloud is transparency to both the cloud and the outsider passive attackers. Every  TPU can provide a secure execution environment (a.k.a., TPU enclave) for a RU which is secure against the other RU, the cloud and outsider passive attackers. 
 The inside non-challenge 
 RUs and outside attackers can also be interested to learn challenge RU's data.   Therefore, we introduce  three active adversaries ${\cal{A}}^*_1, {\cal{A}}^*_2, {\cal{A}}^*_3, $ which can  simulate the malicious actions corresponding to the outside attackers, non-challenge RUs, UnS, respectively. The goal of these adversaries is to get the challenge RU's plaintext or try to let the challenge RU get wrong computation result with the following capabilities:
  
1) ${\cal{A}}^*_1$ acts as the outside attacker that  may \textit{eavesdrop} on all communication links and CP's UnS, and  try to decrypt the challenge RU's  encrypted data. 
  2) ${\cal{A}}^*_2$ may \textit{compromise} RUs, with the exception of the challenge RU, to get access to their decryption capabilities, and try to guess all plaintexts belonging to the challenge RU.  
     3) ${\cal{A}}^*_3$ may \textit{compromise} the TPU to guess plaintext values of all data shares sent from the UnS by executing an interactive protocol.  
Noting that the above  adversaries ${\cal{A}}^*_1, {\cal{A}}^*_2, {\cal{A}}^*_3$ are restricted from compromising (i) all the TPUs concurrently\footnote{Note that ${\cal{P}} \geq 3$ TPUs are required in LightCom for the security consideration.}, and (ii) the challenge RU. 


 

\section{Basic Privacy Preserving   Computation Protocols}
\label{PCDDDDD}
In this section, we introduce our general design method of the mathematical function for LightCom. Moreover, the dynamic private/data share update without the participation of the DO are also introduced. 

 

 
 

   \subsection{The LightCom Design Method for the Single  Functions}
   \label{OCSGX:single}
  
   Our LightCom achieves the user data's privacy during the efficiency in the outsourced cloud with three-dimensional protection: 1) secure storing in the untrusted cloud storage; 2) secure processing in TPUs against side-channel attack; 3) efficient and dynamic outsourced key and data shares updating. Specifically, 
     to outsource the data to the cloud, the RU first initializes the system, uses the RU's public key to encrypt the data and outsource these encryptions along with the system parameters to UnS for storage. To achieve the second-dimensional protection,   our LightCom  uses the data sharing-based secure computation method between TPUs which can resist the side-channel attacks even the PPCD ciphertexts are decrypted. After finishing the   processing,   the data are sent back to UnS for further processing to finish the corresponding functionality defined in the program, and the enclaves in TPUs are released.  Moreover, to tackle the leaked private key and data shares, all the TPUs can jointly update these shares without the help of RU. Thus, the LightCom can classify into the following four phases.

 \textit{1) System Initialize Phase:}
 Firstly,  
  the RU generates a    public key $pk$ and  private key is $sk$ of appropriate public key crypto-system, and then   splits the private key $sk$ into ${\cal{P}}$ shares $sk_i $ $(i = 1,\cdots, {\cal{P}})$ with the \texttt{Share} algorithm. After that, 
  for each TPU $i$ in the cloud,    it initials an enclave $i$,  builds a secure channel, and 
   uploads the $sk_i$ to the     enclave $i$ securely. Finally, the TPU $i$ uses the data sealing to securely  stored the   $pk$, $sk_i$ in to UnS.

   \textit{2) Data Upload Phase:} 
  In the phase, the  RU  
    randomly separate  the data  
   $x_{j,1},\cdots,x_{j,{\cal{P}}}  \in \mathbb{G}$, such that $x_{j,1}+ \cdots +x_{j,{\cal{P}}} = x_j   $ for $j = 1,\cdots, v$. Then, the TPU $i$
$(i = 1,\cdots, {\cal{P}})$   creates the enclave $i$. After that, 
       the RU  defines  the program ${\cal{C}}_{i}$ for some specific computation function,  builds a secure channel with TPU enclave  $i$, 
        remotely loads $x_{1,i},\cdots,x_{v,i}, {\cal{C}}_{i}$ into the enclave $i$,  and securely seals      $ x_{1,i},\cdots,x_{v, i}, {\cal{C}}_{i} $ in the UnS.  After that,   TPU $i$ release    enclaves $i$ for all the $i = 1,\cdots, {\cal{P}}$.

    \textit{3) Secure  Computation Phase:}     The goal of the phase is to achieve the secure computation among the TPUs according to the user-defined program ${\cal{C}}_{i}$. Thus, it works as follows:
       \begin{itemize}
\item (3-\uppercase\expandafter{\romannumeral1}) Each TPU $i$ generates an enclave $i$. After that, all the TPUs build a secure channel with each other.     Load  sealed   data  $ x_{1,i},\cdots,x_{v,i}, pk, sk_i,  {\cal{C}}_{i}$  to enclave $i$ from UnS,  and denote them as $S_i$.
\item (3-\uppercase\expandafter{\romannumeral2})  TPUs jointly compute  
        $( y_{1,1},\cdots,y_{w,1}:\cdots: y_{1,{\cal{P}}},\cdots,y_{w,{\cal{P}}} ) \leftarrow \texttt{GenCpt} (S_1:\cdots:S_{\cal{P}})$
        according to the user-defined  program ${\cal{C}}_1,\cdots,{\cal{C}}_{\cal{P}}$. \footnote{The construction of  General Secure Function Computation Algorithm (GenCpt) can be found in section \ref{Sec:GenCpt}.}
        
  \item  (3-\uppercase\expandafter{\romannumeral3})      All the TPUs   jointly  update the private key shares and data shares dynamically. 
\end{itemize}
    
     After  the above computation, the TPU $i$ seals $y_{1,i},\cdots,y_{w,i}$ into the UnS, and releases the enclave.

     
    \textit{4) Data Retrive  Phase:}
If the RU needs to retrieve the computation results from the cloud, the  TPU $i$ creates an enclave $i$, opens the sealed data $y_{1,i},\cdots,y_{w,i} $,  builds a secure channel with the RU, and sends the data shares back to RU. 
   Once all the shares are sends to   RU, the RU computes $y_{j} = \sum_{i=1}^{\cal{P}} y_{j,i}$ for $j = 1,\cdots,w$.

     \subsection{The LightCom Design for   Combination of the Functions}
      \label{OCSGX:Double}
   Our LightCom can support for single data outsourced with  multiple function  operations. The procedure is as follows:
   
   \textit{1) System Initialize Phase:} Same to the LightCom with single function in Section \ref{OCSGX:single}.
   
     \textit{2) Data Upload Phase:}  
     After the system initialize phase, the  RU  
    defines  the program ${\cal{C}}_{i,t}$ for TPU $i$ $(i = 1,\cdots, {\cal{P}})$ with function computation step $t$ ($t = 1,\cdots, \zeta$) and  randomly separates  the data  
   $x_{j,1,1},\cdots,x_{j,1,{\cal{P}}}$, such that $x_{j,1,1}+ \cdots +x_{j,1,{\cal{P}}} = x_{j}$ for $j = 1,\cdots, v$ \footnote{Data share $x_{j,t,i} $ is for TPU enclave $i$ for data $j$ of function computation step-$t$.}.
After that, 
       the RU    builds a secure channel with TPU enclave  $i$, 
        remotely loads ${\cal{C}}_{1,i},\cdots,{\cal{C}}_{\zeta,i}$,  $x_{1,1,i},\cdots,x_{v,1,i}$ into the enclave $i$,  and securely seals     these data in the UnS.  After that,   TPU $i$ release    enclaves $i$ for all the $i = 1,\cdots, {\cal{P}}$.
      
            \textit{3) Secure  Computation Phase:}  The goal of the phase is to achieve the secure computation among the TPUs according to the user-defined program ${\cal{C}}_{t,i}$ for function $t$ ($t = 1,\cdots, \zeta$). Thus,  for each step    $t$, the   phase works as follows:
   
       \begin{itemize}
\item (3-\uppercase\expandafter{\romannumeral1}) Each TPU $i$ generates an enclave $i$. After that, all the TPUs build a secure channel with each other.     Load  sealed   data  $x_{1,t,i},\cdots,x_{v,t,i}, pk, sk_i, $ ${\cal{C}}_{1,i},\cdots,{\cal{C}}_{\zeta,i}$  to enclave $i$ from UnS,  and    put   them in a set ${\cal{E}}_{t,i}$.

\item (3-\uppercase\expandafter{\romannumeral2})  TPUs jointly compute  
        $( y_{1,t,1},\cdots,y_{w,t,1}:\cdots: y_{1,t,n},\cdots,y_{w,t,{\cal{P}}} ) \leftarrow \texttt{GenCpt} ({\cal{E}}_{t,i}:\cdots:{\cal{E}}_{t,i})$, 
        according to the user-defined  program ${\cal{C}}_{1,i},\cdots,{\cal{C}}_{\zeta,i}$.
        
  \item  (3-\uppercase\expandafter{\romannumeral3})      All the TPUs   jointly  update the private key and data shares. If $t =  \zeta$, the TPU $i$ seals $y_{1,\zeta,i},\cdots,y_{w,\zeta,i}$  into the UnS,   release the enclave. Otherwise, move to  (3-\uppercase\expandafter{\romannumeral4}) for further computation.

\item (3-\uppercase\expandafter{\romannumeral4}) Select  $x_{1,t+1,i},\cdots,x_{v,t+1,i}$ from the  $y_{1,t,i},\cdots,y_{w,t,i}$ for TPU $i$.  Then, the TPU $i$ seals $x_{1,t+1,i},\cdots,x_{v,t+1,i}$  into the UnS,   release the enclave, and move to  (3-\uppercase\expandafter{\romannumeral1}) for next step computation.    \end{itemize} 
      After the $t$ step is finished, the TPU $i$ seals the set ${\cal{E}}_j$   into the UnS, and releases the corresponding enclave.
         
             \textit{4) Data Retrieve  Phase: } 
                After the computation, TPU $i$ new an enclave $i$, opens the sealed  data $y_{1,\zeta,i},\cdots,y_{w,\zeta,i}$,  builds a secure channel with the RU, and sends these data back to  the RU. 
   Once all the TPU's data are sent, the RU computes the result   $y_{j,\zeta} = \sum_{i=1}^{\cal{P}} y_{j,\zeta,i}$ for step $\zeta$ ($j = 1,\cdots,w$) to get the final results.

        \subsection{General Secure Function Computation Algorithm  (GenCpt)}   
        \label{Sec:GenCpt}
      As the key component of the LightCom, the General Secure Function Computation Algorithm  (\texttt{GenCpt}) are proposed to achieve the  secure deterministic multiple-output multivariable function
     $\cal{F}$ computation which is introduced  in definition 1. Assume TPU $i$  $(i = 1,\cdots, {\cal{P}})$ holds $x_{1,i},\cdots, x_{v,i}$, \texttt{GenCpt} can securely output 
       $y_{1,i},\cdots, y_{w,i}$ for each TPU $i$, such that $(y_1,\cdots, y_w) \leftarrow {\cal{F}}(x_1,\cdots,x_v)$, where $x_{j,1}+\cdots + x_{j,{\cal{P}}} = x_j$ and $y_{k,1}+\cdots + y_{k, {\cal{P}}} = y_k$ for $j = 1,\cdots, v; k = 1,\cdots, w$. The \texttt{GenCpt} can be classified into offline/online stages and constructed as follows:
        
    \textbf{Offline Stage}:  Each TPU $i$ ($i  = 1,\cdots,  {\cal{P}}$) creates an
enclave $i$, loads the sealed   keys $pk, sk_i$ and program ${\cal{C}}_i$ into the enclave from the UnS, builds a secure channel with the other TPUs.\footnote{As  offline stage of   the secure computations needs to do TPU enclave initialization,  we just omit the description in the rest of the section.}   
    With the help of homomorphic public key cryptosystem, all the  TPUs    collaboratively 
    generate the shares of random numbers   and put them into a set ${\cal{R}}_i$. Note the     shares in set   ${\cal{R}}_i$ cannot be known by all the other TPUs during the generation. 
 After the above  computation, each TPU $i$ seals the ${\cal{R}}_i$ into the UnS, respectively.

     \textbf{Online Stage}\footnote{The input data $ x_{1,i},\cdots,x_{v,i},$,  public key $pk$,  private key shares $sk_i$,  and the program  ${\cal{C}}_{i}$ are loaded in the step of (3-I) of both section \ref{OCSGX:single}  and \ref{OCSGX:Double}}.  
   For each TPU $i$ ($i  = 1,\cdots, {\cal{P}}$), loads the sealed   random numbers set ${\cal{R}}_i$ from offline stage into the enclave $i$. 
 All the TPUs cooperatively compute and output the results 
 \begin{center}
 $(y_{1,i},\cdots, y_{w,i}) \leftarrow f_i(x_{1,i},\cdots,x_{v,i}, {\cal{R}}_i),$ \end{center} where $f_i$ is the  combination of $+, \times$ for $\mathbb{Z}_N$ and $\oplus, \land$ for $\mathbb{Z}_2$ with specific functionality according to the   program ${\cal{C}}_i$.


     \subsection{Private Key Share Update against Side-Channel Attack}
   \label{sec:PKSU}
 The private key shares are more sensitive and vulnerable, as the adversary can use the private key  to decrypt the RU's data in the untrusted storage if all shares of the private key are leaked by side channel attack.  Thus, we should frequently update the key shares in the TPU enclave. The intuitive idea is to let the  RU   choose a new private key, separate the new private key into different key shares,   update these key shares in the different individual enclaves, and update all the ciphertext with the new key.  However, the above strategy has the main drawback:  the RU has to be involved in the private/public key update phase which brings extra computation and communication cost. Thus, in this case,   the RU needs to generate and update the public/private keys frequently which is impractical. Thus, we bring the idea of proactive secret sharing into the LightCom: keeps the public/private key unchanged, the TPU will periodicity refresh the key shares without the participation of the RU. 
Mathematically,     to renew the shares at period $t$ $(t = 0, 1,2,\cdots)$,  we need to update the shares  such that $  \sum^{\cal{P}}_{i=1} sk_i^{(t+1)}=     \sum^{\cal{P}}_{i=1}  sk_i^{(t)} +  \sum^{\cal{P}}_{i=1}   \sum^{\cal{P}}_{j=1}   \delta^{(t)}_{i,j}  $, where $ \sum_{j=1}^{{\cal{P}}} \delta_{i,j} =0  $, $\sum_{i=1}^{{\cal{P}}}  sk^{(0)}_i = sk $ and $sk^{(0)}_i= sk_i$ for $i = 1,\cdots, {\cal{P}}$ (See Fig. \ref{fig:KSU} for example of private key update procedure with ${\cal{P}}=3$).  The special  construction is as follows:
 
 \begin{figure}[htbp]
\centering
\includegraphics[width=0.4\textwidth]{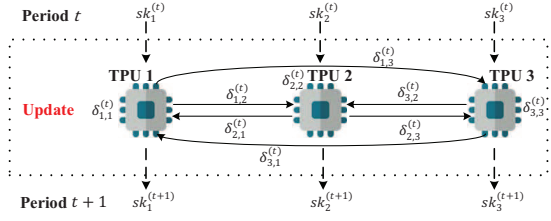}
\caption{Key Shares Update (example of ${\cal{P}} =3 $)}
\label{fig:KSU}
\end{figure}

1) Each TPU $i (i = 1,\cdots, {\cal{P}})$ creates an enclave $i$. After that,  TPU $i$ the  builds a secure channel with TPU $j$'s enclave  $ (j = 1,\cdots, {\cal{P}}; j \neq i)$.  

2) TPU $i$ picks random numbers $\delta_{i,1},\cdots,\delta_{i,{\cal{P}}} \in \mathbb{G}$ such that $\delta_{i,1}+\cdots+\delta_{i, {\cal{P}}} = 0$ under the group $\mathbb{G}$, and then  sends $\delta_{i,j}  $ to TPU  enclave $j$.

 3)  After received  $\delta_{j,i}$, TPU   $i$ 
   computes the new shares $sk^{(t+1)}_{i} \leftarrow sk^{(t)}_{i} + \delta^{(t)}_{1,i}+ \delta^{(t)}_{2,i} + \cdots+\delta^{(t)}_{{\cal{P}},i} \in \mathbb{G}$. After that, TPU $i$ 
 erases all the variables which  it used, except for its
current secret key $sk^{(t+1)}_{i} $.

     \subsection{Data Shares Update against Side-Channel Attack}
   
   As data shares need to load to TPU for processing, the shares can be leaked to the adversary  
    by side channel attack, and reconstruct the RU's original data.   Thus, we also need to dynamically update data shares $x^{(t)}_1,\cdots, x^{(t)}_{\cal{P}}$ at period $t$ $(t = 0,1,2,\cdots)$,   such that $  \sum^{\cal{P}}_{i=1} x_i^{(t+1)}=     \sum^{\cal{P}}_{i=1} x_i^{(t)} +  \sum^{\cal{P}}_{i=1}  \sum^{\cal{P}}_{j=1}  \delta_{i,j}$, where  $\sum^{\cal{P}}_{i=1} x^{(0)}_i = x$, $x^{(0)}_i= x_i$, and $ \sum^{\cal{P}}_{j=1} \delta_{i,j} =0$ for $i = 1,\cdots, {\cal{P}}$.  The construction is same to the private key share update method in section   \ref{sec:PKSU}.

   \section{TPU-based Basic Data Shares Operations}
   In this section, we introduce some basic TPU-based  data shares operations which can be used as the basis of LightCom.  
   
      \subsection{Data Domain and  Storage Format}
   
   Here, we introduce three  the data group domain for LightCom:  $\mathbb{Z}_N = \{0, 1, \cdots,N-1\}$,      $\mathbb{D}_N = \{-\lfloor \frac{N}{2}\rfloor, \cdots,0,\cdots, \lfloor \frac{N}{2}\rfloor)$, and  $\mathbb{Z}_2 = \{0,1\}$. As we use PCDD for offline processing and its plaintext domain is $\mathbb{Z}_N$, we define the   operation $\lceil x \rfloor_N$ which   transforms data $x $ from group $\mathbb{Z}_N$ into the   group $\mathbb{D}_N$, i.e., 
  \begin{equation}
\lceil x \rfloor_N \leftarrow \left\{
\begin{array}{lr}
  x,    & 0 \leq x<N/2 \\
x-N,    &N/2 \leq x<N. \\
\end{array}
\right. \notag
\end{equation}
  Moreover, the data $\lceil x \rfloor_N$ in group $\mathbb{D}_N$ can be directly transformed into group $\mathbb{Z}_N$ with $x = \lceil x \rfloor_N \mod N$. It can be easily verified that group $\mathbb{D}_N$ and  $\mathbb{Z}_N$ are \textit{isomorphism}. 
  
  To guarantee the security of secret sharing, two types of data shares are used in the LightCom, called integer share (belonged to $\mathbb{Z}_N$) and binary share (belonged to $\mathbb{Z}_2$). For the integer share separation, RU only needs to execute $\texttt{Share}(m)$, such that $m = m_1+\cdots+m_{ {\cal{P}}}$, where $m, m_1,\cdots,m_{ {\cal{P}}} \in \mathbb{D}_N$. For the binary shares, RU    executes $\texttt{Share}(\mathfrak{m})$, such that $\mathfrak{m} = \mathfrak{m}_1+\cdots+ \mathfrak{m}_{ {\cal{P}}} $, where $\mathfrak{m}, \mathfrak{m}_1,\cdots, \mathfrak{m}_{ {\cal{P}}} \in \mathbb{Z}_2.$  After that, RU securely sends integer share $m_i$ or binary shares $\mathfrak{m}_i$ to TPU $i$, and seals to UnS for securely storage.

      \subsection{System Initial and Key Distribution}

The LightCom system should be initialized before  achieving the secure computation. Firstly,  
  the RU executes \texttt{KeyGen} algorithm,  output  public key $pk=(N,g)$ and  private key is $sk=\theta$. Then, use \texttt{KeyS} to split key $\theta$ into $\cal{P}$ shares $sk_i = \theta_i$ $(i = 1,\cdots, { {\cal{P}}})$. After that, 
  for each TPU $i$ in the cloud,    it initials an enclave $i$,  builds a secure channel, and 
   uploads the $sk_i$ to the     enclave $i$ securely. Beside, the RU's  PCDD public key  $pk$ and   program ${\cal{C}}_{i}$ for the specific function $\cal{F}$ are needed to securely send to TPU $i$ $(i = 1,\cdots, { {\cal{P}}})$.
   Finally, the TPU $i$ securely seals the data     $pk$, $sk_i$, ${\cal{C}}_{i}$ into UnS.  
   As all the parameters need to load   to the TPU enclaves along with the data shares according the specific functionality, we will not specially  describe it in the rest of the section.



      \subsection{Secure Distributed Decryption Algorithm (\texttt{SDD})}
    Before constructing the  TPU-based operation, we need first to construct the algorithm called  Secure Distributed Decryption   (\texttt{SDD}) which allows all the TPUs decrypt PCDD's ciphertext. Mathematically, 
    if enclave in TPU $\chi$ contains the encryption $[\![x]\!]$, the goal of \texttt{SDD} is to output $x$ which contains following steps:  
   1) The TPU enclave $\chi$  establishes a secure channel with the other TPU
enclave $i ( i \neq \chi)$.  Then, enclave $\chi$
  sends  $ [\![x]\!]  $ to all the other enclave $i$. 
 2)   Once received $ [\![x]\!] $, the TPU $i$ uses
     \texttt{PDec} algorithm to get $CT_i$  and securely  send  $CT_i$ to enclave $\chi$. 3) Finally, the  TPU $\chi$  securely     uses    $CT_\chi$ 
      with \texttt{TDec} algorithm to get $x$.

      \subsection{Secure TPU-based Data Seal \& UnSeal}
      \label{STPUDATA}
    As TPU enclaves are only provide an isolated  computing   environment during the secure processing, the data in the TPU enclave needs to seal to UnS for long-term storage. Thus, we propose two algorithms called \texttt{Seal} and \texttt{UnSeal} to achieve. 
 
 $\texttt{Seal}(x_i):$ The   TPU $i$ encrypts the data share $x$ into $[\![x_i]\!]$, then uses hash function $H: \{0,1\}^* \rightarrow \mathbb{Z}_N$ with input   the $[\![x_i]\!]$ associated  with TPU $t$-time period private key share $sk^{(t)}_i$ to compute 
 $S_{t,i} \leftarrow H([\![x_i]\!] ||sk^{(t)}_i||ID_i||t)$, where $ID_i$ is the transaction identity  for $[\![x_i]\!]$. Then, TPU $i$ sends $[\![x_i]\!]$ with $S_{t,i}$ to UnS for storage. 
    
     $\texttt{UnSeal}([\![x]\!], S_{t,i}):$ The   TPU $i$ loads $[\![x_i]\!]$ with $S_{t,i}$ to the enclave $i$, and computes $H([\![x]\!] ||sk^{(t)}_i||ID_i||t)$ to test whether the result is equal to $S_{t,i}$. If the equation does not  holds, the algorithm stops and outputs $\perp$. Otherwise, the TPU $i$ uses \texttt{SDD} to get the share  $x_i$.

  \subsection{Random  Shares Generation}
The secret sharing based privacy computation requires one-time random numbers for processing. Before constructing the TPU-based computation, we design a protocol called       Random Tuple  Generation Protocol (\texttt{RTG}).     
      The goal of \texttt{RTG} is to let TPUs  cooperatively generate random tuple $ \mathfrak{r}_i^{(1)},\cdots, \mathfrak{r}_i^{(\ell)} \in \mathbb{Z}_2$ and $r_i \in \mathbb{D}_N$  for each TPU    $i$ $(i = 1,\cdots, {\cal{P}})$, such that $r = -\mathfrak{r}^{(\ell)}2^{\ell-1} + \sum_{j =1}^{\ell-1} \mathfrak{r}^{(j)} 2^{j-1}$  and $\mathfrak{r}^{(j)} =  \mathfrak{r}_1^{(j)} \oplus \cdots \oplus \mathfrak{r}_{\cal{P}}^{(j)}$ and $r = r_1 +\cdots+ r_{\cal{P}}$ holds, where $\ell$ is the bit-length of random number $r \in \mathbb{D}_N$. The \texttt{RTG} generates as follows:
     
1)  The TPU 1 randomly generates $\mathfrak{r}_1^{(1)},\cdots, \mathfrak{r}_1^{(\ell)} \in \mathbb{Z}_2$,   encrypts them as $[\![\mathfrak{r}^{(1)}_1]\!],\cdots,[\![\mathfrak{r}^{(\ell)}_1]\!]$,    denotes them as  $[\![\mathfrak{r}^{(1)}]\!],\cdots,[\![\mathfrak{r}^{(\ell)}]\!]$, and sends these ciphertexts to TPU $2$. 
 
2) The TPU $i$  ($i=2,\cdots, {\cal{P}}$) generates $\mathfrak{r}_i^{(1)},\cdots, \mathfrak{r}_i^{(\ell)} \in \mathbb{Z}_2$  and computes 
\begin{center}
$[\![\mathfrak{r}^{(j)}]\!] \leftarrow [\![\mathfrak{r}^{(j)}]\!]^ {(1- \mathfrak{r}_i^{(j)})} \cdot  ([\![1]\!] \cdot [\![\mathfrak{r}^{(j)}]\!]^{N-1})^{\mathfrak{r}_i^{(j)}} =  [\![\mathfrak{r}^{(j)} \oplus \mathfrak{r}_i^{(j)}]\!]$. 
\end{center}
    If $i \neq {\cal{P}}$, the TPU $i$
     sends $[\![\mathfrak{r}^{(1)}]\!],\cdots,[\![\mathfrak{r}^{(\ell)}]\!]$ to  TPU $i+1$. 
     If $i = {\cal{P}}$, the TPU $\cal{P}$ computes 
     \begin{center}
   $ [\![r]\!] \leftarrow  [\![\mathfrak{r}^{(\ell)}]\!]^{N-2^{\ell-1}} \cdot  [\![\mathfrak{r}^{(\ell-1)}]\!]^{2^{\ell-2}} \cdot \cdots \cdot  [\![\mathfrak{r}^{(1)}]\!].$
          \end{center}
3)   
    For   TPU $i$ $(i =  {\cal{P}},\cdots, 2)$,   randomly generates      $r_i \in \mathbb{D}_N$  and computes  $ [\![r]\!]  \leftarrow  [\![r]\!]  \cdot [\![-r_i]\!],$ and sends $ [\![r]\!]$ to TPU $i-1$. 
    Once TPU $1$ gets  $[\![r]\!]$, uses \texttt{SDD} to get $r$, and denotes $ \lceil r \rfloor_N $ as $r_1$. After computation,  each TPU $i$ $(i = 1,\cdots,  {\cal{P}})$ holds randomly bits $\mathfrak{r}_i^{(1)},\cdots, \mathfrak{r}_i^{(\ell)}\in \mathbb{Z}_2$ and integer  $r_i \in \mathbb{D}_N$.

   \subsection{Share Domain Transformation}

  \subsubsection{Binary  Share to Integer Share Transformation  (\texttt{B2I}) }
 
    Suppose TPU $i$ hold  a bit share $\mathfrak{a}_i \in \mathbb{Z}_2$, where $\mathfrak{a}_1  \oplus \cdots \oplus \mathfrak{a}_{\cal{P}} =  \mathfrak{s} \in   \mathbb{Z}_2$, the   goal of the protocol is to generate  a random integer share $  {b}_i  \in \mathbb{Z}_N$ for each TPU $i$, such that $  {b}_1  + \cdots +  {b}_{\cal{P}} =  \mathfrak{s} $.  To execute \texttt{B2I},  the TPU 1 randomly generates $ {b}_1  \in \mathbb{Z}_N$, denotes  ${x} =   {b}_1$ and
   $ \mathfrak{s} =   \mathfrak{a}_1 $,   encrypts $ {x}$  as $[\![  {x}]\!] $,   $\mathfrak{s}$ as $[\![ \mathfrak{s}]\!] $, and sends $[\![ {x}]\!]$  and  $[\![\mathfrak{s}]\!] $ to TPU $2$.
   After that, the TPU $i$ $(i=2,\cdots, {\cal{P}}-1)$
    generates $ {b}_i  \in \mathbb{Z}_N$  and computes 
    \begin{center}
$[\![\mathfrak{s}]\!] \leftarrow [\![\mathfrak{s}]\!]^ {(1- \mathfrak{a}_i)} \cdot  ([\![1]\!] \cdot [\![\mathfrak{s} ]\!]^{N-1})^{\mathfrak{a}_i } =  [\![\mathfrak{s}\oplus \mathfrak{a}_i]\!]$, 
$[\![x]\!] \leftarrow [\![x]\!] \cdot [\![b_i]\!] $,
\end{center} 
   and sends  $[\![ x]\!] , [\![ \mathfrak{s}]\!] $  to TPU $i+1$.
Once received the $[\![ x]\!], [\![ \mathfrak{s}]\!]$, TPU ${\cal{P}}$ computes 
\begin{center}
$[\![\mathfrak{s}]\!] \leftarrow [\![\mathfrak{s}]\!]^ {(1- \mathfrak{a}_{\cal{P}})} \cdot  ([\![1]\!] \cdot [\![\mathfrak{s} ]\!]^{N-1})^{\mathfrak{a}_{\cal{P}}} =  [\![\mathfrak{s}\oplus \mathfrak{a}_{\cal{P}}]\!]$, 
$[\![ b_{\cal{P}}]\!] \leftarrow [\![\mathfrak{s}]\!]  \cdot [\![x]\!]^{N-1}$, 
\end{center}
and uses the \texttt{SDD} to decrypt and gets $ {b_{\cal{P}}}$.

  \subsubsection{Integer Share to Binary  Share Transformation  (\texttt{I2B}) }
 
    Suppose TPU $i$ hold  an integer share $ {a}_i \in \mathbb{Z}_N$, where $ {a}_1  + \cdots +  {a}_{\cal{P}} =  \mathfrak{s} \in   \mathbb{Z}_2$, the   goal of the \texttt{I2B} protocol is to generate  a random bit share $\mathfrak{b}_i  \in \mathbb{Z}_2$ for each TPU $i$, such that $  \mathfrak{b}_1  \oplus \cdots \oplus  \mathfrak{b}_{\cal{P}} =    \mathfrak{s}$.  To execute  \texttt{I2B},  the TPU 1 lets $  {y} =    a_1 $,  encrypts $ {y}$ as $[\![  {y}]\!] $, and sends the ciphertext to TPU 2 for computation.  After that, the TPU $i$ $(i=2,\cdots, {\cal{P}})$
    uses the share to compute $[\![y]\!] \leftarrow [\![y]\!] \cdot [\![a_i]\!] $. If  $i \neq {\cal{P}}$, TPU $i $ sends $[\![y]\!]$ to TPU $i+1$.      After that, denote $[\![\mathfrak{s} ]\!]\leftarrow [\![ y ]\!]$, and each TPU $i$ $(i={\cal{P}},\cdots, 2)$  generates $\mathfrak{b}_i  \in \mathbb{Z}_2$  and computes 
      \begin{center} 
$[\![\mathfrak{s}]\!] \leftarrow [\![\mathfrak{s}]\!]^ {(1- \mathfrak{b}_i)} \cdot  ([\![1]\!] \cdot [\![\mathfrak{s}]\!]^{N-1})^{\mathfrak{b}_i } =  [\![\mathfrak{s} \oplus \mathfrak{b}_i]\!]$, 
\end{center} 
    and sends $[\![\mathfrak{s}]\!]$ to TPU $i-1$. Once received $[\![\mathfrak{s}]\!]$, TPU 1  uses the \texttt{SDD} to decrypt  $[\![\mathfrak{s}]\!]$ and denotes the result $\mathfrak{s}$ as $ \mathfrak{b}_1$.

   \section{TPU-based Secure Outsourced  Computing Toolkits in the Cloud}
   \label{TPU:SOCT}
   In this section, we introduce and construct the commonly used  secure outsourced  binary and  integer   computation  sub-protocols for  a single cloud. 
 
      \subsection{Secure Computation over Binary Shares}

       \subsubsection{Secure Bit Multiplication   Protocol (\texttt{SBM})}

The   {\texttt{SBM}}    can achieve plaintext multiplication on bit shares and output are bit shares, i.e.,  given two shares $  \mathfrak{x}_i  ,    \mathfrak{y}_i   \in \mathbb{Z}_2$ $(i = 1,\cdots,{\cal{P}})$ for  TPU $i$ as input,  {\texttt{SBM}}   securely outputs $   \mathfrak{f}_i   \in \mathbb{Z}_2$ for TPU $i$, such that $\bigoplus_{i=1}^{\cal{P}}    \mathfrak{f}_i   = (\bigoplus_{i=1}^{\cal{P}}   \mathfrak{x}_i)   \land  (\bigoplus_{i=1}^{\cal{P}}   \mathfrak{y}_i ).$

\textbf{Offline Stage}:   All the TPUs  initialize  their enclaves and load the public parameters to UnS. For
 enclave  1,  generate $\mathfrak{a}_1, \mathfrak{b}_1 \in \mathbb{Z}_2$,  compute $\mathfrak{c} = \mathfrak{a}_1 \cdot \mathfrak{b}_1 \in \mathbb{Z}_2$. 
 Encrypt $[\![\mathfrak{a}_1]\!],[\![\mathfrak{b}_1]\!] $ and $[\![\mathfrak{c}]\!]$, and denote them as   $[\![\mathfrak{a}]\!],[\![\mathfrak{b}]\!], [\![\mathfrak{c}]\!]$, respectively. After that, TPU enclave $i$ $(i = 1,\cdots, {\cal{P}}-1)$ sends $[\![\mathfrak{a}]\!], [\![\mathfrak{b}]\!], [\![\mathfrak{c}]\!]$ to enclave $i+1$, TPU $i+1$ generates $\mathfrak{a}_{i+1}, \mathfrak{b}_{i+1}$ and compute 
 \begin{center}
 $[\![A]\!] \leftarrow [\![\mathfrak{a}]\!] \cdot [\![\mathfrak{c}]\!]^{N-1} = [\![\mathfrak{a}  \cdot (1-\mathfrak{b})]\!],$
    \end{center}
     \begin{center}
  $[\![B]\!] \leftarrow [\![\mathfrak{b}]\!] \cdot [\![\mathfrak{c}]\!]^{N-1} = [\![\mathfrak{b} \cdot (1-\mathfrak{a})]\!],$
      \end{center}
     \begin{center}
      $[\![C]\!] \leftarrow [\![1]\!] \cdot [\![\mathfrak{a}]\!]^{N-1}  \cdot [\![\mathfrak{b}]\!]^{N-1}  \cdot [\![\mathfrak{c}]\!] = [\![(1-\mathfrak{b})\cdot (1-\mathfrak{a})]\!],$
          \end{center}
     \begin{align}
      [\![\mathfrak{c}]\!] & \leftarrow [\![\mathfrak{c}]\!]^{(1-\mathfrak{a}_{i+1}) (1-\mathfrak{b}_{i+1}) } \cdot [\![C]\!]^{\mathfrak{a}_{i+1}\cdot \mathfrak{b}_{i+1}} \cdot  [\![A]\!]^{(1-\mathfrak{a}_{i+1}) \mathfrak{b}_{i+1}}  \notag  \\   &   \cdot  [\![B]\!]^{\mathfrak{a}_{i+1}(1-\mathfrak{b}_{i+1})}=  [\![(\mathfrak{a} \oplus \mathfrak{a}_{i+1}) \land (\mathfrak{b} \oplus \mathfrak{b}_{i+1}) ]\!]. \notag
              \end{align}
     \begin{center}
     $[\![\mathfrak{a}]\!] \leftarrow [\![\mathfrak{a}]\!]^ {(1- \mathfrak{a}_{i+1})} \cdot  ([\![1]\!] \cdot [\![\mathfrak{a}]\!]^{N-1})^{\mathfrak{a}_{i+1}} = [\![\mathfrak{a} \oplus \mathfrak{a}_{i+1}]\!], $     \end{center}
          \begin{center}
  $[\![\mathfrak{b}]\!] \leftarrow [\![\mathfrak{b}]\!]^ {(1- \mathfrak{b}_{i+1})} \cdot  ([\![1]\!] \cdot [\![\mathfrak{b}]\!]^{N-1})^{\mathfrak{b}_{i+1}} =  [\![\mathfrak{b} \oplus \mathfrak{b}_{i+1}]\!].$
   \end{center}

  After the above computation,       enclave $i $  $(i = {\cal{P}},\cdots, 2)$ randomly generates $\mathfrak{c}_{i} \in \mathbb{Z}_N$ and computes $[\![\mathfrak{c}]\!] \leftarrow [\![\mathfrak{c}]\!]^ {(1- \mathfrak{c}_{i})} \cdot  ([\![1]\!] \cdot [\![\mathfrak{c}]\!]^{N-1})^{\mathfrak{c}_{i}} = [\![\mathfrak{c} \oplus \mathfrak{c}_{i}]\!]$. 
When the TPU $2$ sends   $[\![\mathfrak{c}]\!]$ to TPU 1, 
 the TPU $1$ uses \texttt{SDD} to get $\mathfrak{c}$ and denotes as $\mathfrak{c}_1$. 
 After the above computation, each enclave holds $\mathfrak{a}_i, \mathfrak{b}_i, \mathfrak{c}_i$,  which satisfies  $\mathfrak{a}_1\oplus \cdots \oplus \mathfrak{a}_{\cal{P}} = \mathfrak{a}$,  $\mathfrak{b}_1 \oplus \cdots \oplus \mathfrak{b}_{\cal{P}} = \mathfrak{b}$, $\mathfrak{c}_1 \oplus \cdots \oplus \mathfrak{c}_{\cal{P}} = \mathfrak{c}$ and $\mathfrak{c} = \mathfrak{a} \land \mathfrak{b}$. Finally, each TPU $i$ seals $\mathfrak{a}_i, \mathfrak{b}_i,\mathfrak{c}_i$ to UnS for storage individually.


    
    \textbf{Online Stage:}  For each TPU   $i$ $(i = 1,\cdots, {\cal{P}})$, load the   $\mathfrak{a}_i, \mathfrak{b}_i, \mathfrak{c}_i$ into the enclave $i$. Then, 
compute $X_i =  \mathfrak{x}_i \oplus \mathfrak{a}_i  $ and $Y_i =  \mathfrak{y}_i  \oplus \mathfrak{b}_i$. Securely send $X_i$ and $Y_i$ to other enclave $j$ $(j = 1,\cdots,{\cal{P}}; j \neq i)$.
After receiving other $X_j$ and $Y_j$,  each TPUs  computes  $X = \bigoplus_{i=1}^{{\cal{P}}}X_j$ and    $Y = \bigoplus_{i=1}^{{\cal{P}}} Y_j$.
For TPU $i$ $(i = 1,\cdots, {\cal{P}}-1)$, compute $\mathfrak{f}_i\leftarrow \mathfrak{c}_i \oplus (\mathfrak{b}_i \land X) \oplus (\mathfrak{a}_i \land Y)$. 
Then,  TPU ${\cal{P}}$  
computes $\mathfrak{f}_{\cal{P}} \leftarrow  \mathfrak{c}_{\cal{P}} \oplus(\mathfrak{b}_{\cal{P}} \land X) \oplus (\mathfrak{a}_{\cal{P}}  \land Y) \oplus (X \land Y)$.   Here, we denote the protocol as 
    $\langle  \mathfrak{f} \rangle \leftarrow   \texttt{SBM}(\langle  \mathfrak{x} \rangle,\langle \mathfrak{y}\rangle )$.   
    
     \subsubsection{Secure Bit-wise Addition Protocol (\texttt{BAdd})}
  The  \texttt{BAdd} describes as follows: 
  the TPU $i$ holds bit shares $\mathfrak{a}_i^{(\ell)},\cdots, \mathfrak{a}_i^{(1)}$ of $\ell$ bit-length integer $a$ and $\mathfrak{r}_i^{(\ell)},\cdots, \mathfrak{r}_i^{(1)}$ of $\ell$ bit-length integer $r$. The goal is to compute $\mathfrak{y}_i^{(\ell)},\cdots, \mathfrak{y}_i^{(1)}$, such that $y = a+r$, where $ {y} = - \mathfrak{y}^{(\ell)} 2^{\ell-1} + \sum_{j=1}^{\ell-1} \mathfrak{y}^{(j)} 2^{j-1} $,  $\mathfrak{a}^{(j)} = \mathfrak{a}^{(j)}_1 \oplus \mathfrak{a}^{(j)}_2 \oplus \cdots \oplus \mathfrak{a}^{(j)}_{\cal{P}}$, $\mathfrak{y}^{(j)} = \mathfrak{y}^{(j)}_1 \oplus \mathfrak{y}^{(j)}_2 \oplus \cdots \oplus \mathfrak{y}^{(j)}_{\cal{P}}$ and  $\mathfrak{r}^{(j)}  = \mathfrak{r}^{(j)}_1 \oplus \mathfrak{r}^{(j)}_2 \oplus \cdots \oplus \mathfrak{r}^{(j)}_{\cal{P}}$. The idea is easy and simple:  use the binary addition circuit to achieve the addition, i.e,  compute  the integer addition as  $\mathfrak{y}^{(j)} = \mathfrak{a}^{(j)} \oplus \mathfrak{r}^{(j)} \oplus \mathfrak{c}^{(j)}$ and $\mathfrak{c}^{(j+1)} = (\mathfrak{a}^{(j)} \land \mathfrak{r}^{(j)} ) \oplus ((\mathfrak{a}^{(j)} \oplus \mathfrak{r}^{(j)}) \land \mathfrak{c}^{(j)})$ for $j = 1,\cdots,\ell.$
    The procedure of \texttt{BAdd} works as follows:
    
   1)   For each TPU $i$ $(i = 1,\cdots, {\cal{P}})$ and each bit position $j = 1,\cdots, \ell$, all the TPUs jointly  compute 
   $\mathfrak{d}_i^{(j)} \leftarrow \mathfrak{a}_i^{(j)} \oplus \mathfrak{r}_i^{(j)}$ and
  $ \langle \mathfrak{e}^{(j)}  \rangle \leftarrow  \texttt{SBM}(\langle \mathfrak{a}^{(j)} \rangle, \langle \mathfrak{r}^{(j)} \rangle)$.   After using  the computation of \texttt{SBM}, it indeed computes $\mathfrak{e}^{(j)} = \mathfrak{a}^{(j)} \land \mathfrak{r}^{(j)}.$

2)  Each TPU $i $  sets $\mathfrak{c}_i^{(1)} \leftarrow 0$ and $\mathfrak{y}_i^{(1)} \leftarrow \mathfrak{d}_i^{(1)}$.  Then,    for $j = 2,\cdots, \ell$, all TPUs jointly computes 
  \begin{center} 
 $\langle \mathfrak{d}^{(j-1)} \rangle \leftarrow  \texttt{SBM}(\langle \mathfrak{d}^{(j-1)} \rangle, \langle \mathfrak{c}^{(j-1)} \rangle)$.
  \end{center} 
   Moreover, for each TPU $i$  locally computes
  \begin{center}
 $\mathfrak{c}_i^{(j)} \leftarrow  \mathfrak{d}_i^{(j-1)}  \oplus \mathfrak{e}_i^{(j-1)}$ and 
 $\mathfrak{y}_i^{(j)} \leftarrow    \mathfrak{d}_i^{(j)}  \oplus \mathfrak{c}_i^{(j)}$.    
 \end{center}
 and outputs  $\mathfrak{y}_i^{(j)} $ for all $j$.

      \subsubsection{Secure Bit Extraction  Protocol (\texttt{BExt})}
  
  Suppose  TPU  $i$  $(i=1,\cdots,{\cal{P}})$ contains an integer share $ {u}_i$, where  $u = \sum_{i=1}^{{\cal{P}}} u_i$.
  The goal of \texttt{BExt} is to output the bit extraction  shares  $\mathfrak{u}^{(\ell)}_i,\cdots, \mathfrak{u}^{(1)}_i$ for each TPU $i$  ($i = 1,\cdots, {\cal{P}}$), where $ {u} = - \mathfrak{u}^{(\ell)} 2^{\ell-1} + \sum_{j=1}^{\ell-1} \mathfrak{u}^{(j)} 2^{j-1} $  and $\mathfrak{u}^{(j)} = \bigoplus_{i=1}^{{\cal{P}}} \mathfrak{u}^{(j)}_i$. The \texttt{BExt} also contains offline/online phase which describes as follows:

  \textbf{Offline Phase:} Execute \texttt{RTG} to get $\mathfrak{r}_i^{(\ell)},\cdots, \mathfrak{r}_i^{(1)}$ and $r_i$ for party $i$.
  Then,  all the TPUs need to jointly compute $\mathfrak{a}^{(\ell)},\cdots, \mathfrak{a}^{(1)} \in \mathbb{Z}_2$, such that $\mathfrak{a}^{(\ell)} \oplus \cdots \oplus \mathfrak{a}^{(1)} = 0$.
  Firstly, 
  TPU 1 randomly generates $\mathfrak{a}_1^{(\ell)},\cdots, \mathfrak{a}_1^{(1)} \in \mathbb{Z}_2$ and let
   $\mathfrak{t}^{(j)} = \mathfrak{a}_1^{(j)} $ for $j = 1,\cdots,{\cal{P}}$. After that, the TPU $i$
    generates $\mathfrak{a}_i^{(\ell)},\cdots, \mathfrak{a}_i^{(1)} \in \mathbb{Z}_2$, computes 
    \begin{center}
$[\![\mathfrak{t}^{(j)}]\!] \leftarrow [\![\mathfrak{t}^{(j)}]\!]^ {(1- \mathfrak{a}_i^{(j)})} \cdot  ([\![1]\!] \cdot [\![\mathfrak{t}^{(j)}]\!]^{N-1})^{\mathfrak{a}_i^{(j)}} =  [\![\mathfrak{t}^{(j)} \oplus \mathfrak{a}_i^{(j)}]\!]$, 
\end{center}
    and sends these ciphertexts to TPU $i+1$.
Once  the $[\![\mathfrak{t}^{(\ell-1)}]\!],\cdots, [\![\mathfrak{t}^{(0)}]\!]  $ are received, the TPU ${\cal{P}}$   uses the \texttt{SDD} to decrypt, gets $\mathfrak{t}^{(\ell)},\cdots, \mathfrak{t}^{(1)}$ and denotes them as $\mathfrak{a}_{\cal{P}}^{(\ell)},\cdots, \mathfrak{a}_{\cal{P}}^{(1)}$.  After that, each TPU $i$ seals $\mathfrak{r}_i^{(\ell)},\cdots, \mathfrak{r}_i^{(1)}$ and $r_i$ and $\mathfrak{a}_i^{(\ell)},\cdots, \mathfrak{a}_i^{(1)}$ in UnS, respectively.

 \textbf{Online Phase:} The TPU  $i$ computes $v_i = u_i-r_i,$ encrypts $v_i$ and  sends $[\![v_i]\!]$ to TPU ${\cal{P}}$. After received all the encryptions, the TPU ${\cal{P}}$
  computes $[\![v]\!]  \leftarrow \prod_{i=1}^n [\![v_i]\!]$ and executes \texttt{SDD} to get the 
   $v$, and computes $\lceil v \rfloor_N$.  Then,  TPU ${\cal{P}}$ generates its two’s complement binary representation $\mathfrak{v}^{(\ell-1)},\cdots, \mathfrak{v}^{(0)},$
  and computes $\mathfrak{v}^{(j)}_{\cal{P}} \leftarrow \mathfrak{v}^{(j)} \oplus \mathfrak{a}_{\cal{P}}^{(j)}$, where $j = 1,\cdots, \ell$. Other TPU $i$ $(i = 1,\cdots,{\cal{P}}-1)$ keeps other $\mathfrak{v}^{(\ell)}_i \leftarrow \mathfrak{a}^{(\ell)}_i,\cdots, \mathfrak{v}^{(1)}_i  \leftarrow \mathfrak{a}^{(1)}_i$ unchanged.

  After that, all the TPUs jointly compute  
  \begin{center}
  $(\vec{\mathfrak{u}}_1,\cdots,\vec{\mathfrak{u}}_{\cal{P}}) \leftarrow \texttt{BAdd} (\vec{\mathfrak{v}}_1,\cdots,\vec{\mathfrak{v}}_{\cal{P}}; \vec{\mathfrak{r}}_1,\cdots,\vec{\mathfrak{r}}_{\cal{P}})$, 
    \end{center}
  where     
  $\vec{\mathfrak{u}}_i = (\mathfrak{u}_i^{(\ell)} ,\cdots, \mathfrak{u}_i^{(1)})$, $\vec{\mathfrak{v}}_i = (\mathfrak{v}_i^{(\ell)} ,\cdots, \mathfrak{v}_i^{(1)})$, $\vec{\mathfrak{r}}_i = (\mathfrak{r}_i^{(\ell)} ,\cdots, \mathfrak{r}_i^{(1)})$. Finally, the \texttt{BExt} algorithm outputs
   $\vec{\mathfrak{u}}_i = (\mathfrak{u}_i^{(\ell)} ,\cdots, \mathfrak{u}_i^{(1)})$ for TPU $i = 1,\cdots, {\cal{P}}$.

       \subsection{Secure Integer Computation }
       \label{TPUSIC}
   \subsubsection{Secure Multiplication Protocol (\texttt{SM})}

The   {\texttt{SM}}      achieves integer  multiplication over integer shares, i.e.,  given   shares $x_i, y_i$ $(i = 1,\cdots, {\cal{P}})$ for  TPU $i$ as input,  {\texttt{SM}}   securely outputs $f_i$ for TPU $i$, such that $\sum_{i=1}^{\cal{P}} f_i = x \cdot y,$ where data shares $x_i, y_i$ satisfy $x = \sum_{i=1}^{\cal{P}} y_i$ and $y = \sum_{i=1}^{\cal{P}} y_i$.

   \textbf{Offline Stage}:    All the TPUs  initialize  their enclaves and load the public parameters to the UnS. Then, for the 
 enclave  1,  it generates $a_1, b_1 \in \mathbb{D}_N$,  computes $z = a_1 \cdot b_1$,  encrypts $[\![a_1]\!],[\![b_1]\!], [\![z]\!]   $, and lets them be $[\![a]\!],[\![b]\!], [\![c]\!]$, respectively. After that, enclave $i$ ($i = 1,\cdots,{\cal{P}}-1$) sends $[\![a]\!], [\![b]\!], [\![c]\!]$ to enclave $i+1$, TPU $i+1$ generates $a_{i+1}, b_{i+1}$ and computes 
 \begin{center}
 $[\![c]\!] \leftarrow [\![c]\!] \cdot [\![a_{i+1}\cdot b_{i+1}]\!] \cdot  [\![a]\!]^{b_{i+1} }  \cdot  [\![b]\!]^{a_{i+1}} $,   \end{center}\begin{center}
 $[\![a]\!] \leftarrow [\![a]\!] \cdot [\![a_{i+1}]\!],$   $[\![b]\!] \leftarrow [\![b]\!] \cdot [\![b_{i+1}]\!]$.
  \end{center}
  After the computation, for  $i = {\cal{P}},\cdots, 2$,  TPU enclave $i$ generates   $c_i \in \mathbb{D}_N$
  and
  computes $[\![c]\!] = [\![c]\!] \cdot [\![c_i]\!]^{N-1}.$ After the computation, the TPU   $2$ sends   $[\![c]\!]$ to TPU   1.
  Then, TPU $1$ uses \texttt{SDD} to get $c$ and denotes the final result $\lfloor c \rceil_N $ as $c_1$. 
 After the above computation, each enclave hold $a_i, b_i, c_i$, such that $
 \lceil a_1+\cdots + a_{\cal{P}} \rfloor_N
  =  \lceil a \rfloor_N$,  $\lfloor b_1 +\cdots + b_{\cal{P}} \rceil_N = \lceil b \rfloor_N$, $\lceil c_1 + \cdots + c_{\cal{P}} \rfloor_N = \lceil c \rfloor_N$ and $c = a\cdot b \mod N$. After the computation, each TPU enclave $i$ seals $a_i, b_i,c_i$ to UnS for storage individually. 
  
\textbf{Online Stage:}    TPU   $i$ loads the   $a_i, b_i, c_i$ into the enclave $i$. Then,  
compute $X_i =  x_i - a_i  $ and $Y_i =  y_i  - b_i$. Securely send $X_i$ and $Y_i$ to other enclave $j$ $(j = 1,\cdots,{\cal{P}}; j \neq i)$.
After receiving  other $X_j$ and $Y_j$, the each TPU $i$ computes  $X = \sum_{i=1}^{{\cal{P}}}X_j$ and    $Y = \sum_{i=1}^{{\cal{P}}} Y_j$.
After that, 
for each TPU $i$ $(i = 1,\cdots, {\cal{P}}-1)$, compute $f_i\leftarrow 
\lceil c_i + b_i X + a_i Y  \rfloor_N
 $. 
For  TPU ${\cal{P}}$,  
 compute $f_{\cal{P}} \leftarrow    \lceil c_{\cal{P}} + b_{\cal{P}} X + a_{\cal{P}} Y + X \cdot Y  \rfloor_N$.   Here, we denote the protocol as 
   $  \langle f \rangle \leftarrow  \texttt{SM}( \langle x \rangle,\langle y\rangle )$.   
     
   \subsubsection{Secure Monic Monomials Computation   (\texttt{SMM})}
  
 The   {\texttt{SMM}}  protocol can achieve monic monomials  computation over integer shares, i.e.,  given   a share $x_i$ $(i = 1,\cdots, {\cal{P}})$ and a public integer number $k$ for  TPU $i$ as input,  {\texttt{SMM}}   securely outputs $f_i$ for TPU $i$, such that $\sum_{i=1}^{\cal{P}} f_i = x^k,$ where data shares $x_i$ satisfy $x = \sum_{i=1}^{\cal{P}} x_i$. The construction of the {\texttt{SMM}} is list as follows:
Denote $k$ as binary form  $\mathfrak{k}_\ell,\cdots, \mathfrak{k}_1$. Initialize the share $f_i \leftarrow x_i$ for each TPU $i$.
For $j = \ell-1,\cdots, 1$, compute    $  \langle f^* \rangle \leftarrow  \texttt{SM}( \langle f \rangle,\langle f \rangle )$. If $\mathfrak{k}_j = 1$, compute $  \langle f \rangle \leftarrow  \texttt{SM}( \langle f^* \rangle,\langle x \rangle )$. Otherwise, let $  \langle f \rangle \leftarrow  \langle f^* \rangle$.  Here, the algorithm outputs $\langle f \rangle$ and  denotes the protocol as 
   $  \langle f \rangle \leftarrow  \texttt{SMM}( \langle x \rangle, k )$.



      \subsubsection{Secure Binary  Exponential Protocol ($\texttt{SEP}_2$)}
    The    $\texttt{SEP}_2$   can achieve   exponential over binary  shares with a public base, i.e.,  given  a binary  share $\mathfrak{x}_i \in \mathbb{Z}_2$ $(i = 1,\cdots, {\cal{P}})$ and a public integer  $\beta$ for  TPU $i$ as input\footnote{$\beta$ is   a small positive number which satisfies $gcd(\beta, N) =1$. }, $\texttt{SEP}_2$   securely outputs an integer share $f_i \in \mathbb{Z}_N$ for TPU $i$, such that $\sum_{i=1}^{\cal{P}} f_i = \beta^{\mathfrak{x}},$ where   $\mathfrak{x} =   \bigoplus_{i=1}^{{\cal{P}}} \mathfrak{x}_i$.

     \textbf{Offline Stage}:    All the TPUs  initialize  their enclaves and load the public parameters to the UnS. Then, for the 
 enclave  1,  it generates ${\mathfrak{a}}_1\in \mathbb{Z}_2$,   encrypts $\mathfrak{a}_1$ as $[\![\mathfrak{a}_1]\!]$, and lets it be $[\![\mathfrak{a}]\!] $.
  After that, enclave $i$ ($i = 1,\cdots,{\cal{P}}-1$) sends $[\![\mathfrak{a}]\!]  $ to enclave $i+1$, TPU $i+1$ generates $\mathfrak{a}_{i+1} \in \mathbb{Z}_2$,  computes     \begin{center}
     $[\![\mathfrak{a}]\!] \leftarrow [\![\mathfrak{a}]\!]^ {(1- \mathfrak{a}_{i+1})} \cdot  ([\![1]\!] \cdot [\![\mathfrak{a}]\!]^{N-1})^{\mathfrak{a}_{i+1}} = [\![\mathfrak{a} \oplus \mathfrak{a}_{i+1}]\!], $     \end{center}
  
  Once    $[\![a]\!]$ is received, TPU computes  
   \begin{center}
 $[\![b]\!]= [\![\mathfrak{a}]\!]^ {\beta} \cdot  ([\![1]\!] \cdot [\![\mathfrak{a}]\!]^{N-1}) = [\![\beta \cdot \mathfrak{a} +(1-\mathfrak{a})]\!] =  [\![\beta^{\mathfrak{a}} ]\!] $
   $[\![b^*]\!]= [\![\mathfrak{a}]\!] \cdot  ([\![1]\!] \cdot [\![\mathfrak{a}]\!]^{N-1})^ {\beta} = [\![   \mathfrak{a} + \beta(1-\mathfrak{a})]\!] =  [\![\beta^{1-\mathfrak{a}} ]\!] $
  \end{center}

   After the computation, for  $i = {\cal{P}},\cdots, 2$,  TPU   $i$ generates   $b_i, b^*_i  \in \mathbb{D}_N$
  and
  computes $[\![b]\!] = [\![b]\!] \cdot [\![b_i]\!]^{N-1}$ and $[\![b^*]\!] = [\![b^*]\!] \cdot [\![b^*_i]\!]^{N-1}.$ After the computation, the TPU   $2$ sends   $[\![b]\!]$ and $[\![b^*]\!]$ to TPU   1, and TPU 1 uses \texttt{SDD} to get $b$, $b^*$ and denote them as $b_1$ and $b^*_1$, respectively.
 After the above computation, each TPU $i$ holds $\mathfrak{a}_i, b_i$, which satisfies  $
   \mathfrak{a}_1\oplus \cdots \oplus \mathfrak{a}_{\cal{P}}  
  = \mathfrak{a}$,  $
    b_1+ \cdots + b_{\cal{P}}  
  = \beta^\mathfrak{a}$, $
    b^*_1+ \cdots + b^*_{\cal{P}}  
  = \beta^{1-\mathfrak{a}}$. After the computation, each TPU   $i$ seals $ \mathfrak{a}_i, b_i$ to UnS for storage individually.


    
    \textbf{Online Stage:}   
     TPU   $i$ loads the data share  $\mathfrak{x}_i$ and random shares $\mathfrak{a}_i, b_i$ into the its enclave. Then, TPU $i$ locally 
computes $X_i =   {\mathfrak{x}_i \oplus \mathfrak{a}_i}  $. Securely send $X_i$  to other enclave $j$ $(j = 1,\cdots,{\cal{P}}; j \neq i)$.
After receiving  other $X_j$,   each TPU $i$ locally computes  $X = \bigoplus_{i=1}^{{\cal{P}}}X_i$ and   $f_i\leftarrow  \lceil     (b^*_i)^{X} \cdot (b_i)^{1-X} \rfloor_N$. We can easily verify that $\sum_{i=1}^{\cal{P}} f_i = \beta^{(1-\mathfrak{a})({\mathfrak{x} \oplus \mathfrak{a}})+ \mathfrak{a}{(1-\mathfrak{x} \oplus \mathfrak{a}})}=  \beta^{\mathfrak{x}}$. 
   Here, we denote the protocol as 
   $  \langle f \rangle \leftarrow  \texttt{SEP}_{2}( \langle {\mathfrak{x}} \rangle, \beta )$.

   \subsubsection{Secure Integer  Exponential Protocol (\texttt{SEP})}

The    {\texttt{SEP}}    can achieve   exponential over integer  shares with a public base, i.e.,  given  an integer  share $x_i \in \mathbb{D}_N$ $(i = 1,\cdots, {\cal{P}})$ and a public integer  $\beta$ for  TPU $i$ as input , {\texttt{SEP}}   securely outputs shares $f_i \in \mathbb{D}_N$ for TPU $i$, such that $\sum_{i=1}^{\cal{P}} f_i = \beta^{x},$ where data shares $x_i$ satisfy $x = \sum_{i=1}^{\cal{P}} x_i$ and $x $ is relative small positive number with $\ell$ bit-length.

i) Compute 
 $(\vec{\mathfrak{x}}_1,\cdots,\vec{\mathfrak{x}}_{\cal{P}}) \leftarrow \texttt{BExt} (x_1,\cdots,x_{\cal{P}}),$ 
  where $\vec{\mathfrak{x}}_i = (\mathfrak{x}_i^{(\ell)} ,\cdots, \mathfrak{x}_i^{(1)})$ for TPU $i = 1,\cdots, {\cal{P}}$, and $ \mathfrak{x}^{(j)} =  \bigoplus_{i=1}^{\cal{P}} \mathfrak{x}_i^{(j)}$,  and $x = \sum_{j=1}^{\ell} \mathfrak{x}^{(j)} 2^{j-1}.$
  
ii) Execute $ \langle {\mathfrak{f}} \rangle \leftarrow \texttt{SEP}_2 ( \langle {\mathfrak{x}}^{(1)} \rangle, \beta).$ For $j = 2,\cdots, \ell$, compute $ \langle { {f}}_j \rangle \leftarrow \texttt{SEP}_2 ( \langle {\mathfrak{x}}^{(j)} \rangle, \beta)$, $ \langle { {f}}^*_j \rangle \leftarrow \texttt{SMM} ( \langle { {f}}_{j} \rangle, 2^{j-1})$,  and $  \langle { {f}}  \rangle \leftarrow \texttt{SM} ( \langle { {f}}  \rangle, \langle { {f}}^*_j  \rangle)$. The  {\texttt{SEP}} outputs  $ \langle { {f}}  \rangle$, and we denote the protocol as 
   $  \langle f \rangle \leftarrow  \texttt{SEP}( \langle x \rangle,\beta)$.

    \subsubsection{Secure Comparison  Protocol (\texttt{SC})}
 The \texttt{SC} can securely compute the relationship between integer $u$ and $v$, where 
  each TPU $i$ holds shares $u_i$ and $v_i$, where
$u = u_1 + \cdots + u_{{\cal{P}}}$, $v = v_1 + \cdots + v_{{\cal{P}}}$.  The construction of \texttt{SC} is listed  as follows:

 i) Each TPU $i$ ($i =1,\cdots, {\cal{P}}$) locally computes $w_i = u_i - v_i$.  After that,  all TPUs jointly compute  
\begin{center}
 $(\vec{\mathfrak{w}}_1,\cdots,\vec{\mathfrak{w}}_{\cal{P}}) \leftarrow \texttt{BExt} (w_1,\cdots,w_{\cal{P}}).$ 
  \end{center} 
  
   ii)  As we use two’s complement binary representation, the most significant digit of $u-v$ will reflect the relationship between the $u$ and $v$. After the above computation, TPU $i$ outputs $\mathfrak{w}_i^{(\ell-1)} \in  \vec{\mathfrak{w}}_i$. The most significant digit $\mathfrak{w}^{(\ell-1)}$ of $w = \sum_{i=1}^{\cal{P}} w_i$  decides the relationship of  $u$ and $v$, specifically, 
   if  $\bigoplus_{i=1}^{\cal{P}} \mathfrak{w}_i^{(\ell-1)} = 0$, it denotes  $u\geq v$. Otherwise, it denotes $u < v$.
  
  \subsubsection{Secure Equivalent  Protocol  (\texttt{SEQ})}
The goal of secure equivalent  protocol \texttt{SEQ} is to test whether the  two values $u, v$ are equal or not by giving the  shares of the two values $\langle u\rangle ,\langle v \rangle$. Mathematically, 
given  two  shares  $\langle u\rangle$  and $\langle v \rangle$,  {\texttt{SEQ}} \cite{liu2016efficientx} outputs the  shares  $\mathfrak{f}_i$ for each TPU $i$ $(i=1,\cdots, {\cal{P}})$ to determine whether the plaintext of the two   data are equivalent (i.e. test $u\stackrel{?}{=}v $. If $\bigoplus^{\cal{P}}_{i=1} \mathfrak{f}_i = 1 $, then $ u=v$; otherwise, $u \neq v$).  The {\texttt{SEQ}}  is described as follows: 

i) All the TPUs jointly calculate

\begin{center}
  $\langle \mathfrak{t}^*_{1}\rangle \leftarrow \texttt{SC} (\langle u\rangle ,\langle v \rangle);$ 
  $\langle \mathfrak{t}^*_{2}\rangle \leftarrow \texttt{SC} (\langle v \rangle ,\langle u\rangle ).$ 
\end{center}

ii) For each TPU $i$, it computes $\mathfrak{f}_i = \mathfrak{t}^*_{1,i} \oplus \mathfrak{t}^*_{2,i}$ locally, and outputs $\mathfrak{f}_i   \in \mathbb{Z}_2$.

   \subsubsection{Secure   Minimum of Two Number Protocol ($\texttt{Min}_2$)}

 The TPU $i$ $(i=1,\cdots, {\cal{P}})$ stores shares  $\langle x \rangle$ and $\langle y\rangle$ of two  numbers    $x$ and     $y$,     The  ${\texttt{Min}}_2$ protocol outputs share $\langle B \rangle$ of   minimum       number  $B $, s.t.,  $ B = min(x,y )$.  
 The   ${\texttt{Min}}_2$ is described as follows:
 
 i) All the TPUs can jointly compute
 \begin{center}
  $\langle \mathfrak{u} \rangle  \leftarrow \texttt{SC}( \langle  x\rangle ,\langle y\rangle);$ 
  $ \langle u \rangle \leftarrow \texttt{B2I} (\langle \mathfrak{u} \rangle);$
  $\langle X \rangle \leftarrow \texttt{SM} (\langle x \rangle ,\langle u \rangle).$ 
  \end{center}
   \begin{center}
    $\langle Y \rangle  \leftarrow \texttt{SM} (\langle y\rangle ,\langle u\rangle ).$
        \end{center}

ii)    The TPU $i$  computes locally and outputs 
$B_i = y_i - Y_i + X_i$.

  \subsubsection{Secure   Minimum of $H$ Numbers Protocol   ($\texttt{Min}_H$)}

 The goal of $\texttt{Min}_H$ is to get the  minimum number among $H$ numbers. 
Given      the shares  $ x_{1,i},  \cdots,  x_{H,i} $ for TPU $i$, the goal is to compute the  share  $ x^*_i$ for TPU $i$ such that $x^*$   stores the minimum  integer value among  $x_1,\cdots,x_H$, where $x^* = \sum_{i=1}^{{\cal{P}}} x^*_{i}$, $x_j = \sum_{i=1}^{\cal{P}} x_{i,j}$ for $j=1,\cdots, H$.
The $\texttt{Min}_{H}$  executes as follows: Each TPU $i$
puts   $x_{1,i} ,\cdots, x_{H,i}$  into a set $S_i$. If $ {\cal{L}} ({S_i}) = 1$, the share remaining in ${\cal{L}}({S_i}) $  is  the final output. Otherwise, the protocol  is processed  according to the following  conditions.

$\bullet$ If   $ {\cal{L}}({S_i})  \mod 2=0$ and ${\cal{L}}({S_i})  > 1$,  1) set $S'_i \leftarrow \emptyset$;
2) for $j = 1,\cdots,    {\cal{L}}({S_i}) /2  $, 
compute 
\begin{equation}
 \langle x_j \rangle   \leftarrow \texttt{Min}_2( \langle x_{2j-1} \rangle,  \langle x_{2j} \rangle),
\end{equation}
and add $x_{j,i} $ to the set $S'_i$;
3) clear   set $S_i$ and   let  ${S_i} \leftarrow S'_i $.

$\bullet$ If $  {\cal{L}}({S_i})  \mod 2 \neq 0$ and $ {\cal{L}}({S_i})  > 1$,   take out the last tuple  $ x_{{\cal{L}}({S_i})-1,i} $   from   set $S_i$ s.t.,   $ {\cal{L}}({S_i})  \mod 2=0$. Run the above procedure ($ {\cal{L}}({S_i})  \mod 2=0$ and $ {\cal{L}}({S_i})  > 1$) to
 generate set $S'_i$.    Put $ x_{ {\cal{L}}({S_i}) -1,i} $    into a set ${S}_i'$ and denote  $S_i' \leftarrow S_i$.
 
 After computation, each set $S_i$ in TPU $i$ only contains one element and we denote it as $x^*_{i}$. Thus, we 
 denote the algorithm as 
$ \langle x^* \rangle  \leftarrow \texttt{Max}_H( \langle x_{1} \rangle,  \cdots, \langle x_{H} \rangle).$

\subsection{Security Extension of Integer Computation}  

The above the secure computation  only consider the data privacy. Two types of information can be leaked to the adversary: 1) the  access pattern  of function's input, and 2) the access pattern of RU's result retrieve. Here, we give two security extension to achieve access pattern hiding and private information retrieve, respectively. 

    \subsubsection{Achieve Input Access pattern Hiding (APH)}  
    \label{sec:AAPH}
 As data are directly sealed in the UnS, the adversary may analysis the access pattern of  UnS without knowing the function's input.  Suppose the system contains  $H$ data  $x^*_1,\cdots,x^*_H \in \mathbb{D}_N$.  The data share $x_{j,i}$ are hold by each TPU $i$   $(j = 1,\cdots,H; i = 1,\cdots, \cal{P})$, such that $x_{j,1}+\cdots+x_{j,{\cal{P}}} = x^*_{j}$.  
  To achieve access pattern hiding, the homomorphic property of PCDD can be used. Specifically, 
    the RU uploads $[\![a_{1}]\!],\cdots, [\![a_{H}]\!]$ to each TPU $i$, s.t., for a specific  $1 \leq \gamma \leq H$, it has $a_\gamma = 1$, and other $j \neq \gamma$ and  $1\leq j\leq H$, it holds $a_j = 0$. Then, the goal of the  algorithm is to securely select the shares of $x_{\gamma,j}$ from the input shares, and constructs  as follows:
    
1) \textit{Obviously select encrypted shares}. Each  TPU  initializes    an enclave.  Then, for each TPU $i$ $(i=1,\cdots, {\cal{P}})$, compute
\begin{center}
$ [\![b_{i}]\!] \leftarrow [\![a_{1}]\!]^{x_{1,i}}\cdot [\![a_{2}]\!]^{x_{2,i}}\cdot \cdots \cdot[\![a_{H}]\!]^{x_{H,i}} \mod N^2$.
\end{center}
 
2) \textit{Securely update share $[\![b_i]\!]$ for   TPU $i$}. Without any share update, the adversary can still know the access pattern once the ciphertexts are decrypted. Thus, all the shares should be dynamically updated before the decryption.

 The TPU $i$ picks random numbers $\delta_{i,1},\cdots,\delta_{i, {\cal{P}}} \in \mathbb{Z}_N$ such that $\delta_{i,1}+\cdots+\delta_{i, {\cal{P}}} = 0 \mod N$, and then encrypts  $\delta_{i,j}  $ and  sends  $[\![\delta_{i,j}]\!]  $ to TPU  enclave $j$. Once all the update shares are received, TPU $i$ computes 
\begin{center}
$ [\![b^*_i]\!] \leftarrow  [\![b_i]\!] \cdot [\![\delta_{1,i}]\!]\cdot [\![\delta_{2,i}]\!]\cdot \cdots \cdot[\![\delta_{{\cal{P}},i}]\!]  \mod N^2$.
\end{center}

Finally, each TPU $i$ uses the \texttt{SDD} to get $b^*_i$ and denotes   $\lceil b^*_i \rfloor_N$ as the final share output.


 \subsubsection{Achieve Private Information Retrieve (PIR)}  
 \label{Sec:APIR}
If   the computation results is needed,  the RU will let the TPU to send the data shares back   via a secure channel. However, if one of the TPU has been compromised, even if the data cannot been known by the adversary, the retrieve access  pattern has been leaked to the adversary.  
Suppose the system contains  $H$ data  $x^*_1,\cdots,x^*_H \in \mathbb{D}_N$.  The data share $x_{j,i}$ are hold by each TPU $i$   $(j = 1,\cdots,H; i = 1,\cdots, \cal{P})$, such that $x_{j,1}+\cdots+x_{j,{\cal{P}}} = x^*_{j}$.      
Thus, to achieve the private information retrieve, the RU uploads $[\![a_1]\!],\cdots, [\![a_H]\!]$ to each TPU, s.t., for a specific  $1 \leq \gamma \leq H$, it has $a_\gamma = 1$, and other $j \neq \gamma, 1\leq j\leq H$, it holds $a_j = 0$. The goal of \texttt{PIR} is to let RU privately  retrieve $x_\gamma$. Then, the algorithm computes among all TPUs as follows:

  1)  For each TPU $i$, compute
\begin{center}
$ [\![b_{i}]\!] \leftarrow [\![a_{1}]\!]^{x_{1,i}}\cdot [\![a_{2}]\!]^{x_{2,i}}\cdot \cdots \cdot[\![a_{H}]\!]^{x_{H,i}} \mod N^2$.
\end{center}
 
2) TPU $1$ denotes $[\![b^*]\!] \leftarrow [\![b_1]\!]$, and sends $[\![b^*]\!]$ to TPU 2. Then, each TPU $i =2 ,\cdots, {\cal{P}}$, computes  $[\![b^*]\!] \leftarrow [\![b^*]\!] \cdot [\![b_{i}]\!] \mod N^2$. If $i =  {\cal{P}}$, then send $[\![b^*]\!] $  to RU. Otherwise,   $[\![b^*]\!] $ is sent from TPU $i$ to $i+1$.
Finally, RU  uses the \texttt{Dec} to get the  $b^*$, such that $x_{\gamma} = \lceil b^*_i \rfloor_N$ is the output share.
  

      \subsection{Secure  Floating Point Number Computation}
   
   \subsubsection{Data Format of Floating-Point Number}
   
     To achieve the  real number storage and computation, we can  refer to  the IEEE 754 standard to use Floating-Point Number (FPN) for real number storage.  To support the LightCom, we change the traditional FPN and  describe the FPN  by four integers: 1) a radix (or base) $\beta \geq 2$; 2) a precision $\eta  \geq 2$ (roughly speaking, $\eta$ is the number of ``significant digits'' of the representation); 3)  two extremal exponents $e_{min}$ and $e_{max}$ such that   $e_{min} < 0 < e_{max}$. A finite FPN $\hat{a}$ in such a format is a number for which
there exists at least one representation triplet $(m, e)$ with public parameters $\beta,\eta,e_{min},e_{max}$, such that,
  $$\hat{a}  =      m \cdot \beta^{e-\eta+1}.$$
  \begin{itemize}
\item {$m$ is an integer which $-\beta^{\eta} + 1 \leq m \leq \beta^{\eta} - 1$. It is called
the integral \textit{significand} of the representation of $x$;}
\item{$e$ is an integer such that $e_{min} \leq e \leq e_{max}$, called the \textit{exponent} of the
representation of $a$.}
\end{itemize}
 
As only the  \textit{significand} and  \textit{exponent}  contains sensitive information, we assume  all the FPNs have the same public base $\beta = 10$, and use   the fix bit-length to store the integer $m$. 
 Thus, 
 to achieve the secure storage, 
 the RU only needs to random share the $\hat{a}$ into $\hat{a}_1=(m_1,e_1),\cdots, \hat{a}_{\cal{P}}= (m_{\cal{P}},e_{\cal{P}})$, and sends $\hat{a}_i$ to TPU $i$ for storage, respectively.   

  For the secure FPN computation,  if  all the FPNs are transformed with the same exponential, we can directly use  secure integer computation method introduced in Section \ref{TPU:SOCT}.   Thus, the key problem to achieve the secure FPN computation is  how to allow all the FPNs securely transformed with the same exponential.  Here, we first construct an algorithm called Secure Uniform Computation (\texttt{UNI}) and then achieve the commonly-used  FPN computations.   

    \subsubsection{Secure Uniform Computation (\texttt{UNI})}

     Assume each TPU $i (i=1,\cdots, {\cal{P}})$  stores   into $\hat{a}_{j,i}=(m_{j,i},e_{j,i})$ , the goal of \texttt{UNI} is to output $\hat{a}^*_{j,1}=(m^*_{j,1},e^*_{j,1})$ for $j = 1,\cdots,H$,
and  the construction of \texttt{UNI} can be described as follows: 
   
i) All the TPUs jointly  compute     \begin{equation}
  \langle e^* \rangle  \leftarrow \texttt{Min}_H( \langle e_{1} \rangle,  \cdots, \langle e_{H} \rangle).
\end{equation}
    
 ii)  Each TPUs locally computes $\langle c_{j}  \rangle=   {\langle e_{j} \rangle - \langle e^*_{i}}\rangle  $. As $e_{j}- e^*$ is a relative small number,  TPUs jointly executes 
$   \langle 10^{  e_{j} - e^*  } \rangle  \leftarrow \texttt{SEP}( \langle c_{j} \rangle, 10)$ and   
   $\langle  m^*_j \rangle \leftarrow \texttt{SM} (\langle 10^{  e_{j} - e^*}  \rangle,\langle m_j \rangle)$.  
  
  After computation, all the $\langle a_{1} \rangle,\cdots, \langle a_{H} \rangle $ will transform to $\langle a^*_{1} \rangle,\cdots, \langle a^*_{H} \rangle $ which shares the same $e^*$, where $\langle \hat{a}_{j} \rangle = (\langle m^*_j \rangle, \langle e^* \rangle).$

 \subsubsection{Computation Transformation}
 
  The secure floating-point number computation  can be transformed into the secure integer computation protocols with the usage of \texttt{UNI}. Formally, given FPN shares   $\langle \hat{a}_{j}\rangle =(\langle m_{j} \rangle, \langle e_{j} \rangle) $, (for $j = 1,\cdots,H$), we can first compute 
  \begin{center}
  $(\langle \hat{a}^*_{1}\rangle,\cdots, \langle \hat{a}^*_H \rangle) \leftarrow \texttt{UNI}(\langle \hat{a}_{1}\rangle,\cdots, \langle \hat{a}_H \rangle),$
   \end{center}
   where $\langle \hat{a}^*_{j}\rangle = (\langle  {m}^*_{j}\rangle,\langle \hat{e}^*\rangle)$. 
 Then, 
   \begin{center}
  $( \langle y^*_1 \rangle,\cdots, \langle y^*_\zeta \rangle) \leftarrow {\cal{SIF}}( \langle m^*_1 \rangle,\cdots, \langle m^*_\xi \rangle),$
   \end{center}
   where $\cal{SIF}$ denote secure integer computation protocol designed in Section \ref{TPU:SOCT},  and $\langle y^*_1 \rangle,\cdots, \langle y^*_\zeta \rangle$  can be either integer shares or binary shares according to the function type. 
   If the ${\cal{SIF}}$ is the \texttt{SC} and \texttt{SEQ}, then the ${\cal{SIF}}$ output the binary share $\langle y^* \rangle$ as the final output, and we denote these two algorithms as secure FPN comparison  (\texttt{FC}) and secure FPN  equivalent  test protocol (\texttt{FEQ}).    If the ${\cal{SIF}}$ is the \texttt{SM},   \texttt{SMM}, $\texttt{Min}_2$ and  $\texttt{Min}_H$, then the ${\cal{SIF}}$ outputs the integer  share $\langle  {y}^*_1 \rangle$, and denotes
   $\langle \hat{y}^* \rangle = (\langle  {y}^*_1 \rangle, \langle e^*\rangle)$ as the secure FPN's output, and we denote above four algorithms as secure FPN multiplication   (\texttt{FM}), secure FPN monic monomials computation (\texttt{FMM}), secure minimum of two FPNs protocol   ($\texttt{FMin}_2$), and secure minimum of $H$ FPNs protocol   ($\texttt{FMin}_H$), respectively. 
   Specifically, for the multiple  FPN addition (\texttt{FAdd}), given FPN shares   $\langle \hat{a}_{j}\rangle =(\langle m_{j} \rangle, \langle e_{j} \rangle) $, (for $j = 1,\cdots,H$), we can first compute  $\langle \hat{a}^*_{1}\rangle ,\cdots, \langle \hat{a}^*_{H}\rangle $ with the $\texttt{UNI}$, 
   where $\langle \hat{a}^*_{j}\rangle = (\langle  {m}^*_{j}\rangle,\langle  {e}^*\rangle)$. Then, compute $\langle  {y}^*\rangle \leftarrow \sum_{j=1}^{H} \langle  {m}^*_{j}\rangle$ and denote the final FPN addition  result as $\langle \hat{y}\rangle  =  (\langle  {y}^*\rangle, \langle  {e}^*\rangle)$.

\subsubsection{Secure Extension  for FPN Computation}  
Similar to the secure integer computation, we have the three following extension for LightCom.

    \textbf{Access Pattern Hiding: }
 As all the secure FPN  computation can be transformed in to secure integer computation with the  help of the \texttt{UNI}, we can also use  the same  method   in section   \ref{sec:AAPH} to achieve  input access pattern hiding  for the secure FPN  computation.

 \textbf{Achieve Private FPN Retrieve:}  
In out LightCom, one floating point number can be securely stored as two    integer numbers.  Thus, we can use the method in section \ref{Sec:APIR} to  privately retrieve integer for twice to   achieve the private  FPN retrieve.

     \subsection{Functional Extension for LightCom}  
   \subsubsection{Non-numerical Data Storage and Processing}

   For the non-numerical data storage,   the   traditional character encodings with Unicode~\cite{unicode1997unicode}   and its standard Unicode Transformation Format (UTF) schemes can be used  which maps a character into an integer.    
   Specifically, for secure storage,   use UTF-8 to map the character into 32-bit number $x$, randomly splits $x$ into $x_1,\cdots,x_{\cal{P}}$, such that $x_1 +\cdots+x_{\cal{P}} = x $, and sends   $a_i$ to TPU $i$ for processing.   In this case, all the non-numerical data processing can be transformed into secure integer computation which can be found in section \ref{TPU:SOCT}. For the secure storage, each TPU $i$ securely seals the share $a_i$ into the UnS with the algorithm \texttt{Seal} in Section \ref{STPUDATA}.
   Once the data shares are needed for processing, TPUs needs to use \texttt{UnSeal} algorithm to recover the message from UnS.
   
     \subsubsection{Extension of Multiple User Computation}  
     
All the secure  computations in the previous section are designed for the single user setting, i.e., all the data are encrypted under a same RU's public key.   If all RUs  want to jointly achieve a secure computation, each RU $j (j=1,\cdots,  \psi)$  executes \texttt{KeyGen}  to generate  public key $pk_j$ and  private key is $sk_j$ locally. Then, RU $j$ uses \texttt{KeyS} to split key $sk_j$ into $\cal{P}$ shares $\langle sk_j \rangle$, and sends these shares to TPUs in the cloud.  Assume RU $j$'s ciphertext $[\![x_j]\!]_{pk_j}$ is securely stored in UnS, TPUs  can  get data shares $\langle x_j \rangle$  with $\texttt{UnSeal}$ and achieve the corresponding secure computations $\texttt{GenCpt}$ in Section \ref{Sec:GenCpt} with these shares.

    \section{Security Analysis }
  \label{sec:Securityanalysis}
  In this section, we first analyze the security of the basic crypto primitives and the sub-protocols, before demonstrating the security of our LightCom framework.

  \subsection{Analysis of  Basic Crypto Primitives}

   \subsubsection{The Security of   Secret Sharing Scheme}
   Here, we give the following theorem to show the  security of the additive secret sharing scheme.
   \begin{theorem}
   \label{Theo:1}
  A additive secret sharing scheme achieves an information theoretic secure when the  $\cal{P}$ participants can reconstruct the secret $x \in \mathbb{G}$, while any smaller set cannot discover anything information about the secret.
     \end{theorem}
     \begin{proof}
     The   shares  $X_1,\cdots,X_{\cal{P}}$ are selected with random uniform distribution  among $\cal{P}$ participants such that $X_1+\cdots+X_{\cal{P}} = m \in  \mathbb{G}$.  Even the attacker $\cal{A}$ holds ${\cal{P}}-1$ shares, (s)he can only compute $x' = \sum_{i=1}^{{{\cal{P}}-1}} X'_i$, where $X'_i$ is selected from  $X_1,\cdots, X_{\cal{P}}$. The element $x$ is still protected due to the $x = x' + X'_{{\cal{P}}}$. Since random value $X'_{{\cal{P}}}$ is unknown for $\cal{A}$, it   leaks no information about the value $x$. 
       \end{proof}
       
          \begin{theorem}
   \label{Theo:2}
  A proactive additive  secret
sharing  scheme achieves an information theoretic secure if satisfies the following properties:
\textbf{I. Robustness:} The new updated shares are corresponding  to the secret  $x$ (i.e., all the new shares can reconstructed the secret  $x$). 
\textbf{II. Secrecy:} The adversary at any time period knows no more than $\cal{P}$ shares (possible a different shares in each time period) learns nothing about the secret.  
 
     \end{theorem}
     \begin{proof}
The data    shares  $X^{(t)}_i$ in time period $t$ are stored in party $i$, s.t.,   $\sum_{i=1}^{\cal{P}}X^{(t)}_1=x$.
Each party $i$ generates shares $\delta^{(t)}_{i,1},\cdots, \delta^{(t)}_{i,{\cal{P}}}$ which satisfies  $\delta^{(t)}_{i,1}+\cdots+\delta^{(t)}_{i,{\cal{P}}} = 0 \mod N$. Thus, the new shares denote $X^{(t+1)}_i = X^{(t)}_i + \delta^{(t)}_{1,i} + \cdots + \delta^{(t)}_{{\cal{P}},i}$, and satisfy $\sum_{i=1}^{\cal{P}} X_i^{(t+1)} = \sum_{i=1}^{\cal{P}} X_i^{(t+1)} + \sum_{i=1}^{\cal{P}}\sum_{i=1}^{\cal{P}}\delta^{(t)}_{i,j} = x$ which   the robustness property hold.

To guarantee the  secrecy property, the data shares in time period $t$  can achieve the  information theoretic secure according to the theorem    \ref{Theo:1}.  Even adversary can get ${\cal{P}} -1$ shares in each time period $t$ ($t \leq t^*$), the adversary can compute $x^{(t)} = x - X^{(t)}_{{\cal{P}}_t} = \sum_{i=1, i \neq {\cal{P}}_t}^{\cal{P}} X_i^{(t)}$, where $X^{(t)}_{{\cal{P}}_t}$ is the non-compromised share  in  time period $t$. The adversary $\cal{A}^*$ still cannot get any information from $x^{(1)},\cdots,x^{(t_*)}$ as $\delta^{(1)}_{{\cal{P}}_1,{\cal{P}}_1},\cdots,\delta^{(t_*)}_{{\cal{P}}_{t_*},{\cal{P}}_{t_*}}$ are independently  and randomly generated and cannot be compromised by the adversary. Thus, the  secrecy property holds.
       \end{proof}

  \subsubsection{The Security of PCDD}

 The security of our PCDD is given by the following theorem.
  \begin{theorem}
  \label{Theo:3}
  The  PCDD scheme described in Section \ref{sec:AHES} is semantically secure, based on the assumed intractability of the DDH assumption over $ {\mathbb{Z}}_{N^2}^*$. \end{theorem}
  \begin{proof}
The   security of   PCDD   has been proven to be  semantically  secure under the  DDH assumption over $ {\mathbb{Z}}_{N^2}^*$ in the standard model  \cite{DBLP:conf/asiacrypt/BressonCP03}.  
  \end{proof}
 
   \subsection{Security of  TPU-based Basic Operation}
\begin{theorem}
\label{SP:RTG}
The \texttt{RTG} can securely  generate random shares against adversary who can compromise at most ${\cal{P}}-1$ TPUs, assuming the semantic security of the
PCDD cryptosystem.
\end{theorem}
\begin{proof}
For each TPU $i \ (0 \leq i<{\cal{P}})$, only the PCDD encryption $[\![\mathfrak{r}^{(1)}]\!],\cdots,[\![\mathfrak{r}^{(\ell)}]\!]$ are sent to TPU $i+1$.  After that, PCDD encryption  $[\![r]\!]$ is sent from TPU  $i+1$ to $i$. 
 According to   semantically secure of the PCDD  ({theorem}
  \ref{Theo:3}), the TPU $i+1$ cannot get any information  from the ciphertext sent from TPU $i$.
 Even the adversary can compromise at most ${\cal{P}}-1$ TPUs and get the shares $\mathfrak{r}_i^{(1)},\cdots, \mathfrak{r}_i^{(\ell)}, r_i$, (s)he cannot get the secret $\mathfrak{r}^{(1)},\cdots, \mathfrak{r}^{(\ell)}, r$ due to $\mathfrak{r}^{(1)}_{{\cal{P}}},\cdots, \mathfrak{r}^{(\ell)}_{{\cal{P}}}, r_{{\cal{P}}}$ are unknown to adversary according to the security of Theorem \ref{Theo:1}.
\end{proof}

  The security proof of the secure share domain transformation in section, secure binary shares operation in section, secure integer computation, and secure FPN computation are similar to the proof of theorem \ref{SP:RTG}. The security of above operations are based on the semantic security of the
PCDD cryptosystem. Next, we will show that \texttt{AHP} and \texttt{PIR} can achieve its corresponding functionality. 
  
  \begin{theorem}
\label{SP:AHP}
The \texttt{AHP} can securely achieve the access pattern hidden for the function input  under the semantic security of the
PCDD cryptosystem.
\end{theorem}
\begin{proof}
 In the select share phase,  all  $a_{1},\cdots, a_{H}$ are selected and encrypted  by RU, and are sent to TPUs for processing.
 It is impossible for the adversary to know the plaintext of the ciphertext due to the semantic security of  
PCDD. Also, the shares are dynamically update by computing  
$ b^*_i \leftarrow  b_i + \delta_{1,i} + \delta_{2,i}+ \cdots +\delta_{{\cal{P}},i}  \mod N$. As $\delta_{j,i}$ is randomly generated by TPU $i$ and is sent from TPU $j$ to TPU $i$, it is hard for the adversary to recover $b_i$ even adversary compromise the other ${\cal{P}}-1$ TPUs due to the secrecy of Theorem \ref{Theo:2}. Thus,  it is still impossible for the adversary to trace the original shares with the update shares which can achieve the  access pattern hidden.
\end{proof}

  \begin{theorem}
\label{SP:PIR}
The \texttt{PIR} can securely achieve the private information retrieve  under the semantic security of the
PCDD cryptosystem.
\end{theorem}
\begin{proof}
 In  \texttt{PIR}, all  $a_{1},\cdots, a_{H}$ are selected and encrypted  by RU, and sent to TPUs for processing.  After that, $[\![b^*]\!]$
 is transmitted among TPUs. As all the computations in the  \texttt{PIR}  are executed in the ciphertext domain,
 it is impossible for the adversary to know the plaintext of the ciphertext due to the semantic security of  
PCDD, which can achieve the   private information retrieve.
\end{proof}
  
  \subsection{Security of LightCom}
 \begin{theorem}
The LightCom is  secure against side-channel  attack   if $t_c+t_p+t_d< {\cal{P}} \cdot t_{a}$, where $t_c, t_p$  and $t_d$  are the runtime of secure computation $\texttt{GenCpt}$, private key update, and data share update, respectively; $t_{a}$ is the   runtime for attacker successfully compromising the  TPU enclave; ${\cal{P}}$ is the number of TPUs in the system. 
\end{theorem}
\begin{proof}
In the data upload phase, 
RU's data are randomly separated and uploaded to TPUs via secure channel. According to theorem \ref{Theo:1}, no useful information about the RU's data are leaked to the adversary with compromising  ${\cal{P}}-1$   TPUs encalves.
For the  long-term storage, the data shares are securely sealed  in the UnS with PCDD crypto-system. With the theorem \ref{Theo:3},  we can find the encrypted data shares are semantically secure stored in the UnS.

In the secure online computing phase, all the ciphertext are securely  load to the TPUs with \texttt{UnSeal}. Then, all the TPUs   jointly achieves the secure computation with the  $\texttt{GenCpt}$. During the computing phase, the system attacker can launch the following three types of attacks: \textit{1)  compromise the TPU enclave}: adversary can compromise a TPU enclave to get current data shares and private key shares
with the time $t_a$; \textit{2) stores the old private key shares}: the adversary  tries to  recover the RU's private key with current   and  old private key shares.
\textit{3) stores the old data shares and try to  recover the RU's original data}: the adversary  tries to  recover the RU's data  with current   and  old data shares.  To prevent first type of attack, RU's data are separated  and  distributed among $\cal{P}$ TPUs. Unless adversary can compromise all the TPU enclaves at the same time, $\cal{A}$ can get nothing useful information from compromised shares according to theorem \ref{Theo:1}.  Thanks to the secrecy property of proactive additive secret sharing scheme in Theorem \ref{Theo:2}, it is impossible for the adversary to recover the private key and RU's data by getting ${\cal{P}}-1$ TPUs   at each time period. As the TPU enclaves are dynamically     release after the computation, the attacker needs to restart to compromise the TPU enclaves after the enclaves are built for secure computation.

Thus, the adversary fails to attack our LightCom system   if the data shares are successfully seals in the UnS and  all the TPU enclaves are released before the adversary compromises all the   enclaves in the secure computation phase.  In this case, the LightCom is secure  against adversary side-channel attack   if $t_c+t_p+t_d< {\cal{P}} \cdot t_{a}$.
\end{proof}

\section{Evaluations}
In this section, we evaluate the performance of   LightCom. 
\label{sec:performance}
\subsection{Experiment Analysis}
For evaluating the performance of the LightCom, we build  the framework with C code under the Intel$^\circledR$ Software Guard Extensions (SGX) environment as a special case of TPU, and the experiments are performed on a personal computer (PC) with 3.6 GHz single-core processor and 1 GB  RAM memory (single-thread program are used) on virtual machine with Linux operation system. 
  To test the efficiency of our LightCom, there are two types of metrics are considered, called \textit{runtime}  and \textit{security level} (associate with PCDD parameter $N$). The runtime refers to the secure  outsourced computation   executing duration on  server or user's side   in our testbed.  The security level is an indication of the security strength of a cryptographic primitive. Moreover, we use SHA-256 as the hash function $H(\cdot)$
in LightCom. As the communication latency among CPUs is very low (use Intel$^\circledR$   UltraPath Interconnect (UPI) with 10.4 GT/s transfer speed and theoretical bandwidth is 20.8 GB/s)\footnote{https://www.microway.com/knowledge-center-articles/performance-characteristics-of-common-transports-buses/}, we do not consider the communication overhead as a performance metric   in our LightCom.

\subsubsection{Basic Crypto and System Primitive}

       \renewcommand\arraystretch{1}
\begin{table*}[htbp]
\caption{Protocol Performance: A Comparative Summary ($\ell = 32, H=8, {\cal{P}} =3, 
$ 100-time for average)} \label{Table:Performance} \centering
\begin{tabular}{Ic||c|c|c|c|c|c|c||c|c|c|c|c|c|cI}
  \thickhline
   & \multicolumn{7}{c||}{\textbf{Online Computation Cost (Millisecond)}} & \multicolumn{7}{cI}{\textbf{Offline Computation Cost (Second)}}\\
  \hline
  ${N}$ & \textbf{512} & \textbf{768} & \textbf{\textcolor[rgb]{0,0,1}{1024}} & \textbf{1280} & \textbf{1536} & \textbf{1792} & \textbf{2048} & \textbf{512} & \textbf{768} & \textbf{\textcolor[rgb]{0,0,1}{1024}} & \textbf{1280} & \textbf{1536} & \textbf{1792} & \textbf{2048}\\
  \thickhline
  \textbf{\texttt{RTG}} & 16.66 &  53.61 & \textcolor[rgb]{0,0,1}{117.0}  & 225.92 &  369.92  &  675.15  & 933.13
  & - & - & \textcolor[rgb]{0,0,1}{-} & - & - & - & - \\
  \hline
  \textbf{\texttt{B2I}} & 2.2 &  6.29 & \textcolor[rgb]{0,0,1}{13.92}  & 27.15 &  47.29  &  85.8  & 116.7
  & - & - & \textcolor[rgb]{0,0,1}{-} & - & - &- & - \\
  \hline
  \textbf{\texttt{I2B}} & 2.72 &  8.19 & \textcolor[rgb]{0,0,1}{16.38}  & 31.21 &  51.13  &  70.25  & 102.0
  & - & - & \textcolor[rgb]{0,0,1}{-} & - & - & - & - \\
  \hline
  \textbf{\texttt{SBM} } & 0.001	& 0.001 & 	\textcolor[rgb]{0,0,1}{0.001}& 	 0.002	&  0.002	& 0.002     & 0.002
  & 0.004	& 0.014	&\textcolor[rgb]{0,0,1}{0.024}	& 0.047	& 0.076	& 0.139	& 0.192\\
  \hline
  \textbf{\texttt{BAdd}} & 0.053	& 0.053	& \textcolor[rgb]{0,0,1}{0.054}	& 0.054	& 0.057	&  0.059	& 0.084
  & 0.222	& 0.915	& \textcolor[rgb]{0,0,1}{1.569}	& 3.024	& 4.868 	& 8.897	& 12.346 \\
  \hline
  \textbf{\texttt{BExt}} & 1.36	& 4.24	&\textcolor[rgb]{0,0,1}{8.54}	& 16.6	& 29.1	& 51.93	& 70.13
  & 0.268 & 1.077	&\textcolor[rgb]{0,0,1}{1.882}	& 3.634	& 6.016	& 10.842  	& 14.868\\
  \hline
  \textbf{\texttt{SM}} & 0.003	& 0.003 	&\textcolor[rgb]{0,0,1}{0.005}	& 0.007	& 0.008 	& 0.012	& 0.013
  & 0.009 	& 0.031	&\textcolor[rgb]{0,0,1}{0.066}	&  0.128	& 0.228	& 0.387	& 0.530\\
  \hline
  \textbf{\texttt{SMM}} & 0.140	&  0.215	&\textcolor[rgb]{0,0,1}{0.248}	& 0.356 & 0.457	& 0.46	& 0.614
  & 0.305	& 1.006	&\textcolor[rgb]{0,0,1}{2.114}	& 4.108	& 7.313	& 12.396	& 16.965\\
  \hline
    $\textbf{\texttt{SEP}}_2 $ & 0.001	& 0.001 	&\textcolor[rgb]{0,0,1}{0.001}	& 0.001	&  0.002	& 0.002 	& 0.003
  & 0.004	& 0.013	&\textcolor[rgb]{0,0,1}{0.028}	&  0.056	& 0.103	& 0.171	& 0.232\\
  \hline
  \textbf{\texttt{SEP}} & 2.73	& 5.95	&\textcolor[rgb]{0,0,1}{10.79}	& 19.44	& 34.84	&	53.8 & 74.08
  & 9.867	&  32.739	&\textcolor[rgb]{0,0,1}{68.414}	& 132.96	& 236.28	& 400.87	& 548.76\\
  \hline
  \textbf{\texttt{SC}} & 1.32	& 4.11	&\textcolor[rgb]{0,0,1}{8.77}	& 16.73	& 29.68	& 49.85	& 68.39
  & 	 0.267  &	1.077  &\textcolor[rgb]{0,0,1}{1.882}	&	3.634  &	6.017  &	10.743 & 14.869\\
  \hline
      \textbf{\texttt{SEQ}} & 2.59	& 8.06	&\textcolor[rgb]{0,0,1}{16.96}	& 32.84 	& 54.17	& 98.8	& 139.1
  & 0.535	& 2.155	&\textcolor[rgb]{0,0,1}{3.764}	& 	7.269 & 	12.034 &  21.486	& 29.378\\
  \hline
   $\textbf{\texttt{Min}}_2$ & 3.37	&  9.26	&\textcolor[rgb]{0,0,1}{20.34}	& 38.14	& 66.58	&	97.9 & 145.4
  & 0.286	& 1.141	&\textcolor[rgb]{0,0,1}{2.015}	& 3.891	& 6.474	& 11.517	& 15.929\\
  \hline
 $\textbf{\texttt{Min}}_H$ & 3.75	& 68.89	&\textcolor[rgb]{0,0,1}{150.85}	& 294.49	& 510.82	&  869.74	& 1453.6
  & 2.007	& 7.983	& \textcolor[rgb]{0,0,1}{14.101}	& 27.24	& 45.318	& 80.624	& 111.51\\
  \hline
  \textbf{\texttt{APH}} & 4.955	& 15.12 	&\textcolor[rgb]{0,0,1}{33.29}	& 63.58	& 117.79	& 191.25	& 228.5
  & - 	& -	&\textcolor[rgb]{0,0,1}{-}	& -	&-	& -	&-\\
  \hline
  \textbf{\texttt{PIR}} & 0.926	& 2.492	&\textcolor[rgb]{0,0,1}{5.057}	&9.28 	& 16.6	&  26.6	& 31.67
  & -	& -	&\textcolor[rgb]{0,0,1}{-}	& -	& -	& -	&-\\
  \hline
  \textbf{\texttt{UNI}} & 44.58	& 119.03	&\textcolor[rgb]{0,0,1}{247.27}	& 460.56	& 810.9	&	1181.5 	&1686.9
  & 81.016	& 270.14	&\textcolor[rgb]{0,0,1}{561.94}	& 1090.9	& 1937.4	& 3290.7	&4505.8\\
    \hline
  \thickhline
\end{tabular}
      \label{table111x}

\end{table*}

We first evaluate   the  performance of  our basic operation of cryptographic primitive (PCDD cryptosystem) and basic system operations (\texttt{Seal}, \texttt{UnSeal} and \texttt{SDD} protocol).  We first let $N$ be 1024 bits to achieve 80-bit security \cite{barker2007nist} to test the basic crypto primitive and basic protocol. For PCDD,   it takes 1.153 ms to encrypt a message (\texttt{Enc}), 1.171 ms for \texttt{Dec},  1.309 ms to run \texttt{PDec}, 5.209 $\mu$s to run \texttt{TDec}.   For the basic system operations,   it takes 1.317 ms for \texttt{Seal}, 1.523 ms for \texttt{UnSeal}, and 1.512 ms for \texttt{SDD} (${\cal{P}}=3$).  Moreover,     \texttt{Seal}, \texttt{UnSeal} and  \texttt{SDD} are affected by the PCDD parameter  $N$ and  the number of TPUs ${\cal{P}}$ (See Fig. \ref{fig:111} and Fig. \ref{fig:222} respectively).
From the Figs. \ref{fig:111} and   \ref{fig:222}, we can that the parameter  $N$ will affect greatly on the runtime and communication overhead of the protocols.

\begin{figure}[htbp]
    \centering
    \subfigure[Performance with $N$  (Vector Length =10)]
    {
        \includegraphics[width=0.22\textwidth]{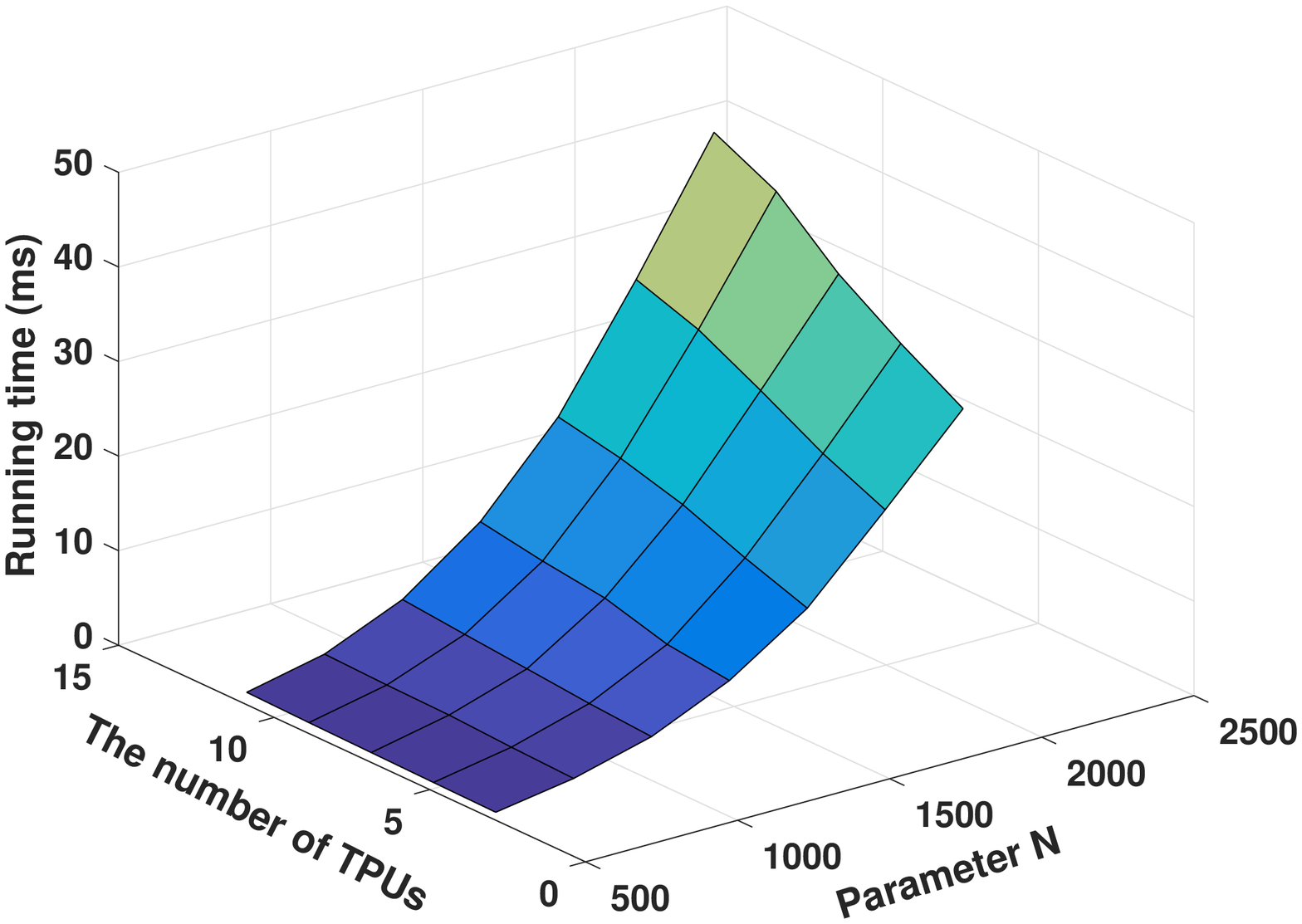}
        \label{fig:111}
    }
    \subfigure[Performance with Encrypted Vector Length  ($N=1024$)]
    {
        \includegraphics[width=0.22\textwidth]{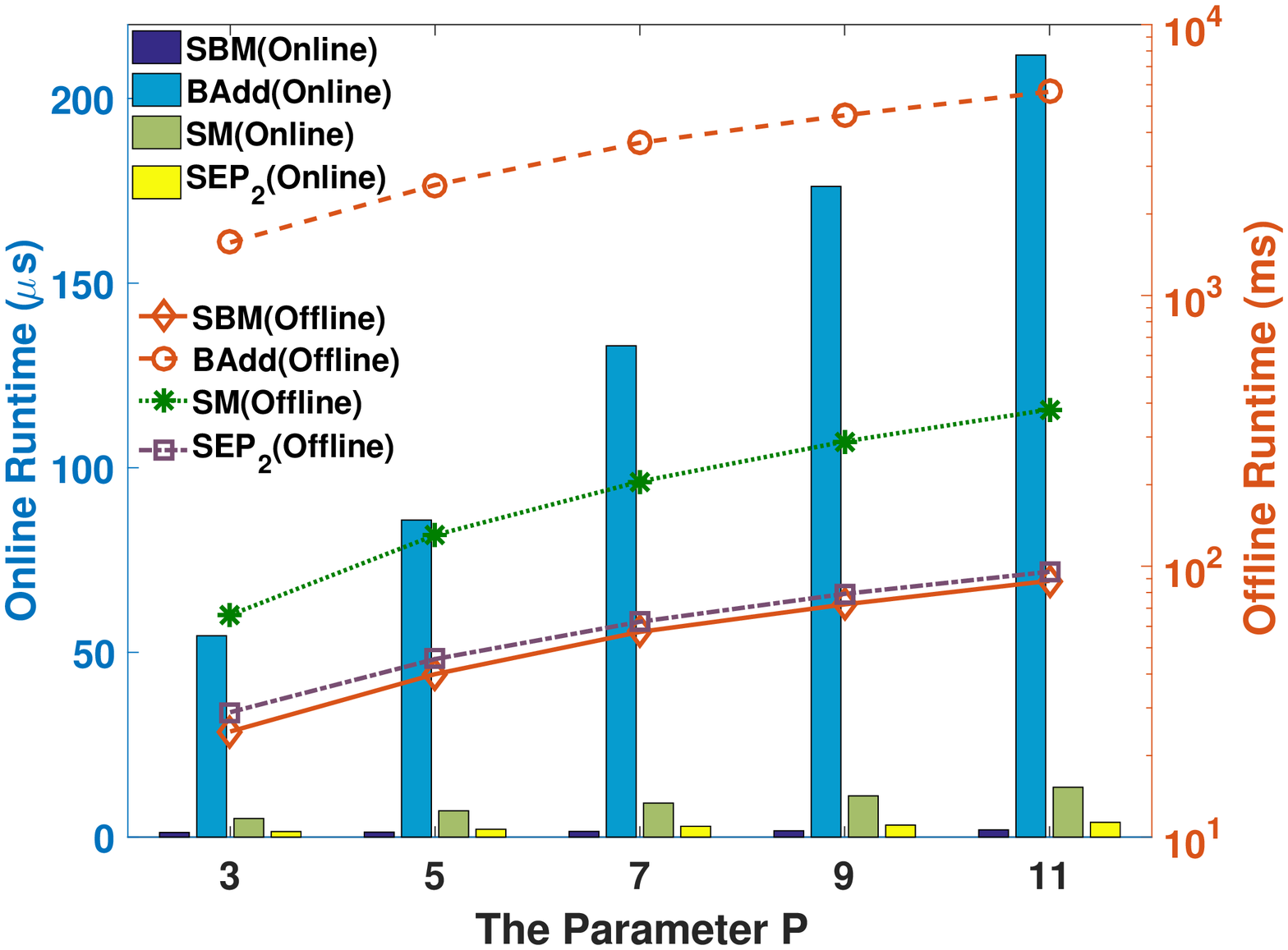}
        \label{fig:222}
    }
     \caption{Simulation results of Basic Protocols}
\end{figure}


\subsubsection{Performance of TPU-based  Integer Computation}
\label{SEC:TPU-I}
Generally, 
there are four factors that affect the performance of TPU-based  integer computation: 1) the number of TPUs $\cal{P}$; 2) the PCDD  parameter $N$;  3) the bit-length  of the integer $\ell$;  4) the number of encrypted data   $H$.  In Fig. \ref{fig:aaa}-\ref{fig:eee}, we can see that   the runtime    of all the protocols  increase  with $\cal{P}$.  It is because  more runtime are needed and more data in online phase and random numbers in offline phase are required to process with extra parties.
Also, 
we can see that   the runtime   of all the TPU-based integer computations  increase   with the bit-length of   $N$ 
from Table \ref{table111x}. It is because the running time of the basic operations (\texttt{Enc} and \texttt{Dec} algorithms of PCDD)   increases  when $N$ increases. Moreover, in Fig. \ref{fig:fff}-\ref{fig:kkk},  the  performance of \texttt{RTG}, \texttt{SMM}, \texttt{BAdd}, \texttt{BExt}, \texttt{SEP},    \texttt{SC},  \texttt{SEQ},  $\texttt{Min}_2$, $\texttt{Min}_H$, $\texttt{UNI}$     are associated with $\ell$. The computational  cost of above protocols are increased  with  $\ell$, as more computation resources are needed to process when $\ell$ increase. Finally, we can  see that performance of  \texttt{APH} and \texttt{PIR}  are increased with  $H$  in Fig.  \ref{fig:lll}. It is because more numbers of PCDD ciphertexts cost more energy with the homomorphic and 
module exponential operations.  



\subsubsection{Performance of TPU-based  FPN Computation}
For the basic TPU-based FPN computation, 
there are four factors that affects    performance of LightCom: 
 1) the number of TPUs; 2) the PCDD  parameter $N$;  3) the bit-length  of the integer $\ell$; 4) the number of encrypted data $H$.    The runtime trends of  FPN computation protocols  (e.g.  \texttt{FC},  \texttt{FEQ}, \texttt{FM},  \texttt{FMM},  $\texttt{FMin}_2$, $\texttt{FMin}_H$)   are similar to the trends of corresponding  secure integer computation (e.g.  \texttt{SC},  \texttt{SEQ}, \texttt{SM},  \texttt{SMM},  $\texttt{Min}_2$, $\texttt{Min}_H$), as the runtime of  FPN computation is equal to the runtime of corresponding  secure integer computation add the runtime of \texttt{UNI}.
 

\begin{figure*}[htbp]
    \centering
    \subfigure[Runtime with $\cal{P}$ ($\ell = 32$, $\|N\| =1024, H = 8$)]
    {
        \includegraphics[width=0.23\textwidth]{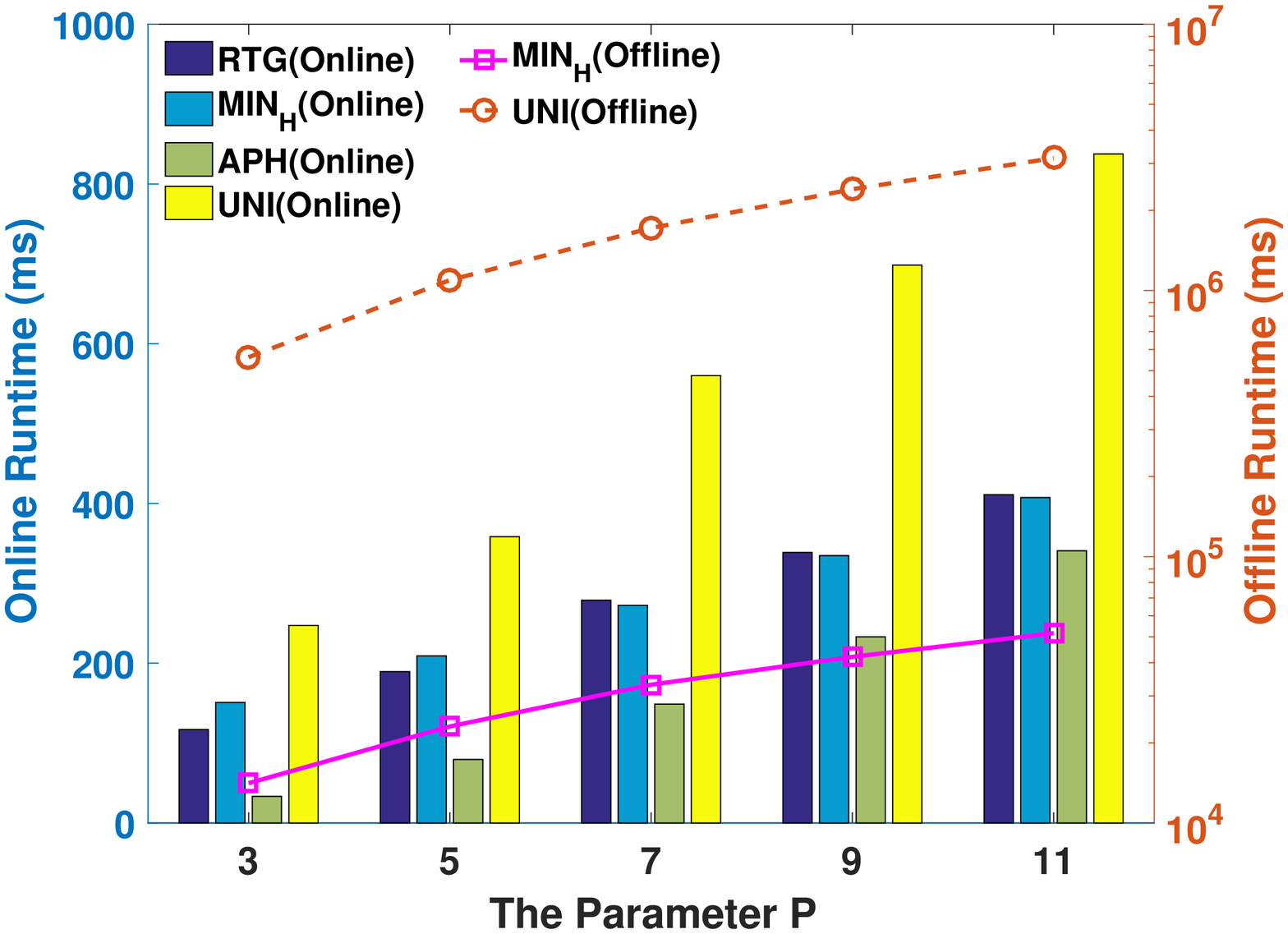}
        \label{fig:aaa}
    }
    \subfigure[Runtime with $\cal{P}$ ($\ell = 32$, $\|N\| =1024, H = 8$)]
    {
        \includegraphics[width=0.23\textwidth]{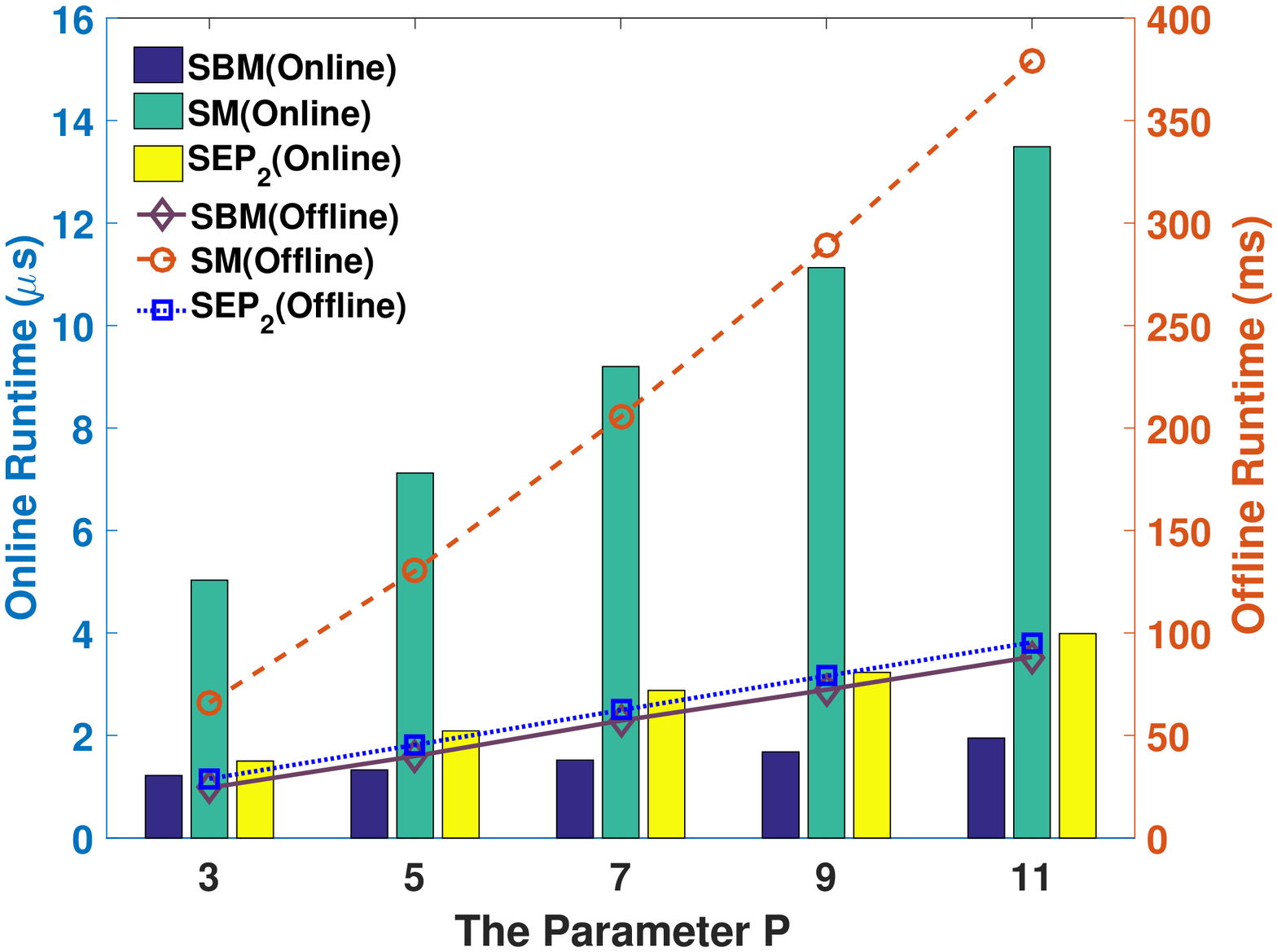}
        \label{fig:bbb}
    }
\subfigure[Runtime with $\cal{P}$ ($\ell = 32$, $\|N\| =1024, H = 8$)]
    {
        \includegraphics[width=0.23\textwidth]{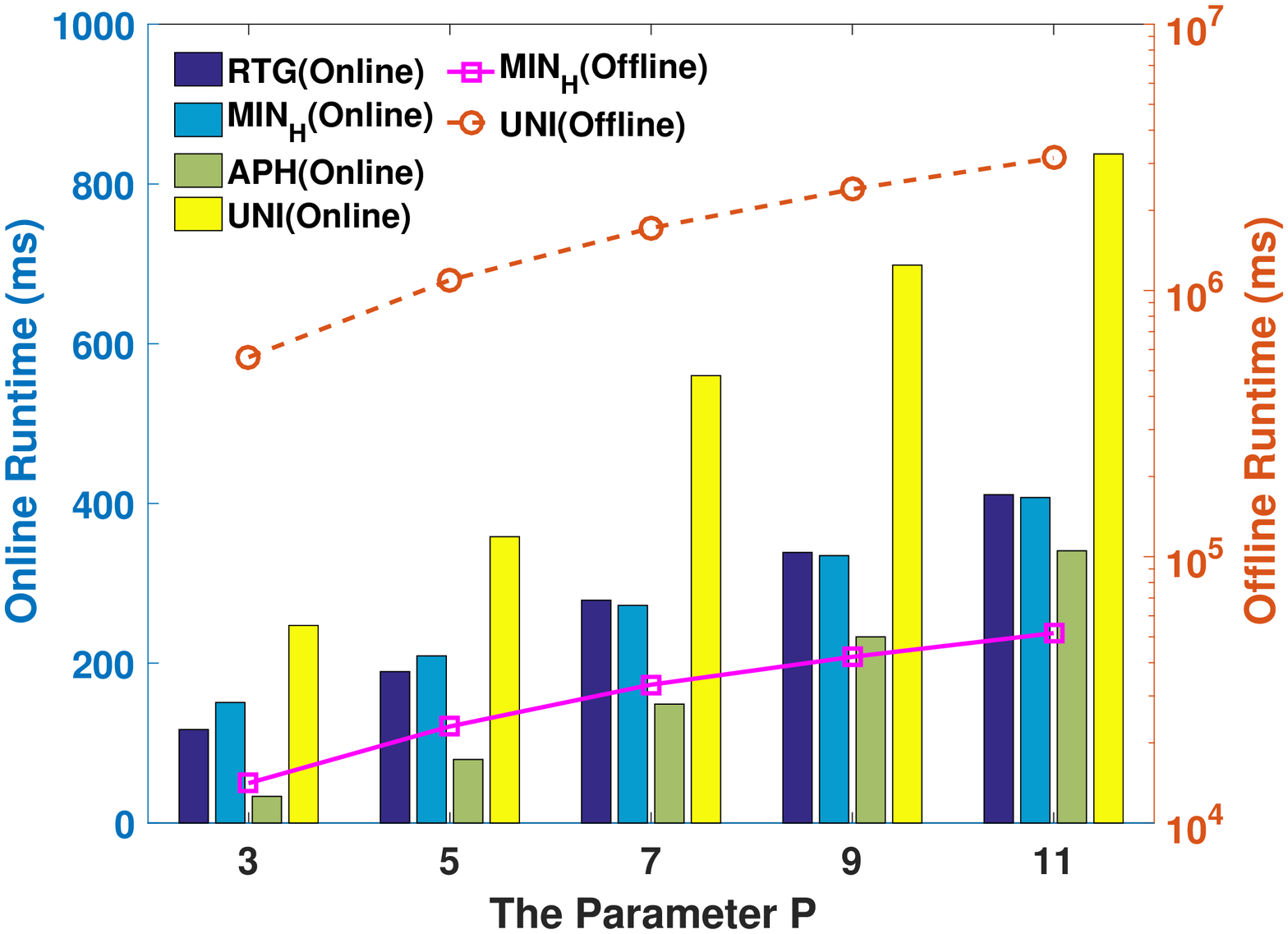}
        \label{fig:ccc}
    }
        \subfigure[Runtime with $\cal{P}$ ($\ell = 32$, $\|N\| =1024, H = 8$)]
    {
        \includegraphics[width=0.23\textwidth]{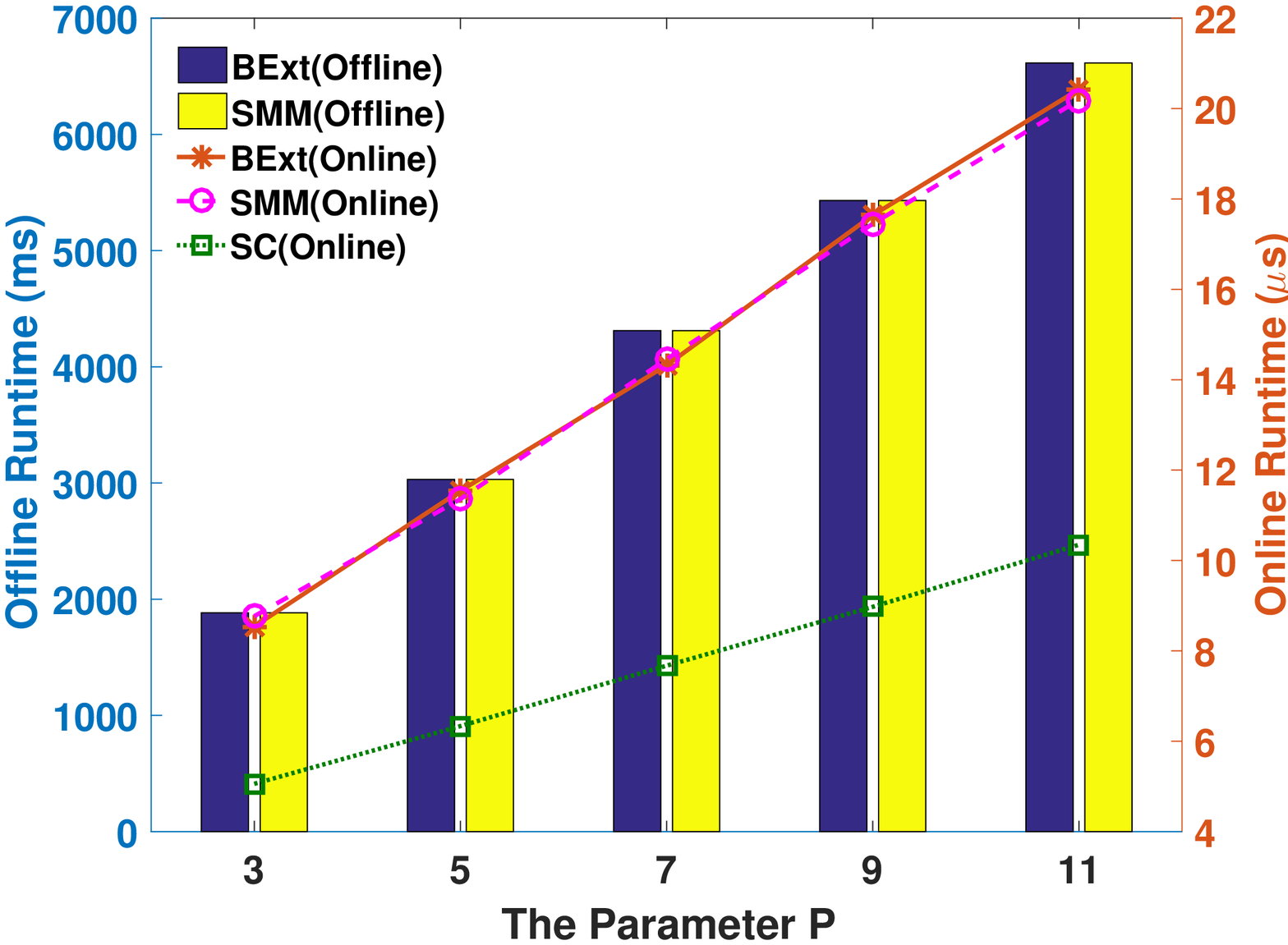}
        \label{fig:ddd}
    }
    \subfigure[Runtime with $\cal{P}$ ($\ell = 32$, $\|N\| =1024, H = 8$)]
    {
        \includegraphics[width=0.23\textwidth]{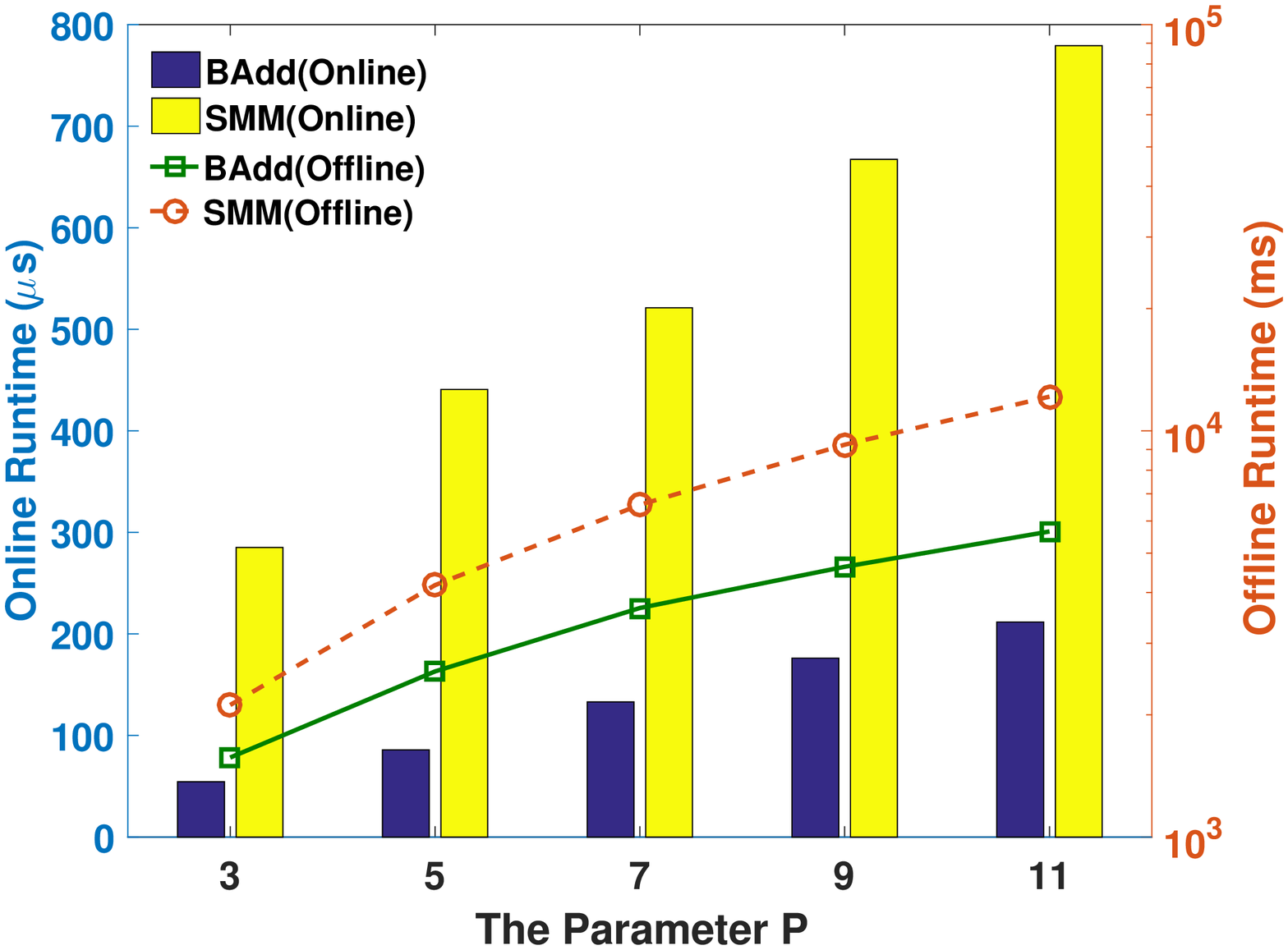}
        \label{fig:eee}
    }
    \subfigure[Runtime with $\ell$ (${\cal{P}} = 3$, $\|N\| =1024, H = 8$)]
    {
        \includegraphics[width=0.23\textwidth]{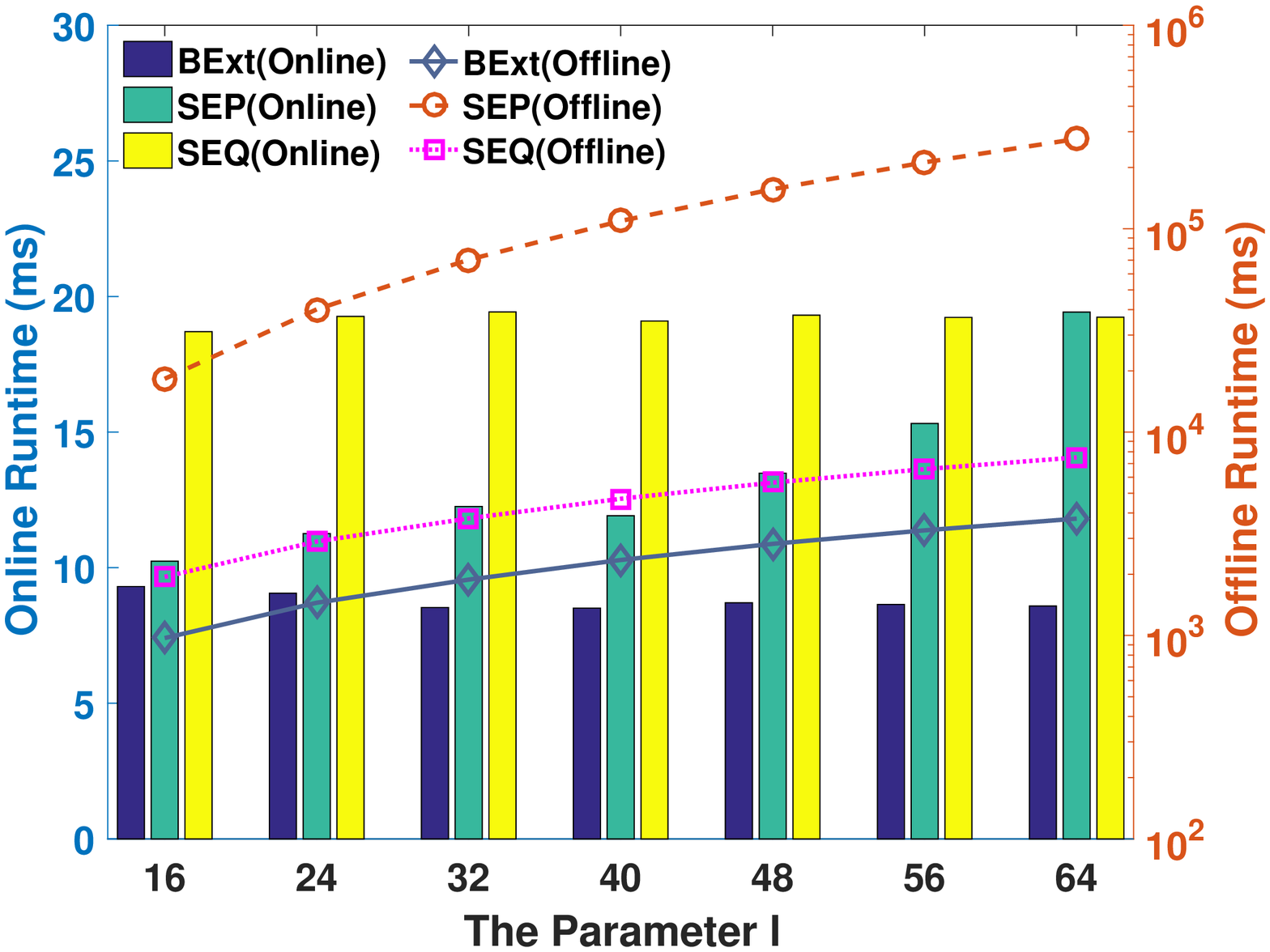}
        \label{fig:fff}
    }
    \subfigure[Runtime with $\ell$ (${\cal{P}} = 3$, $\|N\| =1024, H = 8$)]
    {
        \includegraphics[width=0.23\textwidth]{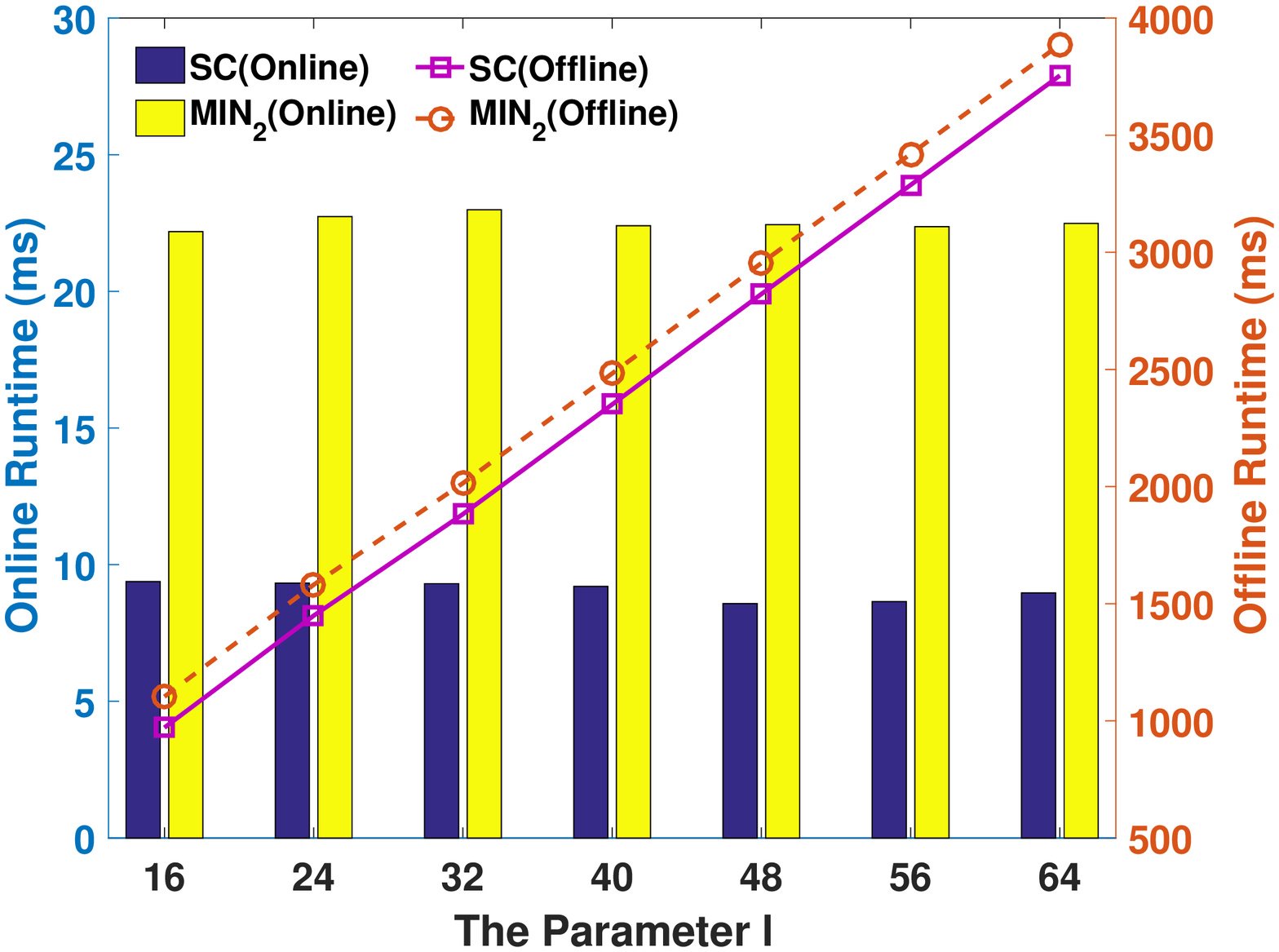}
        \label{fig:ggg}
    }
        \subfigure[Runtime with $\ell$ (${\cal{P}} = 3$, $\|N\| =1024, H = 8$)]
    {
        \includegraphics[width=0.23\textwidth]{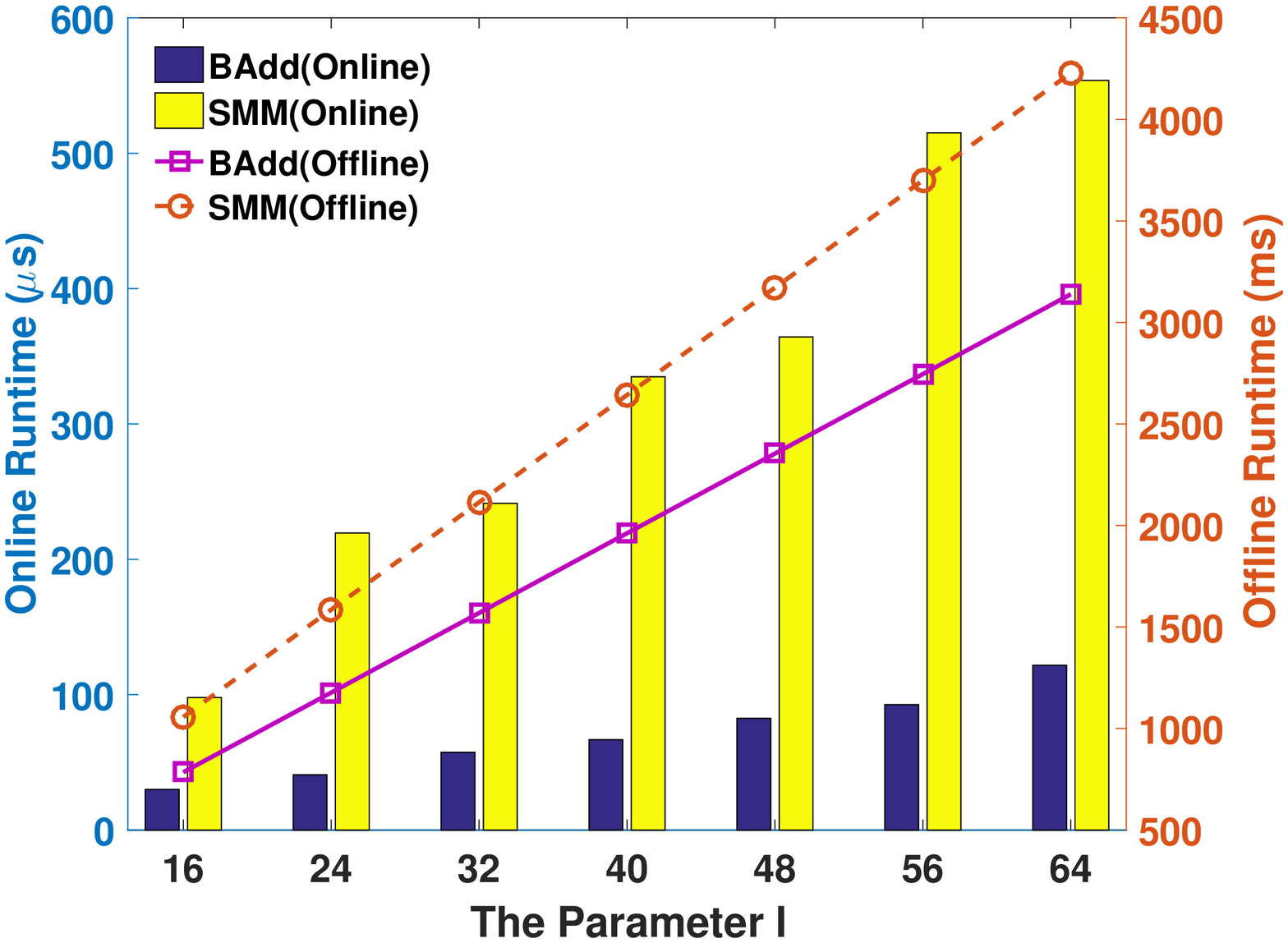}
        \label{fig:hhh}
    }
        \subfigure[Runtime with $\ell$ (${\cal{P}} = 3$, $\|N\| =1024, H = 8$)]
    {
        \includegraphics[width=0.23\textwidth]{FIGV9.eps}
        \label{fig:iii}
    }
    \subfigure[Runtime with $\ell$ (${\cal{P}} = 3$, $\|N\| =1024, H = 8$)]
    {
        \includegraphics[width=0.23\textwidth]{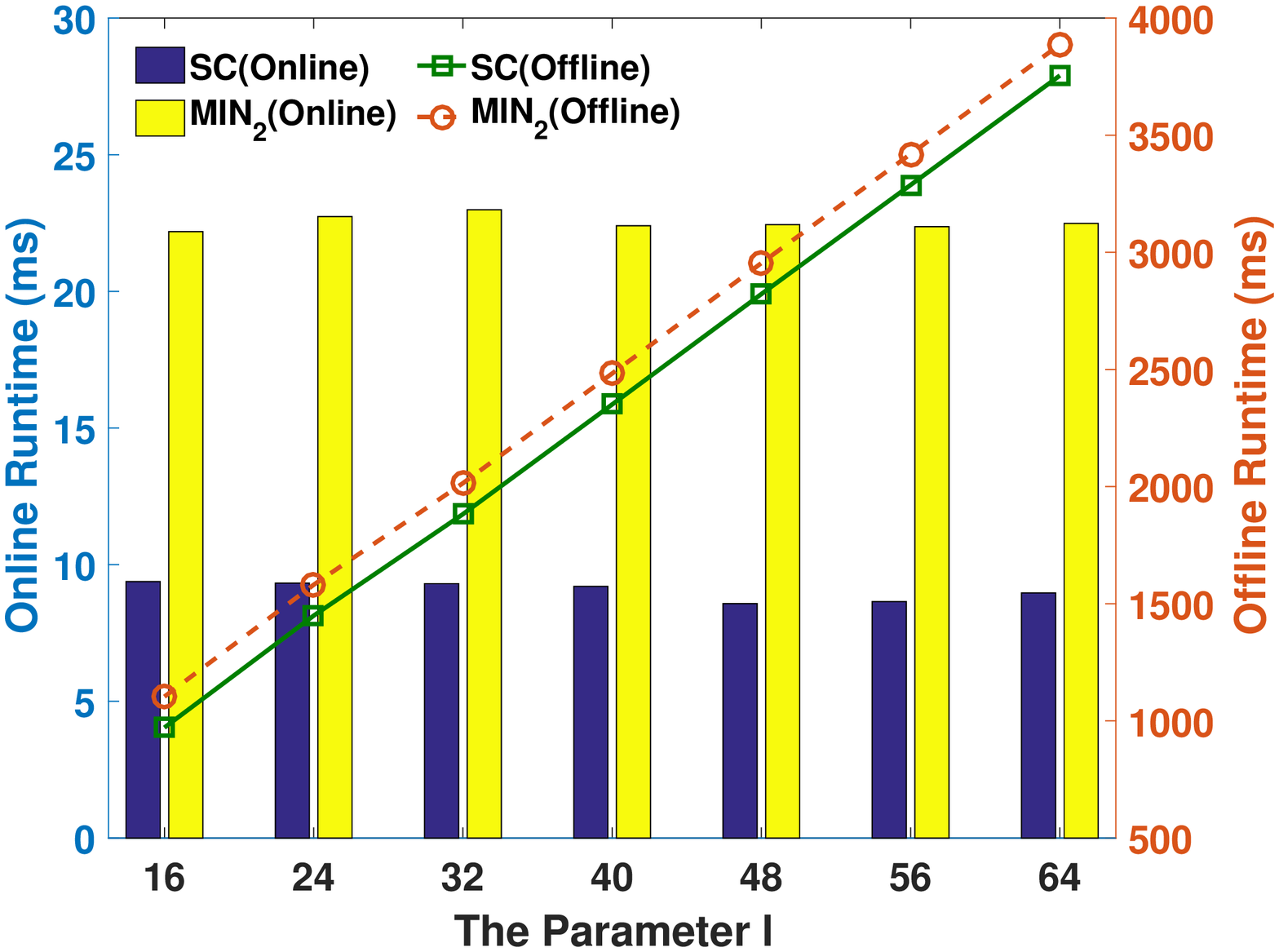}
        \label{fig:jjj}
    }
    \subfigure[Runtime with $\ell$ (${\cal{P}} = 3$, $\|N\| =1024, H = 8$)]
    {
        \includegraphics[width=0.23\textwidth]{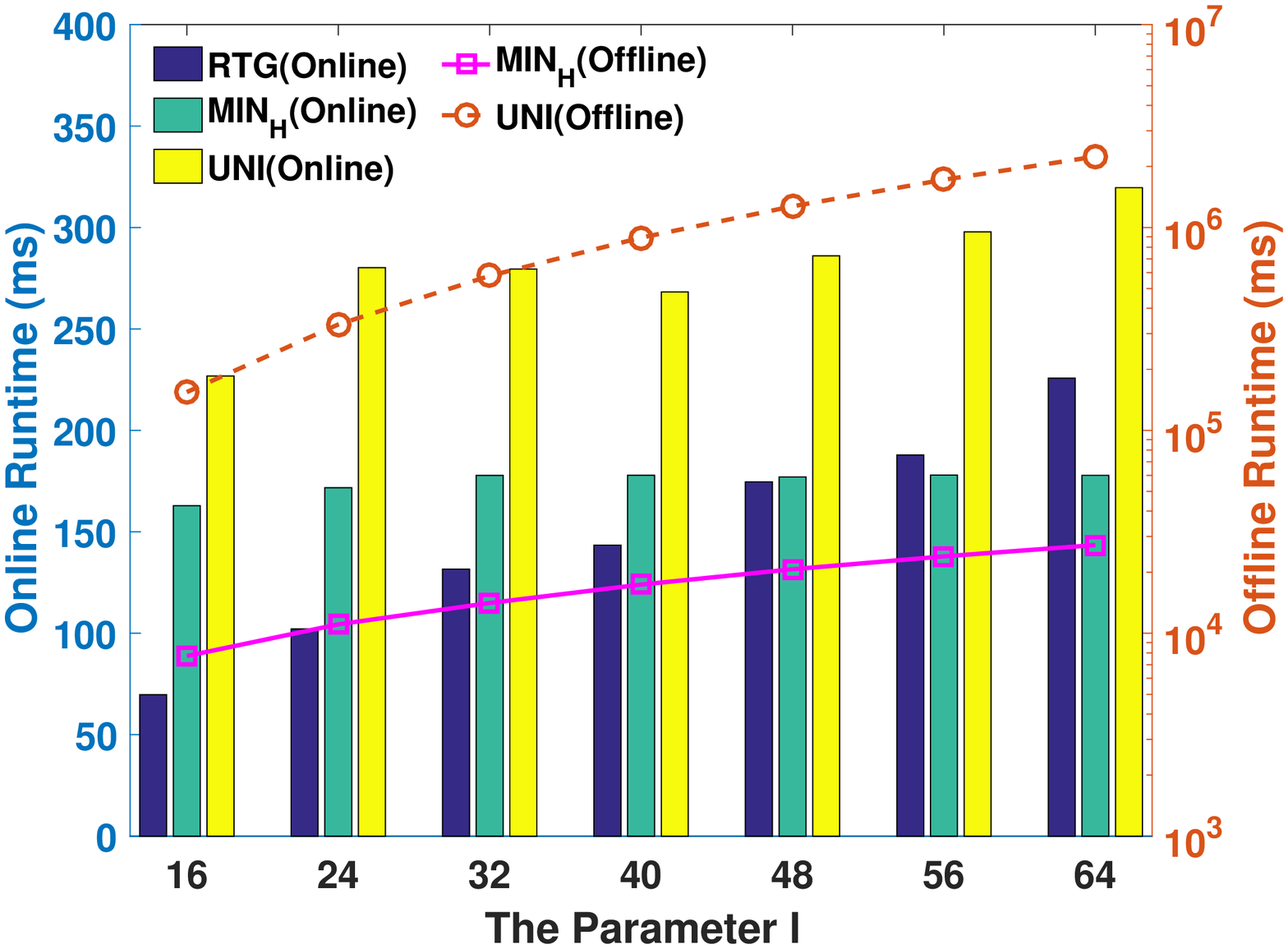}
        \label{fig:kkk}
    }
        \subfigure[Runtime with $H$ (${\cal{P}} = 3$, $\|N\| =1024, \ell = 32$)]
    {
        \includegraphics[width=0.23\textwidth]{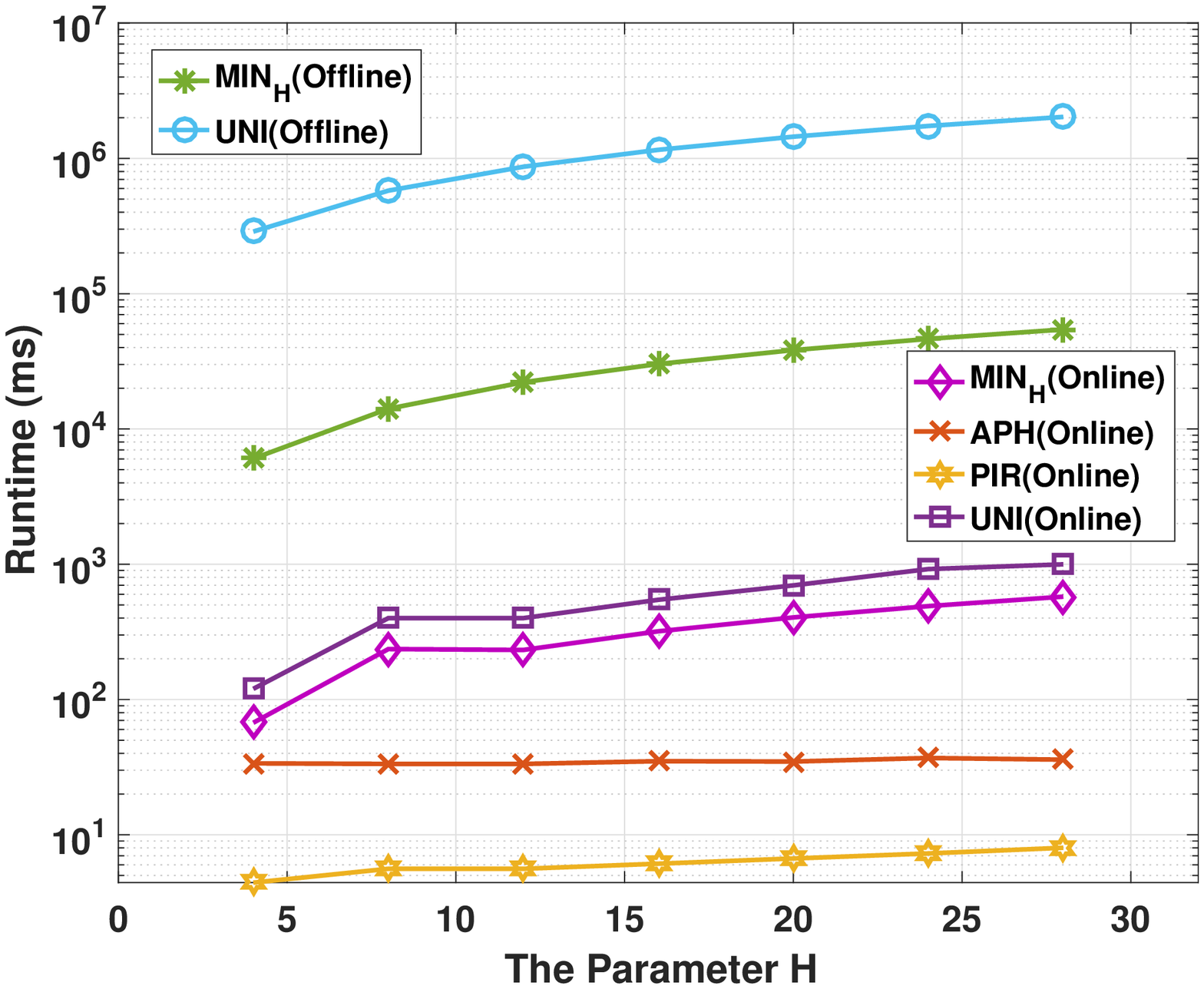}
        \label{fig:lll}
    }

      \caption{Simulation results of LightCom}
    \label{fig:DIM}
\end{figure*}

    \subsection{Theoretical  Analysis}

  Let  us  assume  that  one  regular  exponentiation  operation  with  an  exponent  of $\|N\|$ requires 1.5 $\|N\|$ multiplications \cite{knuth2014art}.  For PCDD, it takes $3\|N\|$ multiplications for \texttt{Enc}, $1.5\|N\|$ multiplications for \texttt{Dec}, $1.5\|N\|$ multiplications for \texttt{PDec}, $\cal{P}$ multiplications for \texttt{TDec}, $1.5\|N\|$ multiplications for \texttt{CR}. For the basic operation of LightCom, it takes $1.5 {\cal{P}} \|N\|$ multiplications to run $\texttt{SDD}$, $3\|N\|+t_{hash}$ multiplications for $\texttt{Seal}$, $1.5{\cal{P}}\|N\|+t_{hash}$ multiplications for $\texttt{UnSeal}$, ${\cal{O}} ((\ell+{\cal{P}})\|N\|)$ multiplications for $\texttt{RTG}$,  ${\cal{O}} ({\cal{P}}\|N\|)$ multiplications for $\texttt{B2I}$, $\texttt{I2B}$. For the integer and binary  protocol in LightCom, it takes ${\cal{O}} ({\cal{P}}\|N\|)$ multiplications for  offline phase of $\texttt{SBM}$ and $\texttt{SM}$, ${\cal{O}} (\ell {\cal{P}}\|N\|)$ multiplications for  offline phase of $\texttt{BAdd}$, $\texttt{BExt}$, $\texttt{SC}$, $\texttt{SEQ}$, $\texttt{Min}_2$,  ${\cal{O}} (\ell {\cal{P}}\|N\|)$ multiplications for both offline and online phase of $\texttt{SEP}$,   ${\cal{O}} (H {\cal{P}}\|N\|)$ multiplications for   offline     phase of $\texttt{APH}$ and $\texttt{PIR}$,  ${\cal{O}} (\lceil \log_2H \rceil \cdot \ell {\cal{P}}\|N\|)$ multiplications for   offline     phase of $\texttt{Min}_H$. 
  For the FPN computation in LightCom, it takes ${\cal{O}} (H \ell {\cal{P}}\|N\|)$ multiplications for  offline phase $\texttt{UNI}$ and $\texttt{FAdd}$,  ${\cal{O}} (\ell {\cal{P}}\|N\|)$ multiplications for offline phase $\texttt{FM}$,  $\texttt{FMM}$, $\texttt{FC}$,  $\texttt{FEQ}$,  $\texttt{FMin}_2$,   and ${\cal{O}} (\lceil \log_2H \rceil \cdot \ell {\cal{P}}\|N\|)$ multiplications for   offline     phase of $\texttt{FMin}_H$.       All the above protocols only need   ${\cal{O}}(1)$ multiplications in online phase, which is greatly fit for fast processing.

\section{Related Work}

\textit{Homomorphic Encryption}. Homomorphic encryption, allow third-party to do the computation on the ciphertext which reflected on the plaintext, is considered as the best solution to achieve the secure outsourced computation. The first construction of fully homomorphic encryption was proposed by Gentry in 2009 under the ideal lattices, which permits evaluation
of arbitrary circuits over the plaintext \cite{gentry2009fully}. Later, some of the new hard problems (such as Learning With Errors (LWE) \cite{brakerski2014efficient},  Ring-LWE  \cite{brakerski2014leveled}) are used to construct the FHE which can greatly reduce the storage overhead and increase the performance of the homomorphic operations \cite{chillotti2016faster, liu2018privacy}. 
However, the current FHE solutions and libraries  are still not practical enough for the real real-world scenarios \cite{doroz2015accelerating, liu2017privacy}. 
Somewhat homomorphic encryption \cite{damgaard2012multiparty, fan2012somewhat} can allow semi-honest  third-party to achieve the  arbitrary circuits with limited depth. The limited times of homomorphic operations are  restrict the   usage scope   of the    application.  Semi-homomorphic encryption (SHE) can only support additive \cite{paillier1999public} (or multiplicative \cite{elgamal1985public})  homomorphic operation. However, with the help of the extra semi-honest computation-aid server,  a new computation framework can be constructed to achieve commonly-used secure  rational number computation \cite{liu2018efficientTDSC}, secure multiple keys computation \cite{peter2013efficiently}, and floating-point number computation \cite{liu2016privacy}. The new framework can greatly balance the security and efficiency concerns, however, the extra server will still complex the system which brings more risk of information leakage.

\textit{Secret Sharing-based Computation}. The user's data in secret sharing-based (SS-based)  computation
 are separated into multiple shares with the secret sharing technique, and each shares are located in one server to guarantee the security.    Multiple parties can jointly together to securely achieve a computation without leaking the original data to the adversary.  
  Different from the heavyweight  homomorphic operation, the SS-based computation \cite{cramer2000general,chen2006algebraic,chida2018fast} can achieve the lightweight computation. Despite the theoretical construction, many real-word computation are constructed for practical usage, such as SS-based  set intersection \cite{dong2013private} and    top-$k$ computation \cite{burkhart2010fast}.  These basic computations can be used to  solve data security problem in data mining technique, such as deep learning \cite{huang2019lightweight}. 
  Emek{\c{c}}i  \textit{et al.} \cite{emekcci2007privacy} proposed a secure ID3 algorithm to construct a decision tree in a privacy-preserving manner. 
  Ma \textit{et al.} \cite{8668517} constructed a lightweight privacy-preserving    adaptive boosting (AdaBoost) for the face recognition. The new secure natural exponential and secure natural logarithm which can securely achieve the corresponding computation computation to balance accuracy and efficiency. Although many of the privacy-preserving  data mining techniques with secret sharing  are constructed \cite{ge2010privacy,gheid2016efficient},  the SS-based computation still need to build secure channel among these parties. Moreover, the   high communication rounds among the computation parties   still   become an obstacle for a large-scale application.
  
\textit{Intel$^\circledR$ Software Guard Extensions.}
Intel$^\circledR$ SGX is a kind of TEE which provides strong hardware-enforced confidentiality and integrity guarantees and protects  an application form the host OS, hypervisor, BIOS, and other software. 
Although an increasingly number of  real-world industry applications are securely executed in the untrusted remote platforms equipped with SGX, the SGX  still faces side-channel attack to expose the information during the computation.
G\"otzfried \textit{et al.} \cite{gotzfried2017cache} proposed a new attack called  root-level cache-timing attacks which can   obtain secret information
from an Intel$^\circledR$ SGX enclave. Lee \textit{et al.} \cite{lee2017inferring} gave a new side-channel attack cannled branch shadowing which reveals fine-grained control flows  in a SGX enclave.
Bulck  \textit{et al.} \cite{van2017telling}
constructed  two novel attack vectors that infer enclaved memory accesses. Chen  \textit{et al.} \cite{chen2018sgxpectre} presented  a new attack call SGXPECTRE   that can learn secrets
inside the enclave memory or its internal registers.  Currently, three types of  solutions are used to protect the side-channel attack: hardware method \cite{domnitser2012non, costan2016sanctum}, system method \cite{liu2016catalyst, zhou2016software}, and application method \cite{coppens2009practical, shih2017t}. 
These   methods can only  guarantee some dimension of  protection, and cannot be used for all-directional protection even against the unknown side-channel attack. 

   \begin{table*}[htbp]
      \begin{minipage}{\textwidth}
\caption{Comprehensive Comparison with the existing works}
\label{table1254}
\centering
 \renewcommand\arraystretch{0.75}
\begin{tabular}{cccccccccc}
 \hline
 Function/Algorithm  &  \cite{liu2018efficientTDSC} &  \cite{liu2016efficientx}   & \cite{liu2016privacy} & \cite{peter2013efficiently} &  \cite{samanthula2014k} & \cite{liu2018privacy}    & \cite{brakerski2014leveled} & \cite{dong2013private}    \\
\hline
Method& PHE & PHE & PHE & PHE &PHE &FHE &FHE & OT+SS\\
 User-side Non-interactive     &   \checkmark & \checkmark &  \checkmark & \checkmark & $\times$ & \checkmark &  \checkmark   & $\times$      \\ 
 Communication  Round (User)   & 1 & $1$ & $1$  & $1$  & ${\cal{O}}(1)$ & $1$  &  $1$  & ${\cal{O}}(n)$  \\
Against Side-channel Attack & \checkmark & \checkmark & \checkmark &   \checkmark  &    \checkmark  & \checkmark & \checkmark & \checkmark  \\
    Data Storage Server & One    & One & One & One & One &  One    &  One   &  One   \\
    Minimum  Number of  Servers  & Multiple & Two & Two & Two & Two & One& One& One  \\
   Function Type & Specific &  Specific  &  Specific  & Specific &  Specific  & Specific    & Linearly & Intersection     \\
     Multiple  Data Format &  \checkmark   & $\times$ & \checkmark& $\times$ & $\times$ & $\times$ & $\times$ & $\times$   \\
 Without Non-colluded Servers  & $\times$ & $\times$ &  $\times$ & $\times$ & $\times$ & \checkmark & \checkmark & \checkmark  \\
  Without TTP  & $\times$ & $\times$ &  $\times$ &  $\times$  & \checkmark & $\times$ & $\times$& \checkmark  \\
   Support Multiple Keys  & $\times$ & \checkmark &  $\times$ & \checkmark & $\times$ & \checkmark& $\times$  & $\times$ \\
   Server-Side Overhead  & Middle & Middle &  Middle & Middle & Middle & High & High & Middle  \\
 \hline
 \hline
 Function/Algorithm & \cite{burkhart2010fast} &  \cite{emekcci2007privacy}   & \cite{huang2019lightweight} &  \cite{8668517}   & \cite{shaon2017sgx} & \cite{kuccuk2016exploring}  & \cite{chandra2017securing} & Our       \\
\hline
Method& SS & SS & SS & SS &TEE &TEE &TEE & TEE+SS+PHE\\
 User-side Non-interactive     &   $\times$ &  $\times$ &  \checkmark & \checkmark & \checkmark & \checkmark &  \checkmark   & \checkmark     \\ 
 Communication  Round (User)   & ${\cal{O}}(kn^2)$ & ${\cal{O}}(n)$ & $1$  & $1$  & $1$ & $1$  &  $1$  & $1$  \\
Against Side-channel Attack & \checkmark & \checkmark & \checkmark &   \checkmark  &    $\times$  & $\times$ & \checkmark & \checkmark  \\
    Data Storage Server & Multiple    & Multiple & Two & Two & One &  One    &  One   &  One   \\
    Minimum  Number of  Servers  & Multiple & Two & Two & Two & One & One& One& One  \\
   Function Type & Top-$k$ &  Addition  &  Specific  & Adaboost &  Matrix  & Specific    & Specific & Generic \& Specific     \\
     Multiple  Data Format &   $\times$   & $\times$ & $\times$& $\times$ & $\times$ & $\times$ & $\times$ & \checkmark   \\
 Without Non-colluded Servers  & $\times$ & $\times$ &  $\times$ & $\times$ & $\times$ & \checkmark & \checkmark & \checkmark  \\
  Without TTP  & \checkmark & \checkmark &  $\times$ &  $\times$  & \checkmark & \checkmark & \checkmark & \checkmark  \\
   Support Multiple Keys  & \checkmark & \checkmark &  \checkmark & \checkmark & $\times$ &  $\times$  & $\times$  & \checkmark \\
   Server-Side Overhead  & Low  & Low &  Low & Low & Low & Low & Low & Low  \\ \hline
   \multicolumn{9}{c}{\textbf{Note:} In the table, `PHE' is short for  `Partially Homomorphic Encryption', `OT' is short for `Oblivious Transfer',}\\
       \multicolumn{9}{c}{`SS' is short for `Secret Sharing', TEE is short for 'Trusted Execution Environment'.}\\
\end{tabular}
\end{minipage}
\end{table*}
  
\section{Conclusion}
\label{sec:Conclusion}
In this paper, we   proposed LightCom, a framework for practical   privacy-preserving      outsourced  computation framework, which allowed a user to outsource encrypted data to a single cloud service provider for securely data storage  and   process.  We designed two types of   outsourced computation toolkits  which can securely guarantee the achieve secure integer computation and floating-point computation against side-channel attack. 
  The utility and performance  of our LightCom framework  was then demonstrated using simulations. Compared with the existing  secure outsourced computation framework, our LightCom takes 
  fast, scalable, and  secure outsourced  data processing  into account.

As a future research effort, we plan to apply our   LightCom in a specific applications, such as e-health cloud system. It allows us to refine the framework to handle more complex real-world computations.

 \section*{Acknowledgment}

 The work   is supported by the National Natural Science Foundation of China (Grant No.61702105, No.61872091).

\ifCLASSOPTIONcaptionsoff
  \newpage
\fi



%
\bibliographystyle{IEEEtran}
\bibliography{myref}

\begin{thebibliography}{10}
\providecommand{\url}[1]{#1}
\csname url@samestyle\endcsname
\providecommand{\newblock}{\relax}
\providecommand{\bibinfo}[2]{#2}
\providecommand{\BIBentrySTDinterwordspacing}{\spaceskip=0pt\relax}
\providecommand{\BIBentryALTinterwordstretchfactor}{4}
\providecommand{\BIBentryALTinterwordspacing}{\spaceskip=\fontdimen2\font plus
\BIBentryALTinterwordstretchfactor\fontdimen3\font minus
  \fontdimen4\font\relax}
\providecommand{\BIBforeignlanguage}[2]{{%
\expandafter\ifx\csname l@#1\endcsname\relax
\typeout{** WARNING: IEEEtran.bst: No hyphenation pattern has been}%
\typeout{** loaded for the language `#1'. Using the pattern for}%
\typeout{** the default language instead.}%
\else
\language=\csname l@#1\endcsname
\fi
#2}}
\providecommand{\BIBdecl}{\relax}
\BIBdecl

\bibitem{dimitrov2016medical}
D.~V. Dimitrov, ``Medical internet of things and big data in healthcare,''
  \emph{Healthcare informatics research}, vol.~22, no.~3, pp. 156--163, 2016.

\bibitem{naehrig2011can}
M.~Naehrig, K.~Lauter, and V.~Vaikuntanathan, ``Can homomorphic encryption be
  practical?'' in \emph{Proceedings of the 3rd ACM workshop on Cloud computing
  security workshop}.\hskip 1em plus 0.5em minus 0.4em\relax ACM, 2011, pp.
  113--124.

\bibitem{van2010fully}
M.~Van~Dijk, C.~Gentry, S.~Halevi, and V.~Vaikuntanathan, ``Fully homomorphic
  encryption over the integers,'' in \emph{Annual International Conference on
  the Theory and Applications of Cryptographic Techniques}.\hskip 1em plus
  0.5em minus 0.4em\relax Springer, 2010, pp. 24--43.

\bibitem{liu2018privacyxxx}
X.~Liu, R.~Deng, K.-K.~R. Choo, Y.~Yang, and H.~Pang, ``Privacy-preserving
  outsourced calculation toolkit in the cloud,'' \emph{IEEE Transactions on
  Dependable and Secure Computing}, 2018.

\bibitem{bendlin2011semi}
R.~Bendlin, I.~Damg{\aa}rd, C.~Orlandi, and S.~Zakarias, ``Semi-homomorphic
  encryption and multiparty computation,'' in \emph{Annual International
  Conference on the Theory and Applications of Cryptographic Techniques}.\hskip
  1em plus 0.5em minus 0.4em\relax Springer, 2011, pp. 169--188.

\bibitem{farokhi2016secure}
F.~Farokhi, I.~Shames, and N.~Batterham, ``Secure and private cloud-based
  control using semi-homomorphic encryption,'' \emph{IFAC-PapersOnLine},
  vol.~49, no.~22, pp. 163--168, 2016.

\bibitem{liu2016efficientx}
X.~Liu, R.~H. Deng, K.-K.~R. Choo, and J.~Weng, ``An efficient
  privacy-preserving outsourced calculation toolkit with multiple keys,''
  \emph{IEEE Transactions on Information Forensics and Security}, vol.~11,
  no.~11, pp. 2401--2414, 2016.

\bibitem{unicode1997unicode}
U.~Consortium \emph{et~al.}, \emph{The Unicode Standard, Version 2.0}.\hskip
  1em plus 0.5em minus 0.4em\relax Addison-Wesley Longman Publishing Co., Inc.,
  1997.

\bibitem{DBLP:conf/asiacrypt/BressonCP03}
E.~Bresson, D.~Catalano, and D.~Pointcheval, ``A simple public-key cryptosystem
  with a double trapdoor decryption mechanism and its applications,'' in
  \emph{Advances in Cryptology - {ASIACRYPT} 2003, 9th International Conference
  on the Theory and Application of Cryptology and Information Security, Taipei,
  Taiwan, November 30 - December 4, 2003, Proceedings}, 2003, pp. 37--54.

\bibitem{barker2007nist}
E.~Barker, W.~Barker, W.~Burr, W.~Polk, and M.~Smid, ``{NIST} special
  publication 800-57,'' \emph{NIST Special Publication}, vol. 800, no.~57, pp.
  1--142, 2007.

\bibitem{knuth2014art}
D.~E. Knuth, \emph{Art of computer programming, volume 2: Seminumerical
  algorithms}.\hskip 1em plus 0.5em minus 0.4em\relax Addison-Wesley
  Professional, 2014.

\bibitem{gentry2009fully}
C.~Gentry \emph{et~al.}, ``Fully homomorphic encryption using ideal lattices.''
  in \emph{Stoc}, vol.~9, no. 2009, 2009, pp. 169--178.

\bibitem{brakerski2014efficient}
Z.~Brakerski and V.~Vaikuntanathan, ``Efficient fully homomorphic encryption
  from (standard) lwe,'' \emph{SIAM Journal on Computing}, vol.~43, no.~2, pp.
  831--871, 2014.

\bibitem{brakerski2014leveled}
Z.~Brakerski, C.~Gentry, and V.~Vaikuntanathan, ``(leveled) fully homomorphic
  encryption without bootstrapping,'' \emph{ACM Transactions on Computation
  Theory (TOCT)}, vol.~6, no.~3, p.~13, 2014.

\bibitem{chillotti2016faster}
I.~Chillotti, N.~Gama, M.~Georgieva, and M.~Izabachene, ``Faster fully
  homomorphic encryption: Bootstrapping in less than 0.1 seconds,'' in
  \emph{International Conference on the Theory and Application of Cryptology
  and Information Security}.\hskip 1em plus 0.5em minus 0.4em\relax Springer,
  2016, pp. 3--33.

\bibitem{liu2018privacy}
X.~Liu, R.~Deng, K.-K.~R. Choo, Y.~Yang, and H.~Pang, ``Privacy-preserving
  outsourced calculation toolkit in the cloud,'' \emph{IEEE Transactions on
  Dependable and Secure Computing}, 2018.

\bibitem{doroz2015accelerating}
Y.~Dor{\"o}z, E.~{\"O}zt{\"u}rk, and B.~Sunar, ``Accelerating fully homomorphic
  encryption in hardware,'' \emph{IEEE Transactions on Computers}, vol.~64,
  no.~6, pp. 1509--1521, 2015.

\bibitem{liu2017privacy}
X.~Liu, R.~Deng, K.-K.~R. Choo, and Y.~Yang, ``Privacy-preserving outsourced
  clinical decision support system in the cloud,'' \emph{IEEE Transactions on
  Services Computing}, 2017.

\bibitem{damgaard2012multiparty}
I.~Damg{\aa}rd, V.~Pastro, N.~Smart, and S.~Zakarias, ``Multiparty computation
  from somewhat homomorphic encryption,'' in \emph{Annual Cryptology
  Conference}.\hskip 1em plus 0.5em minus 0.4em\relax Springer, 2012, pp.
  643--662.

\bibitem{fan2012somewhat}
J.~Fan and F.~Vercauteren, ``Somewhat practical fully homomorphic encryption.''
  \emph{IACR Cryptology ePrint Archive}, vol. 2012, p. 144, 2012.

\bibitem{paillier1999public}
P.~Paillier, ``Public-key cryptosystems based on composite degree residuosity
  classes,'' in \emph{Advances in cryptology—EUROCRYPT’99}.\hskip 1em plus
  0.5em minus 0.4em\relax Springer, 1999, pp. 223--238.

\bibitem{elgamal1985public}
T.~E. Gamal, ``A public key cryptosystem and a signature scheme based on
  discrete logarithms,'' vol.~31, no.~4, 1985, pp. 469--472.

\bibitem{liu2018efficientTDSC}
X.~Liu, K.-K.~R. Choo, R.~H. Deng, R.~Lu, and J.~Weng, ``Efficient and
  privacy-preserving outsourced calculation of rational numbers,'' \emph{IEEE
  Transactions on Dependable and Secure Computing}, vol.~15, no.~1, pp. 27--39,
  2018.

\bibitem{peter2013efficiently}
A.~Peter, E.~Tews, and S.~Katzenbeisser, ``Efficiently outsourcing multiparty
  computation under multiple keys,'' \emph{IEEE Transactions on Information
  Forensics and Security}, vol.~8, no.~12, pp. 2046--2058, 2013.

\bibitem{liu2016privacy}
X.~Liu, R.~H. Deng, W.~Ding, R.~Lu, and B.~Qin, ``Privacy-preserving outsourced
  calculation on floating point numbers,'' \emph{IEEE Transactions on
  Information Forensics and Security}, vol.~11, no.~11, pp. 2513--2527, 2016.

\bibitem{cramer2000general}
R.~Cramer, I.~Damg{\aa}rd, and U.~Maurer, ``General secure multi-party
  computation from any linear secret-sharing scheme,'' in \emph{International
  Conference on the Theory and Applications of Cryptographic Techniques}.\hskip
  1em plus 0.5em minus 0.4em\relax Springer, 2000, pp. 316--334.

\bibitem{chen2006algebraic}
H.~Chen and R.~Cramer, ``Algebraic geometric secret sharing schemes and secure
  multi-party computations over small fields,'' in \emph{Annual International
  Cryptology Conference}.\hskip 1em plus 0.5em minus 0.4em\relax Springer,
  2006, pp. 521--536.

\bibitem{chida2018fast}
K.~Chida, D.~Genkin, K.~Hamada, D.~Ikarashi, R.~Kikuchi, Y.~Lindell, and
  A.~Nof, ``Fast large-scale honest-majority mpc for malicious adversaries,''
  in \emph{Annual International Cryptology Conference}.\hskip 1em plus 0.5em
  minus 0.4em\relax Springer, 2018, pp. 34--64.

\bibitem{dong2013private}
C.~Dong, L.~Chen, and Z.~Wen, ``When private set intersection meets big data:
  an efficient and scalable protocol,'' in \emph{Proceedings of the 2013 ACM
  SIGSAC conference on Computer \& communications security}.\hskip 1em plus
  0.5em minus 0.4em\relax ACM, 2013, pp. 789--800.

\bibitem{burkhart2010fast}
M.~Burkhart and X.~Dimitropoulos, ``Fast privacy-preserving top-k queries using
  secret sharing,'' in \emph{2010 Proceedings of 19th International Conference
  on Computer Communications and Networks}.\hskip 1em plus 0.5em minus
  0.4em\relax IEEE, 2010, pp. 1--7.

\bibitem{huang2019lightweight}
K.~Huang, X.~Liu, S.~Fu, D.~Guo, and M.~Xu, ``A lightweight privacy-preserving
  cnn feature extraction framework for mobile sensing,'' \emph{IEEE
  Transactions on Dependable and Secure Computing}, 2019.

\bibitem{emekcci2007privacy}
F.~Emek{\c{c}}i, O.~D. Sahin, D.~Agrawal, and A.~El~Abbadi, ``Privacy
  preserving decision tree learning over multiple parties,'' \emph{Data \&
  Knowledge Engineering}, vol.~63, no.~2, pp. 348--361, 2007.

\bibitem{8668517}
Z.~{Ma}, Y.~{Liu}, X.~{Liu}, J.~{Ma}, and K.~{Ren}, ``Lightweight
  privacy-preserving ensemble classification for face recognition,'' \emph{IEEE
  Internet of Things Journal}, pp. 1--1, 2019.

\bibitem{ge2010privacy}
X.~Ge, L.~Yan, J.~Zhu, and W.~Shi, ``Privacy-preserving distributed association
  rule mining based on the secret sharing technique,'' in \emph{The 2nd
  International Conference on Software Engineering and Data Mining}.\hskip 1em
  plus 0.5em minus 0.4em\relax IEEE, 2010, pp. 345--350.

\bibitem{gheid2016efficient}
Z.~Gheid and Y.~Challal, ``Efficient and privacy-preserving k-means clustering
  for big data mining,'' in \emph{2016 IEEE Trustcom/BigDataSE/ISPA}.\hskip 1em
  plus 0.5em minus 0.4em\relax IEEE, 2016, pp. 791--798.

\bibitem{gotzfried2017cache}
J.~G{\"o}tzfried, M.~Eckert, S.~Schinzel, and T.~M{\"u}ller, ``Cache attacks on
  intel sgx,'' in \emph{Proceedings of the 10th European Workshop on Systems
  Security}.\hskip 1em plus 0.5em minus 0.4em\relax ACM, 2017, p.~2.

\bibitem{lee2017inferring}
S.~Lee, M.-W. Shih, P.~Gera, T.~Kim, H.~Kim, and M.~Peinado, ``Inferring
  fine-grained control flow inside $\{$SGX$\}$ enclaves with branch
  shadowing,'' in \emph{26th $\{$USENIX$\}$ Security Symposium ($\{$USENIX$\}$
  Security 17)}, 2017, pp. 557--574.

\bibitem{van2017telling}
J.~Van~Bulck, N.~Weichbrodt, R.~Kapitza, F.~Piessens, and R.~Strackx, ``Telling
  your secrets without page faults: Stealthy page table-based attacks on
  enclaved execution,'' in \emph{26th $\{$USENIX$\}$ Security Symposium
  ($\{$USENIX$\}$ Security 17)}, 2017, pp. 1041--1056.

\bibitem{chen2018sgxpectre}
G.~Chen, S.~Chen, Y.~Xiao, Y.~Zhang, Z.~Lin, and T.~H. Lai, ``Sgxpectre
  attacks: Leaking enclave secrets via speculative execution,'' \emph{arXiv
  preprint arXiv:1802.09085}, 2018.

\bibitem{domnitser2012non}
L.~Domnitser, A.~Jaleel, J.~Loew, N.~Abu-Ghazaleh, and D.~Ponomarev,
  ``Non-monopolizable caches: Low-complexity mitigation of cache side channel
  attacks,'' \emph{ACM Transactions on Architecture and Code Optimization
  (TACO)}, vol.~8, no.~4, p.~35, 2012.

\bibitem{costan2016sanctum}
V.~Costan, I.~Lebedev, and S.~Devadas, ``Sanctum: Minimal hardware extensions
  for strong software isolation,'' in \emph{25th $\{$USENIX$\}$ Security
  Symposium ($\{$USENIX$\}$ Security 16)}, 2016, pp. 857--874.

\bibitem{liu2016catalyst}
F.~Liu, Q.~Ge, Y.~Yarom, F.~Mckeen, C.~Rozas, G.~Heiser, and R.~B. Lee,
  ``Catalyst: Defeating last-level cache side channel attacks in cloud
  computing,'' in \emph{2016 IEEE international symposium on high performance
  computer architecture (HPCA)}.\hskip 1em plus 0.5em minus 0.4em\relax IEEE,
  2016, pp. 406--418.

\bibitem{zhou2016software}
Z.~Zhou, M.~K. Reiter, and Y.~Zhang, ``A software approach to defeating side
  channels in last-level caches,'' in \emph{Proceedings of the 2016 ACM SIGSAC
  Conference on Computer and Communications Security}.\hskip 1em plus 0.5em
  minus 0.4em\relax ACM, 2016, pp. 871--882.

\bibitem{coppens2009practical}
B.~Coppens, I.~Verbauwhede, K.~De~Bosschere, and B.~De~Sutter, ``Practical
  mitigations for timing-based side-channel attacks on modern x86 processors,''
  in \emph{2009 30th IEEE Symposium on Security and Privacy}.\hskip 1em plus
  0.5em minus 0.4em\relax IEEE, 2009, pp. 45--60.

\bibitem{shih2017t}
M.-W. Shih, S.~Lee, T.~Kim, and M.~Peinado, ``T-sgx: Eradicating
  controlled-channel attacks against enclave programs.'' in \emph{NDSS}, 2017.

\bibitem{samanthula2014k}
B.~K. Samanthula, Y.~Elmehdwi, and W.~Jiang, ``K-nearest neighbor
  classification over semantically secure encrypted relational data,''
  \emph{IEEE transactions on Knowledge and data engineering}, vol.~27, no.~5,
  pp. 1261--1273, 2014.

\bibitem{shaon2017sgx}
F.~Shaon, M.~Kantarcioglu, Z.~Lin, and L.~Khan, ``Sgx-bigmatrix: A practical
  encrypted data analytic framework with trusted processors,'' in
  \emph{Proceedings of the 2017 ACM SIGSAC Conference on Computer and
  Communications Security}.\hskip 1em plus 0.5em minus 0.4em\relax ACM, 2017,
  pp. 1211--1228.

\bibitem{kuccuk2016exploring}
K.~A. K{\"u}{\c{c}}{\"u}k, A.~Paverd, A.~Martin, N.~Asokan, A.~Simpson, and
  R.~Ankele, ``Exploring the use of intel sgx for secure many-party
  applications,'' in \emph{Proceedings of the 1st Workshop on System Software
  for Trusted Execution}.\hskip 1em plus 0.5em minus 0.4em\relax ACM, 2016,
  p.~5.

\bibitem{chandra2017securing}
S.~Chandra, V.~Karande, Z.~Lin, L.~Khan, M.~Kantarcioglu, and B.~Thuraisingham,
  ``Securing data analytics on sgx with randomization,'' in \emph{European
  Symposium on Research in Computer Security}.\hskip 1em plus 0.5em minus
  0.4em\relax Springer, 2017, pp. 352--369.

\end{thebibliography}

\begin{IEEEbiography}[{\includegraphics[width=1in,height=1.25in,clip,keepaspectratio]{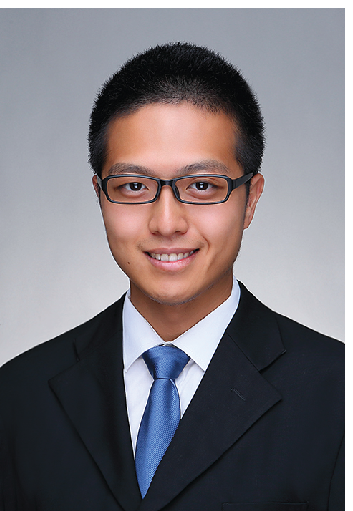}}]{Ximeng Liu}(S'13-M'16)
received the B.Sc. degree in electronic engineering from Xidian University, Xi’an, China, in 2010 and the Ph.D. degree in Cryptography from Xidian University, China, in 2015. Now he is the full professor in the College of Mathematics and Computer Science, Fuzhou University. Also, he is a research fellow at the School of Information System, Singapore Management University, Singapore. He has published more than 100 papers on the topics of cloud security and big
data security—including papers in IEEE Transactions on Computers, IEEE Transactions on Industrial Informatics, IEEE Transactions on Dependable and Secure Computing, IEEE Transactions on Service Computing, IEEE Internet of Things Journal, and so on.  He awards “Minjiang Scholars” Distinguished Professor, “Qishan Scholars” in Fuzhou University, and ACM SIGSAC China Rising Star Award (2018). His research interests include cloud security, applied cryptography and big data security. He is a member of the IEEE, ACM, CCF.
\end{IEEEbiography}

  \begin{IEEEbiography}[{\includegraphics[width=1in,height=1.25in,clip,keepaspectratio]{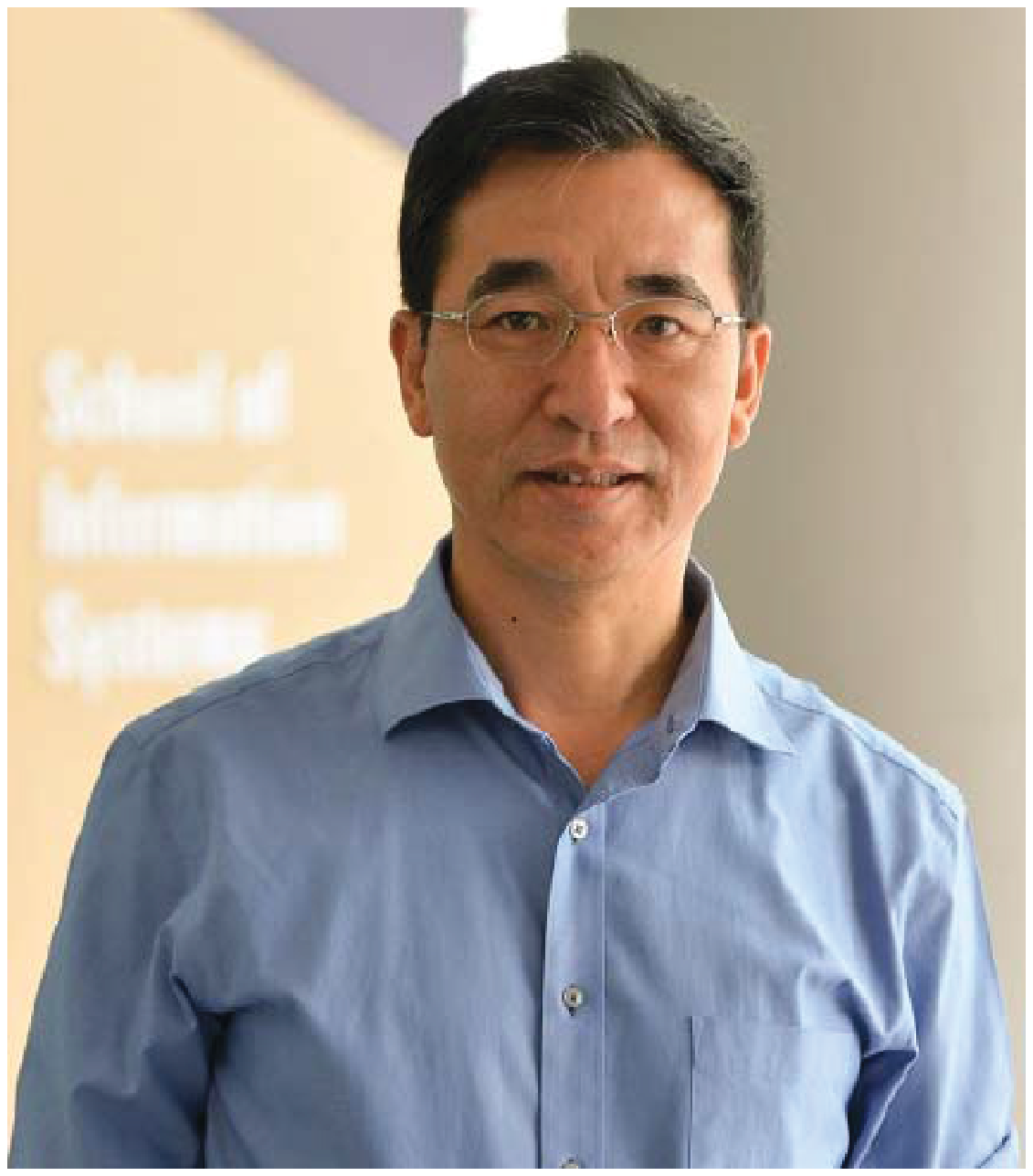}}]{Robert H. Deng} (F'16)
 is AXA Chair Professor of Cybersecurity and Professor of Information Systems in
the School of Information Systems, Singapore
Management University since 2004. Prior to this,
he was a principal scientist and a manager of
Infocomm Security Department, Institute for
Infocomm Research, Singapore.
His research
interests include data security and privacy,
multimedia security, network and system security.
He served/is serving on the editorial boards of many international journals, including the IEEE
Transactions on Information Forensics and
Security, and
IEEE Transactions on Dependable and Secure Computing.  
\end{IEEEbiography}

  \begin{IEEEbiography}[{\includegraphics[width=1in,height=1.25in,clip,keepaspectratio]{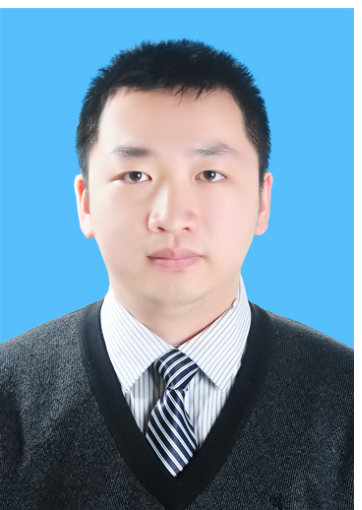}}]{Pengfei Wu} 
 received the B.Sc. degree in software engineering from Shandong University, Jinan, China, in 2016. He is currently pursuing the Ph.D. degree of Software Engineering in Peking University, Beijing, China. His research interests include cloud security and big data security.
\end{IEEEbiography}

  \begin{IEEEbiography}[{\includegraphics[width=1in,height=1.25in,clip,keepaspectratio]{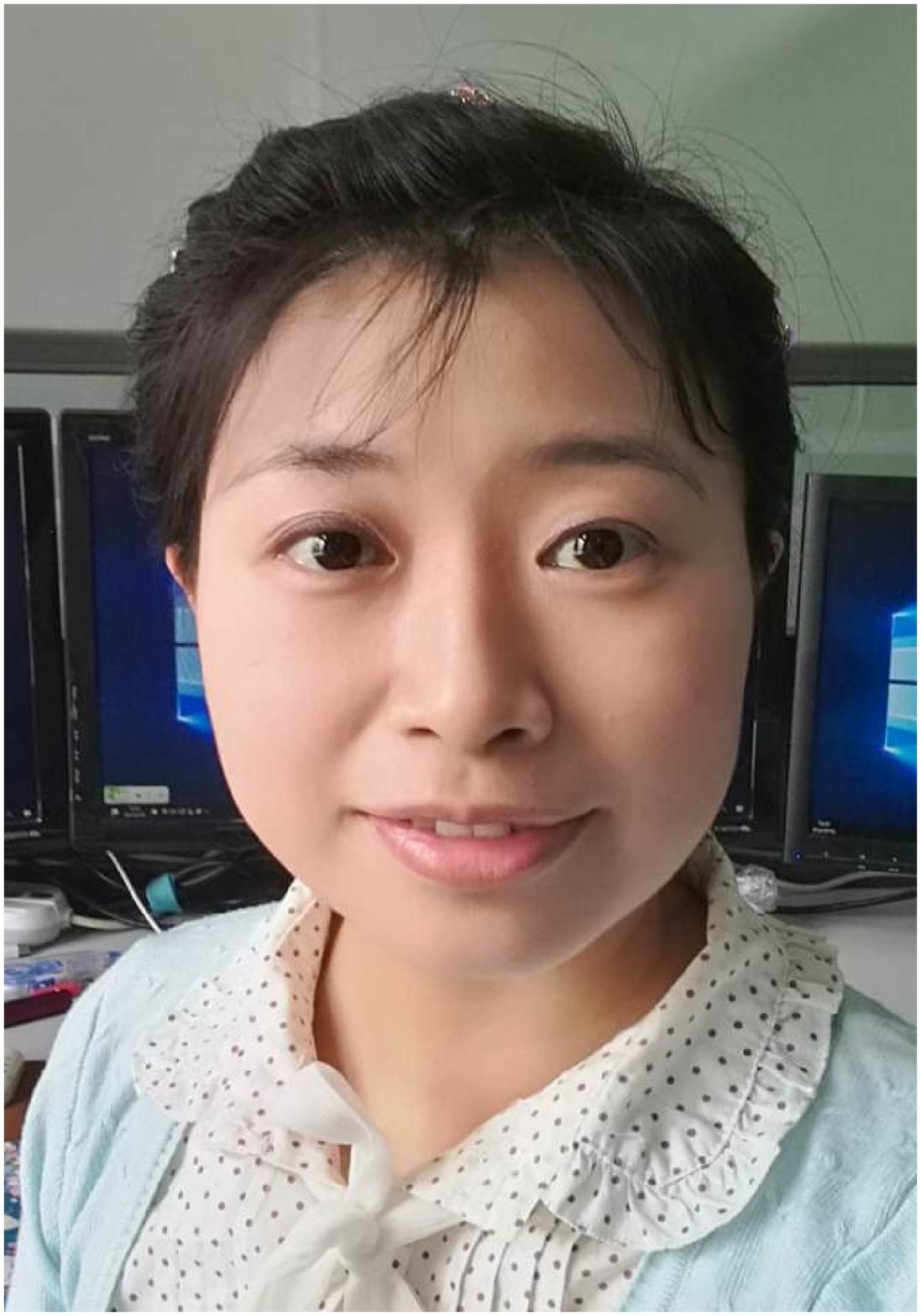}}]
{Yang Yang} (M'17) received the B.Sc. degree from Xidian University, Xi'an,
China, in 2006 and Ph.D. degrees from Xidian University, China, in
2012. She is a research fellow (postdoctor) under supervisor Robert
H. Deng in School of Information System, Singapore Management
University. She is also an associate professor in the college of
mathematics and computer science, Fuzhou University. Her research
interests are in the area of information security and privacy
protection.
    \end{IEEEbiography}

   \section*{Supplementary materials}
    \subsection*{A. Deterministic Turing Machines}
    
    Turing machines are a model of computation which anything   can
be computed that can be computed by a Turing Machine.  As all the function in our PVOA framework can be computed by Deterministic Turing Machine. 
The   rigorous definition is 
defined as follows
    
        \begin{definition}[Deterministic Turing Machine]
        
        A Deterministic Turing Machine contains  a tuple $(Q, \Sigma, \delta, s, h)$ where
1) $Q$ is a finite set of states which contains the states $s, q_{acc}, q_{rej}$.
2) $\Sigma$ is a finite alphabet which contains the symbol \#.
 3) Transition function $\delta: Q -\{q_{acc}, q_{rej}\} \times \Sigma \rightarrow Q  \times \Sigma \cup \{R, L\}$, where $L$ is left shift, $R$ is right shift.
4) $s\in Q $ is the start state, $q_{acc}$ is the accept state, $q_{rej}$ is the reject state.
  \end{definition}

  Suppose  $M$ is a deterministic Turing machine that halts on all inputs. Time complexity function $T_M: \mathbb{N} \rightarrow \mathbb{N}$ is defined as
  \begin{align}
   T_M (n) = \max \{m | \exists w \in \Sigma^*
,|w| = n  \notag\\  \text{such that the computation
of} \ M \ \text{on} \  w \ \text{takes} \ m  \ \text{moves}\}, \notag
      \end{align}
    where numbers are coded in binary format. 
   We call  a Turing machine is \textit{polynomial} if there exists a polynomial $p(n)$, such that  $T_M (n) \leq p(n)$, for all $n \in \mathbb{N}$.

   \subsection*{B. Hard Problem} 

\begin{assumption}$($DDH assumption over $ {\mathbb{Z}}_{N^2}^*$ \cite{DBLP:conf/asiacrypt/BressonCP03}$)$.
For every probabilistic
polynomial time algorithm $\cal{A}$, there exists a negligible function $negl(\cdot)$ such that for
sufficiently large $l$.
\end{assumption}
$$
 \Pr\left[
\begin{array}{l}{\cal{A}}(N,X,Y, \\Z_b \mod N)\\ =b\\\end{array}:\begin{array}{l}
pp\leftarrow Sys(l/2)\\
N=pq, g \leftarrow  \mathbb{G} \\
x,y,z \leftarrow [1, ord( \mathbb{G} )]\\
X=g^x \mod N^2\\
Y=g^y \mod N^2\\
Z_0=g^z \mod N^2\\
Z_1=g^{xy} \mod N^2\\
b \leftarrow \{0,1\}\\
\end{array}
\right] -\frac{1}{2}= negl(l).
$$
\
\begin{theorem}
Let $N$ be a composite modulus product of two large primes.
Let $\mathbb{G}$
be the cyclic group of quadratic residues modulo $N^2$. The  decisional Diffie-Hellman problem over ${\mathbb{Z}}_{N^2}^*$  (in $\mathbb{G}$)
cannot be harder than factoring.
\end{theorem}
\begin{proof}
The detailed proof can be found in \cite{DBLP:conf/asiacrypt/BressonCP03}.
\end{proof}

\subsection*{C. Paillier Cryptosystem Distributed Decryption   (PCDD) }
 
\label{sec:proposedBCRPP}

In order to realize  LightCom, our previous Paillier Cryptosystem Distributed Decryption  (PCDD) \cite{liu2018efficientTDSC} cryptosystem is  used and works as follows:

\textbf{KeyGen:}   Given a security parameter $k$ and two large prime numbers $p,q$, where ${\cal{L}}(p)={\cal{L}}(q)=k$, we have two strong primes $p',q'$, s.t.,  $p' = \frac{p-1}{2}$ and  $q' = \frac{q-1}{2}$ (due to the property of the strong primes). We then compute $N=pq$  and $\lambda=lcm(p-1, q-1)$, define a function $L(x)=\frac{x-1}{N}$, and choose a generator $g$ of order $(p-1)(q-1)/2$.  The public key  is  $pk=(N,g)$, and the corresponding   private key is $sk= \lambda$.

\textbf{Encryption (Enc):} Input a message $m \in \mathbb{Z}_N$, the \texttt{Enc}  chooses a random number $r \in \mathbb{Z}^*_{N^2}$, and output ciphertext  as $[\![m]\!] = g^m r^N \mod N^2$.

\textbf{Decryption (Dec):} Input a ciphertext $[\![m]\!] \in \mathbb{Z}_{N^2}$ and the private key $sk$, the \texttt{Dec}  compute 
$[\![m]\!]^{\lambda} = (1+mN\lambda) \mod N^2.$ Since $gcd(m, \lambda) =1$, the plaintext $m$ can be recovered as 
$ m = L([\![m]\!]^{\lambda}) \cdot \lambda^{-1} \mod N^2.$

\textbf{Private Key Splitting (KeyS):} Input  the private key $\lambda$, the \texttt{KeyS} separates $\lambda$ into $n$ shares as $\lambda^*_i$  such that $\lambda^*_1+\cdots+\lambda^*_n \equiv 0 \mod \lambda$ and $  \lambda^*_1+\cdots+\lambda^*_n \equiv 1  \mod N$.

\textbf{Partially decryption (PDec):}  Once $[\![m]\!]$ is received, with partially private key $\lambda_i^*$, the partially decrypted ciphertext $CT_i$ can be calculated as: 
  $CT_i  = [\![m]\!]^{ \lambda_i^*} \mod N^2.$

\textbf{Threshold decryption  (TDec):}  
Once  $n$ decrypted ciphertexts $CT_1, \cdots ,CT_n $  are received, the   \textbf{TDec} algorithm  can   calculates
  $T = \prod_{i=1}^{n} (CT_i) \mod N^2,$ and 
$m = L(T \mod N^2)$.


 Given $[\![x_1]\!],\cdots, [\![x_n]\!]$ and $a_1,\cdots,a_n$, we   show that our PCDD have the polynomial  homomorphism property (\texttt{Poly}): 
 \begin{center}
 $[\![a_1 \cdot x_1 + a_2 \cdot x_2 + \cdots a_n x_n ]\!]  \leftarrow [\![x_1]\!]^{a_1} \cdot [\![x_2]\!]^{a_2} \cdots [\![x_n]\!]^{a_n}$
  \end{center}

\underline{Homomorphic Properties of DT-PKC:} Here,  we give three homomorphic properties of DT-PKC as follows:

1)  \textit{Additive homomorphism}:  Given ciphertexts $[\![m_1]\!]$ and $[\![m_2]\!]$ under a same public key $pk$, the additive homomorphism can be achieved by ciphertext multiplication, i.e., compute
$[\![m_1]\!]_{pk} \cdot [\![m_2]\!]_{pk} = \{(1+ (m_1+m_2) \cdot N)\cdot h^{ r_1+r_2} \mod N^2, g^{r_1+r_2}  \mod N^2\}
= [\![m_1+m_2]\!]_{pk}$.

2) \textit{Scalar-multiplicative Homomorphism}: Given ciphertexts $[\![m]\!]_{pk}$  and a constant number $c \in \mathbb{Z}_N$, it has 
$([\![m]\!]_{pk})^{c}  =\{(1+m \cdot N)^c \cdot h^{c r_1} \mod N^2, g^{{c}r_1 }  \mod N^2\}
= [\![cm]\!]_{pk}.$ Specifically, let $c = N+1$ and we have
$([\![m]\!]_{pk})^{N-1}  =\{(1+(mN^2- m N)\cdot h^{(N-1)r_1} \mod N^2, g^{{(N-1)}r_1 }  \mod N^2\}
= [\![-m]\!]_{pk}.$

Without any ambiguity, all the ciphertexts below are encrypted under the same public key $pk$, and we use the notion $[\![x]\!]$ instead of $[\![x]\!]_{pk}$.

\end{document}